\documentclass[12pt,oneside]{amsart}
\usepackage[latin9]{inputenc}
\usepackage[a4paper]{geometry}
\usepackage{enumitem}
\usepackage{amstext}
\usepackage{amsthm}
\usepackage{amssymb}
\usepackage{stackrel}
\usepackage{graphicx}
\usepackage{esint}
\usepackage{xargs}[2008/03/08]

\makeatletter

\numberwithin{equation}{section}
\numberwithin{figure}{section}
\theoremstyle{plain}
\newtheorem{thm}{\protect\theoremname}[section]
\theoremstyle{remark}
\newtheorem{rem}[thm]{\protect\remarkname}
\theoremstyle{plain}
\newtheorem{lem}[thm]{\protect\lemmaname}
\theoremstyle{definition}
\newtheorem{example}[thm]{\protect\examplename}
\theoremstyle{definition}
\newtheorem{defn}[thm]{\protect\definitionname}
\theoremstyle{plain}
\newtheorem{prop}[thm]{\protect\propositionname}
\theoremstyle{plain}
\newtheorem{cor}[thm]{\protect\corollaryname}

\@ifundefined{date}{}{\date{}}
\usepackage{amsthm}
\usepackage{longtable}

\DeclareSymbolFont{extraup}{U}{zavm}{m}{n}
\DeclareMathSymbol{\varheart}{\mathalpha}{extraup}{86}
\DeclareMathSymbol{\vardiamond}{\mathalpha}{extraup}{87}
\makeatother

\providecommand{\corollaryname}{Corollary}
\providecommand{\definitionname}{Definition}
\providecommand{\examplename}{Example}
\providecommand{\lemmaname}{Lemma}
\providecommand{\propositionname}{Proposition}
\providecommand{\remarkname}{Remark}
\providecommand{\theoremname}{Theorem}

\newcommand{\term}[1]{\emph{#1}}

\newcommand{\set}[1]{\left\{ #1 \right\}}

\newcommand{\isom}{\cong}
\newcommand{\mphism}{\varphi}
\newcommand{\tens}{\otimes}

\newcommand{\id}{\mathrm{id}}
\newcommand{\idOf}[1]{\id_{{#1}}}

\newcommand{\nord}[1]{#1}

\newcommand{\hI}{0}
\newcommand{\hE}{\frac{1}{2}}
\newcommand{\hS}{\frac{1}{16}}
\newcommand{\cIsing}{\frac{1}{2}}

\newcommand{\vir}{\mathfrak{vir}}
\newcommand{\Lgensym}{\mathrm{L}}
\newcommand{\Cgensym}{\mathrm{C}}
\newcommand{\Lopersym}{\mathsf{L}}
\newcommand{\Fgensym}{\Psi}%{\mathrm{F}}
\newcommand{\Fopersym}{\Psi}
\newcommand{\Lgen}[1]{\Lgensym_{{#1}}}
\newcommand{\Cgen}{\Cgensym}
\newcommand{\Lbar}[1]{\bar{\Lgensym}_{{#1}}}
\newcommand{\Loper}[1]{\Lopersym_{{#1}}}
\newcommand{\Loperbar}[1]{\bar{\Lopersym}_{{#1}}}
\newcommand{\Cbar}{\bar{\Cgensym}}
\newcommand{\Fgen}[1]{\Fgensym_{{#1}}}
\newcommand{\Fbar}[1]{\bar{\Fgensym}_{{#1}}}
\newcommand{\Foper}[1]{\Fopersym_{{#1}}}

\newcommand{\UEA}[1]{\mathcal{U}({#1})}
\newcommand{\UvirC}[1]{\mathfrak{Vir}_{{#1}}}
\newcommand{\UvirHalf}{\UvirC{{\cIsing}}}%{\UvirC{{1/2}}}
\newcommand{\UvirHalfbar}{\overline{\mathfrak{Vir}}_{{\cIsing}}}
\newcommand{\TwoVir}{\mathfrak{Vir}^{\otimes 2}_{{\cIsing}}}%{\UvirHalf \otimes \UvirHalfbar}%{{2\UvirHalf}}
\newcommand{\NeveuSchwarz}{\mathrm{NS}}
\newcommand{\Ramond}{\mathrm{R}}
\newcommand{\CliffNS}{\mathfrak{Cliff}^{\NeveuSchwarz}} % indexed by half-integers, multiplied by integer power monomials
\newcommand{\CliffR}{\mathfrak{Cliff}^{\Ramond}} % indexed by integers, multiplied by half-integer power monomials

\newcommand{\TwoCliffNS}{\mathfrak{Cliff}^{\NeveuSchwarz}}
\newcommand{\TwoCliffR}{\mathfrak{Cliff}^{\Ramond}}
\newcommand{\Igen}{\mathsf{1}}

\newcommand{\Module}{\mathcal{M}}
\newcommand{\VermaSym}{\mathcal{V}}
\newcommand{\Verma}[2]{\VermaSym_{#1,#2}}

\newcommand{\maxsubmodSym}{\mathcal{I}}
\newcommand{\maxsubmod}[2]{\maxsubmodSym_{#1,#2}}
\newcommand{\IrrepSym}{\mathfrak{L}}
\newcommand{\Irrep}[2]{\IrrepSym_{#1,#2}}
\newcommand{\Irrepbar}[2]{\overline{\IrrepSym}_{#1,#2}}
\newcommand{\IrrepI}[1]{\IrrepSym(#1)}
\newcommand{\FockSym}{\mathcal{F}}
\newcommand{\FockNS}{\FockSym^{\NeveuSchwarz}}
\newcommand{\FockNSG}[2]{\FockSym_\ind^{(#1,#2)}}
\newcommand{\FockRG}[2]{\FockSym_\sigma^{(#1,#2)}}
\newcommand{\FockNSbar}{\overline{\FockSym}^{\NeveuSchwarz}}
\newcommand{\FullFockNS}{\FockSym_{\ind}}
\newcommand{\FullFockR}{\FockSym_{\sigma}}

\newcommand{\vacNS}{\ind}

\newcommand{\LFone}{[1]}
\newcommand{\LFspin}{[\sigma]}

\newcommand{\ClInd}{\CliffNS(\LFone)}
\newcommand{\Clspin}{\CliffR(\LFspin)}

\newcommand{\deltaZ}[1]{\delta_{{#1},0}}

\newcommand{\Znn}{\Z_{\ge 0}}

\newcommand{\FockNorm}{2}
\newcommand{\SugawaraNorm}{\frac14}%{\frac12}

\newcommand{\iprod}[2]{\langle {#1} , {#2} \rangle}
\newcommand{\rhomb}[1]{\mathcal{R}(#1)}
\newcommand{\clos}[1]{\operatorname{cl}(#1)}
\newcommand{\setG}[1]{\mathcal{G}\left(#1\right)}
\newcommand{\setQ}[1]{\mathcal{Q}\left(#1\right)}
\newcommand{\ttr}{\mathfrak{t}}
\newcommand{\VertSet}{\mathcal{V}}
\newcommand{\SqLat}{\mathbb{S}}

\newcommand{\Cpsimone}{\frac{\pi}{2}}
\newcommand{\Cpsimoroot}{\left(\frac{2}{\pi}\right)^\frac12}
\newcommand{\dsbar}[1]{\bar{\eta}_{#1}}
\newcommand{\SET}{T}
\newcommand{\SETbar}{\bar{T}}
\newcommand{\SDhol}[1]{\Delta^{\mathrm{h}}_{#1}}

\newcommand{\hihi}[2]{\dsbar{#1}\dsbar{#2}\psi_{#1}\psi_{#2}}
\newcommand{\Gab}[2]{\dsbar{#2}\CK{#2}{#1}}
\newcommand{\Sheet}{\mathcal{S}}

\begin{document}

\global\long\def\C{\mathbb{C}}%
\global\long\def\Cd{\C^{\delta}}%
\global\long\def\Cprim{\C^{\delta,\circ}}%
\global\long\def\Cdual{\C^{\delta,\bullet}}%
\global\long\def\Od{\Omega^{\delta}}%
\global\long\def\Oprim{\Omega^{\delta,\circ}}%
\global\long\def\Odual{\Omega^{\delta,\bullet}}%

\global\long\def\en{\mathcal{\varepsilon}}%
\global\long\def\wone{\widehat{w}}%
\global\long\def\wtwo{w}%
\global\long\def\wo{\wone}%

\global\long\def\wt{\wtwo}%
\global\long\def\cZ{\mathcal{Z}}%

\global\long\def\tfixed{wired}%
\global\long\def\tfree{free}%
\global\long\def\tplus{plus}%
\global\long\def\tminus{minus}%

\global\long\def\fixed{\{\mathtt{\tfixed}\}}%
\global\long\def\free{\{\mathtt{\tfree}\}}%
\global\long\def\plus{\{\mathtt{\tplus}\} }%
\global\long\def\minus{\left\{  \mathtt{\tminus}\right\}  }%

\global\long\def\cvr{\mathrm{\varpi}}%
\global\long\def\Ocvr{\Od_{\cvr}}%
\global\long\def\Ccvr{\Cd_{\cvr}}%
\global\long\def\double#1{\cvr(#1)}%

\global\long\def\Cgr{\mathcal{C}}%
\global\long\def\sCgr{\mathit{c}}%

\global\long\def\Pf{\mathrm{Pf}}%

\global\long\def\vv{v_{1},\dots,v_{n}}%
\global\long\def\uu{u_{1},\dots,u_{m}}%
\global\long\def\svv{\sigma_{v_{1}}\ldots\sigma_{v_{n}}}%
\global\long\def\muu{\mu_{u_{1}}\dots\mu_{u_{m}}}%

\global\long\def\Cutz#1#2{\Cgr_{\{#2\}}(#1)}%

\global\long\def\offdiag#1{\mathcal{\mathring{C}}\vphantom{\mathcal{C}}{}^{#1}_{\{\mathrm{diag}\}}(\Ocvr)}%
\global\long\def\withdiag#1{\overline{\mathcal{C}}\vphantom{\mathcal{C}}{}^{#1}_{\{\mathrm{diag}\}}(\Ocvr)}%

\global\long\def\pesm{\psi^{[\eta]}\!,\en,\mu,\sigma}%
\global\long\def\bcdd{\mathcal{\mathcal{B}}^{\delta}_{\mathrm{mono}}}%
\global\long\def\sfix#1{\sigma^{#1}_{\mathrm{fix}}}%

\global\long\def\any{\diamond}%
\global\long\def\anyother{\triangleright}%

\global\long\def\Opunc{\Omega^{\boxcircle}}%
\global\long\def\Onopunc{\Omega^{\boxempty}}%
\global\long\def\Oother{\Omega'}%
\global\long\def\cvrother{\cvr'}%
\global\long\def\Sother{S'}%
\global\long\def\etaother{\eta'}%

\global\long\def\const{\mathrm{const}}%

\global\long\def\normLoc#1{\beta^{\delta}_{#1}}%

\global\long\def\P{\mathsf{\mathbb{P}}}%

\global\long\def\E{\mathsf{\mathbb{E}}}%

\global\long\def\sF{\mathcal{F}}%

\global\long\def\ind{\mathbb{I}}%

\global\long\def\R{\mathbb{R}}%

\global\long\def\Z{\mathbb{Z}}%

\global\long\def\N{\mathbb{N}}%

\global\long\def\Q{\mathbb{Q}}%

\global\long\def\C{\mathbb{C}}%

\global\long\def\Rsphere{\overline{\C}}%

\global\long\def\re{\Re\mathfrak{e}}%

\global\long\def\im{\Im\mathfrak{m}}%

\global\long\def\arg{\mathrm{arg}}%

\global\long\def\i{\mathfrak{i}}%
\global\long\def\eps{\varepsilon}%
\global\long\def\lamb{\lambda}%
\global\long\def\lambb{\bar{\lambda}}%

\global\long\def\D{\mathbb{D}}%

\global\long\def\H{\mathbb{H}}%

\global\long\def\dist{\mathrm{dist}}%

\global\long\def\reg{\mathrm{reg}}%

\global\long\def\half{\frac{1}{2}}%

\global\long\def\sgn{\mathrm{sgn}}%

\global\long\def\bdry{\partial}%

\global\long\def\cl#1{\overline{#1}}%

\global\long\def\diam{\mathrm{diam}}%
\global\long\def\corr#1{\overline{#1}}%
\global\long\def\Corr#1#2{\E_{#1}(#2)}%

\global\long\def\corr#1#2#3{\langle#1\rangle_{#2,#3}}%
\global\long\def\pa{\partial}%
\global\long\def\tto#1{\stackrel{#1}{\longrightarrow}}%

\global\long\def\res{\text{res}}%

\global\long\def\u{u}%
\global\long\def\v{v}%

\global\long\def\z{z}%
\global\long\def\mod{\;\mathrm{mod\;}}%
\global\long\def\wind{\mathrm{w}}%
\global\long\def\vz{z^{\bullet}}%
\global\long\def\fz{z^{\circ}}%
\global\long\def\CF{\mathfrak{C}}%
\global\long\def\RPF{\text{I}}%

\global\long\def\FFS#1#2#3{F^{#3}_{#2}(#1)}%
\global\long\def\ds#1{\eta_{#1}}%
\global\long\def\dbar{\overline{\partial}}%
\global\long\def\dual#1{\left(#1\right)^{*}}%

\global\long\def\bcond{\mathcal{B}}%
\global\long\def\bcd{\mathcal{B}^{\delta}}%
\global\long\def\Eod{\Od_{+}}%
\global\long\def\lapv{\Delta^{\circ}}%
\global\long\def\lapf{\Delta^{\bullet}}%
\global\long\def\GammaR{\Gamma_{\R}}%
\global\long\def\GammaiR{\Gamma_{i\R}}%

\global\long\def\ccor#1{\langle#1\rangle}%
\global\long\def\bar#1{\overline{#1}}%
\global\long\def\anypsi{\psi^{*}}%
\global\long\def\Op{\mathcal{O}}%
\global\long\def\crossing{c}%
\global\long\def\Fdual{F^{\mathrm{dual}}}%

\global\long\def\zz{z_{1},z_{2}}%
\newcommandx\norm[1][usedefault, addprefix=\global, 1=]{n_{#1}}%
\newcommandx\tang[1][usedefault, addprefix=\global, 1=]{\tau_{#1}}%
\global\long\def\Ocvrc{\Omega_{\cvr}}%
\global\long\def\crad{\text{crad}}%
\global\long\def\feta{f^{[\eta]}}%
\global\long\def\T{\mathbb{T}}%
\global\long\def\reg{\sharp}%
\global\long\def\regg{*}%
\global\long\def\coefA{\mathcal{A}}%
\global\long\def\primaries{\{\ind,\sigma,\mu,\psi,\psistar{},\en\}}%

\global\long\def\pbar#1{#1^{\star}}%
\global\long\def\fdag{f^{\star}}%
\global\long\def\CorrO#1{\langle#1\rangle_{\Omega}}%
\global\long\def\formL{\mathcal{L}}%
\global\long\def\appe{\approx_{\eps}}%
\global\long\def\jayhat{\hat{j}}%
\global\long\def\Csigma{C_{\sigma}}%
\global\long\def\Ceps{C_{\eps}}%
\global\long\def\Cpsi{C_\psi}%
\global\long\def\CCpsi{\left(\frac{\pi}{2}\right)^{\frac12}}%
\global\long\def\Cmu{C_{\mu}}%
\global\long\def\Cmu{C_{\mu}}%
\global\long\def\umax{u_{\text{max}}}%
\global\long\def\psistar#1{\psi^{\star}{}_{\hskip-5pt #1}}%
\global\long\def\sqr#1{(#1)^{\frac{1}{2}}}%

\global\long\def\scvr{\sigma_{\cvr}}%
\global\long\def\Ob{\mathcal{O}}%

\global\long\def\Obs#1{\mathcal{O}\left[#1\right]}%
\global\long\def\wcgr{\Cgr^{\diamondsuit}}%
\global\long\def\bcgr{\Cgr^{\vardiamond}}%
\global\long\def\dzmone#1#2{P_{#2}\left(#1\right)}%
\global\long\def\dmsqrt#1#2{Q_{#2}\left(#1\right)}%
\global\long\def\dsqrt#1#2{R_{#2}\left(#1\right)}%

\global\long\def\proj#1#2{\mathrm{Pr}_{#1}\left(#2\right)}%
\global\long\def\projj#1#2{#1\re\left(\bar{#1}#2\right)}%
\global\long\def\placket{\mathcal{S}}%
\global\long\def\sflat{s_{\flat}}%
\global\long\def\rother#1#2{\tilde{R}_{#1}(#2)}%

\global\long\def\ssharp{c^{\sharp}}%
\global\long\def\sflat{c^{\flat}}%
\global\long\def\sany{c^{\diamond}}%
\global\long\def\Ann{\mathbb{A}}%
\global\long\def\spIn{\sigma_{\mathrm{in}}}%
\global\long\def\spOut{\sigma_{\mathrm{out}}}%
\global\long\def\S{\text{\ensuremath{\mathbb{S}}}}%
\global\long\def\FK{\text{\text{FK}}}%
\global\long\def\SOd{\text{\ensuremath{\hat{\Omega}^{\delta}}}}%
\global\long\def\ArcHat{\text{\ensuremath{\hat{\gamma}^{\delta}}}}%
\global\long\def\CCHat{\text{\ensuremath{\hat{\gamma}_{\text{in}}}}}%
\global\long\def\SOmega{\text{\ensuremath{\hat{\Omega}}}}%
\global\long\def\AOC{\text{\ensuremath{\gamma_{1}}}}%
\global\long\def\ANC{\text{\ensuremath{\gamma_{2}}}}%
\global\long\def\Neigh{\text{\ensuremath{\tilde{\Omega}}}}%
\global\long\def\ArcNeig{\text{\ensuremath{\tilde{\gamma}}}}%
\global\long\def\CCNeig{\text{\ensuremath{\tilde{\gamma}_{\text{in}}}}}%
\global\long\def\nextc#1{#1_{\text{next}}}%
\global\long\def\CCHatd{\text{\ensuremath{\hat{\gamma}^{\delta}_{\text{in}}}}}%
\global\long\def\ProbeF{F^{\delta}_{\text{Ref}}}%
\global\long\def\ProbeFc{f_{\text{Ref}}}%
\global\long\def\ProbeFct{\tilde{f}_{\text{FK}}}%
\global\long\def\NormC{C}%
\global\long\def\Square{Q}%
\global\long\def\Side#1{l_{#1}}%
\global\long\def\normAny{\beta_{\delta}}%
\global\long\def\normPf#1{\beta^{(#1)}_{\delta}}%
\global\long\def\intFsquare#1#2{(\delta)\int^{#2}\im\left(#1dz\right)}%
\global\long\def\intFQ#1#2{(\delta)\int_{#2}\im\left(#1dz\right)}%
\global\long\def\sn#1#2{\mathrm{sn}\left(#1\,|\,#2\right)}%
\global\long\def\sc#1#2{\mathrm{sc}\left(#1\,|\,#2\right)}%
\global\long\def\sd#1#2{\mathrm{sd}\left(#1\,|\,#2\right)}%
\global\long\def\nd#1#2{\mathrm{nd}\left(#1\,|\,#2\right)}%
\global\long\def\cn#1#2{\mathrm{cn}\left(#1\,|\,#2\right)}%
\global\long\def\cd#1#2{\mathrm{cd}\left(#1\,|\,#2\right)}%
\global\long\def\dn#1#2{\mathrm{dn}\left(#1\,|\,#2\right)}%
\global\long\def\nc#1#2{\mathrm{nc}\left(#1\,|\,#2\right)}%
\global\long\def\dexpk#1#2#3#4{\mathrm{e}\left(#1\,|\,#2;\,#3\,,#4\right)}%
\global\long\def\dc#1#2{\mathrm{dc}\left(#1\,|\,#2\right)}%
\global\long\def\ns#1#2{\mathrm{ns}\left(#1\,|\,#2\right)}%

\global\long\def\corvz#1#2{\mathrm{G}_{[#1]}(#2)}%
\global\long\def\coruz#1#2{\mathrm{G}_{[#1]}(#2)}%
\global\long\def\corzz#1#2{\mathrm{G}_{(#1)}(#2)}%
\global\long\def\fctr{\Gamma}%
\global\long\def\Erruz{R_{\chi\sigma}}%
\global\long\def\Errvz{R_{\chi\mu}}%
\global\long\def\Errzz{R_{\chi\chi}}%
\global\long\def\ea#1{#1}%
\global\long\def\mup{\nu}%
\global\long\def\ang#1{\alpha_{#1}}%
\global\long\def\eang#1{\alpha_{#1}}%
\global\long\def\lamb{\lambda}%
\global\long\def\lambb{\bar{\lambda}}%
\global\long\def\ell#1{\check{#1}}%
\global\long\def\ella#1{\check{\alpha}_{#1}}%
\global\long\def\ang#1{\alpha_{#1}}%
\global\long\def\dexp#1#2#3{\mathrm{e}\left(#1;\,#2,\,#3\right)}%
\global\long\def\cst{\mathrm{cst}}%
\global\long\def\MCon#1{\mathrm{F}^{\delta}_{#1}}%
\global\long\def\thetaa{\hat{\theta}}%
\global\long\def\thetag{\bar{\theta}}%
\global\long\def\wired{\mathrm{w}}%
\global\long\def\BAP{\text{BAP}}%
\global\long\def\lambi{\lambda}%
\global\long\def\dd{\mathfrak{d}}%
\global\long\def\Ver{\mathcal{V}}%
\global\long\def\spanned{\mathrm{span}}%
\global\long\def\stspin{|\sigma\rangle}%
\global\long\def\ket#1{|#1\rangle}%
\global\long\def\Msm{M_{\sigma\mu}}%
\global\long\def\Mid{M_{\ind}}%
\global\long\def\Corm#1#2{\mathcal{C}_{#1,#2}}%
\global\long\def\Fun#1{\mathrm{Fun}\left(#1\right)}%
\global\long\def\LF{\mathcal{L}}%
\global\long\def\LFo{\mathcal{L}_{\sigma}}%
\global\long\def\LFe{\mathcal{L}_{\ind}}%
\global\long\def\LFP{\mathcal{L}^{\psi}}%
\global\long\def\LFPo{\mathcal{L^{\psi}_{\sigma}}}%
\global\long\def\LFPe{\mathcal{L^{\psi}_{\ind}}}%
\global\long\def\LFPoo{\mathcal{L}^{\psi,\mathrm{odd}}_{\mathrm{\sigma}}}%
\global\long\def\LFPoe{\mathcal{L}^{\psi,\mathrm{even}}_{\sigma}}%
\global\long\def\LFPeo{\mathcal{L}^{\psi,\mathrm{odd}}_{\ind}}%
\global\long\def\LFPee{\mathcal{L}^{\psi,\mathrm{even}}_{\ind}}%
\global\long\def\Eval#1#2{\mathrm{Ev}_{#2}\left(#1\right)}%
\global\long\def\Cliff#1{\mathcal{F}\left(#1\right)}%
\global\long\def\Curr#1{\mathrm{\Psi}_{#1}}%
\global\long\def\Currbar#1{\mathrm{\bar{\Psi}}_{#1}}%
\global\long\def\RCurr#1#2{\mathrm{\Psi}^{[#2]}_{#1}}%
\global\long\def\Vir{\mathrm{Vir}}%
\global\long\def\HV{\mathcal{H}}%
\global\long\def\isG{\Gamma}%
\global\long\def\MainBij{\mathcal{B}}%
\global\long\def\root{o}%
\global\long\def\rootE{e(\root)}%
\global\long\def\Mft{M(4,3)}%
\global\long\def\WeightProj#1#2{P_{#1}(#2)}%
\global\long\def\SCkx#1#2{\bar{\eta}_{#1}Q_{(#1)}}%
\global\long\def\SCknoeta#1#2{Q_{(#1)}(#2)}%
\global\long\def\Que#1{\bar{\eta}_{#1}\tilde{Q}_{(#1)}}%
\global\long\def\isGoc{\Gamma_{[\root^{\circ}]}}%
\global\long\def\isGd{\Gamma_{\cvr}}%
\global\long\def\isGob{\Gamma_{[\root^{\bullet}]}}%
\global\long\def\CK#1#2{\bar{\eta}_{#2}G_{(#2)}(#1)}%
\global\long\def\CKc#1#2{G_{(#2)}(#1)}%
\global\long\def\Gu#1#2{G_{[#2]}(#1)}%
\global\long\def\Gv#1#2{G_{[#2]}(#1)}%
\global\long\def\SCKc#1#2{Q^{[#2]}_{(#1)}}%
\global\long\def\Quec#1{Q_{(#1)}}%
\global\long\def\BallO#1{B_{(#1)}}%
\global\long\def\RadO#1{R_{#1}}%
\global\long\def\Null{\mathcal{N}}%
\global\long\def\Ins{\mathcal{\mathcal{I}^{\psi}}}%
\global\long\def\EvP{\mathfrak{E}}%
\global\long\def\EvPe{\mathfrak{E}_{\mathrm{even}}}%
\global\long\def\EvPo{\mathfrak{E}_{\mathrm{odd}}}%
\global\long\def\Config{\mathrm{Conf}_{k}(\Omega)}%
\global\long\def\ConfUC{\mathrm{\widetilde{Conf}}_{k}(\Omega)}%
\global\long\def\MainBijSpins{\mathcal{\hat{B}}}%

\global\long\def\bspin{\mathbb{\sigma}}%
\global\long\def\HVM#1{\mathcal{M}(#1)}%
\global\long\def\bmu{\mathbb{\mu}}%
\global\long\def\bpsi{\mathbb{\psi}}%
\global\long\def\bpsibar{\mathbb{\psi^{\star}}}%
\global\long\def\ben{\mathbb{\eps}}%
\global\long\def\CANS{\mathbb{\mathfrak{Cliff}}_{NS}}%
\global\long\def\CAR{\mathbb{\mathfrak{Cliff}}_{R}}%
\global\long\def\vectn{\overline{n}}%

\author{Dmitry Chelkak, Konstantin Izyurov, Kalle Kyt\"ol\"a}
\title{Critical Ising model as the $M(4,3)$ minimal model of conformal field
	theory}
\begin{abstract}
	We construct a linear map from the space of spin polynomials in the
	critical Ising model on a lattice to the state space of the $M(4,3)$
	minimal model of conformal field theory. The kernel of the map consists
	precisely of null fields, i.e., spin polynomials that vanish in all
	correlations. For a polynomial $\Op$ which is not a null field, we
	define its scaling dimension $\Delta_{\Op}\in\Z_{\geq0}\cup(\Z_{\geq0}+\frac{1}{8}),$
	to be the lowest $\Loper{0}+\Loperbar{0}$ eigenvalue in
	the decomposition of its image into homogeneous components. We show
	that as the mesh size $\delta$ goes to zero, the rescaled field $\delta^{-\Delta_{\Op}}\Op$
	has a non-trivial scaling limit, when inserted in correlations with
	other such fields. Moreover, all these limits can be expressed in
	terms of correlations of primary fields of $M(4,3)$ and the Virasoro
	action in the continuum. The highest weight vectors of $M(4,3)$,
	the joint eigenvectors of $\Loper{0}$ and $\Loperbar{0}$ with the same eigenvalues $0,1/16,1/2$ for both, are the images of the
	constant $1$, the spin at a single lattice vertex, and the energy
	at a single lattice edge, respectively. 
	
	For the proof, we construct a map from a larger lattice field space,
	accommodating also fermion and disorder insertions, to a direct sum
	of two fermionic Fock spaces. This direct sum contains contains the
	state spaces of both the $M(4,3)$ and the free fermion CFT. The main technical ingredient of the proof is a new version of the exact fermionic current action on the space of lattice local fields due to C. Hongler, K.K. and F. Viklund, combined with some new advances in discrete complex analysis. Our results are valid on arbitrary isoradial graphs.
\end{abstract}

\maketitle
\section{Introduction }

It is a truth universally acknowledged, in the physics literature,
that the scaling limit of the critical planar Ising model is the $M(4,3)$
minimal model of conformal field theory. This was pointed out in the
very paper of Belavin, Polyakov and Zamolodchikov \cite{BPZ} where
conformal field theory (CFT) was first introduced, and for forty years
since, it has been a major motivating example for application of CFT
to statistical mechanics, and for the power of the conformal invariance
ansatz. Minimal models of CFT are non-perturbative, exactly solvable
quantum field theories, and therefore relating (a scaling limit of)
a lattice model to one of those allows one, in principle, to compute
all correlations in this scaling limit \cite{BPZ, Yellow_book, mussardo2010statistical}. For a physicists, the vast amount of \emph{exact results} available for the Ising model starting from the seminal work of Onsager \cite{onsager1944crystal, kaufman1949crystal, yang1952spontaneous,mccoy2014two, palmer2007planar} and the matching of, e.g., critical exponents and correlation functions, provided decisive evidence for identifying the scaling of the critical Ising model with $M(4,3)$. In its turn, the CFT perspective allowed to predict new exact results for the Ising model \cite{di1987critical}. For the rest of the paper, $M(4,3)$ refers to the full (i.e., two-chiral) diagonal minimal CFT.

Mathematically, a rigorous proof and even a precise formulation of
the above assertion was missing until now, despite the emergence of
two very successful rigorous approaches to conformal invariance in
statistical mechanics. The first of them, Schramm--Loewner evolution
\cite{Schramm_SLE,Lawler_book} (SLE) and conformal loop ensembles \cite{sheffield2012conformal, benoist2019scaling, kemppainen2016conformal}
(CLE), strives to describe the scaling limits of models of statistical
mechanics in terms of random fractal curves. There has been a rich
interplay between CFT on one side, and SLE and CLE on the other side \cite{bauer2002slekappa, friedrich2003, bauer2004conformaltransformations, doyon2013conformal, DubedatSLEvirasorolocalization, DubedatSLEvirasorofusion, baverez2024cft, peltola2019toward, kytola_Peltola_general}. In the case of the Ising model,
however, the latter deal with \emph{non-local} observables that are
somewhat different from \emph{local} fields of CFT, although in the
case of the spin field \cite{camia2015planar}, the latter have been expressed in terms of
the former in a number of ways \cite{camia2025conformal, lopez2026ising, miller2017cle}.

The second approach, that of discrete complex analysis, pioneered by Kenyon and Smirnov, remains the
main tool to actually prove convergence of lattice models to conformally
covariant scaling limits \cite{Kenyon_conf_inv,Smirnov_Ising,Smirnov_percolation,Smirnov_Towards,ChelkakSmirnov2}.
The idea is to recast integrability properties of the models as discrete
holomorphicity (or harmonicity) of correlations of a well-chosen observable
depending on a lattice site as a parameter. Discrete holomorphicity
then allows one to compute scaling limits, which are automatically
conformally invariant or covariant. The most interesting and natural
observables, such as the spins or energy in the Ising model \cite{HonglerSmirnov,ChelkakHonglerIzyurov,chelkak2021correlations},
or height function, electric or magnetic correlations in the dimer
model \cite{DubedatDimers,Kenyon_conf_inv}, are usually not discrete
holomorphic, but the scaling limits of their correlations can be computed
by expressing them in terms of discrete holomorphic ones: respectively,
the fermion and spin-fermion correlations in the Ising case, and the inverse Kasteleyn
matrix in the dimer case. 

This approach gives some interesting correlations directly, as has
been done in the case of the Ising model for all primary fields \cite{chelkak2021correlations}, and
one might then check that they coincide with the CFT correlations
as predicted in the physics literature \cite{bayraktaroglu2023bosonization}.
Besides primary fields, however, there are many more fields in the Ising model, and the same is true for CFT, as highlited by \cite{BPZ}. The discrete holomorphicity methods have been
extended in some cases beyond primary fields \cite{ChelkakGlazmanSmirnov,izyurov2024energy}, or to general fields that have the same scaling dimension as one of the primaries \cite{gheissari2019ising}, 
but to date, each of these extensions required an \emph{ad hoc} adaptation
of discrete holomorphicity arguments. 

In the present paper, we show, for the first time, that the scaling limits
and the scaling dimensions of all possible local fields in the critical
Ising model -- i.e., polynomials in spins -- are classified by (homogeneous)
states of the $M(4,3)$ model, and their corresponding scaling dimensions.
More precisely, we construct a surjective linear map $\MainBijSpins$ that
associates to each spin polynomial $\Op$ an $M(4,3)$ state, such
that the kernel of $\MainBijSpins$ consist precisely of null fields, i.e.,
fields that vanish in correlations with all fields supported sufficiently
far away. We moreover show that each non-null field $\Op$ has \emph{a
	scaling exponent} $\Delta_{\Op}\in\Z_{\geq0}\cup(\Z_{\geq0}+\frac{1}{8})$
such that the rescaled field $\delta^{-\Delta_{\Op}}\Op$ has a non-trivial
scaling limit as the lattice mesh size $\delta$ goes to zero, when
inserted in correlations with similar fields. The scaling exponent
is given by the lowest $\Loper{0}+\Loperbar{0}$ eigenvalue
of a non-vanishing term in the decomposition of $\MainBijSpins(\Op)$ into
homogeneous components. As is well known, the $M(4,3)$ model has
three primary fields, namely, the identity $\ind$, the spin $\sigma$,
and the energy $\eps$, and we show that in our correspondence, they
are images of the constant, a spin a a single vertex, and an energy
at a single edge.

A curious feature of the critical Ising model is that depending on
the set of observables one is looking at, the model has \emph{two}
different $c=\frac{1}{2}$ CFTs associated to it, the other one being
the free fermion theory. (The free fermionic structure underlying the Ising model goes back to early work of Onsager and Kaufman \cite{kaufman1949crystal}.) Fermions are not \emph{local fields}, in
that a single fermion insertion cannot be represented as a polynomial
in spins (with a reasonably bounded number of spins involved), although
a \emph{fermion pair }does admit such a representation. Conversely,
a single spin or more generally an \emph{odd} spin polynomial does
not admit a local expression in terms of fermions, although an \emph{even}
spin polynomial does. In spite of this somewhat delicate relation
to spin observables, fermions are a central tool both in the classical
approaches to the Ising model since Onsager, and a more recent celebrated
insight by Smirnov that their discrete holomorphicity can be leveraged
in proofs of conformal invariance. On the CFT side, it has been well
known to the physicists that the \emph{stress-energy tensor} in the
Ising model CFT admits a quadratic expression in terms of fermions~\cite{Yellow_book}. This representation
is an analytic counterpart of the \emph{Sugawara construction
}of the Virasoro action on a fermionic Fock space. 

In our proof, the discrete holomorphicity of fermionic observables
also plays a central role. In fact, we construct a surjective linear
map $\MainBij$ from a larger ``lattice field'' space to a direct
sum of two fermionic Fock spaces, and we prove the corresponding convergence
result for this larger space. The fields in this larger space are
Grassmann polynomial{\small s} in fermions with, possibly, one spin
insertion; equivalently, they can be realized as spin polynomials
with, possibly, one fermion insertion. The Fock spaces in question
carry Clifford algebra actions that give rise to two commuting $c=\frac{1}{2}$
Virasoro actions via Sugawara construction; as a $\Vir_{\frac{1}{2}}\oplus\Vir_{\frac{1}{2}}$--
module, the direct sum has a submodule isomorphic to $M(4,3)$ as
well as a submodule isomorphic to the free fermionic CFT state space.
The former is the image of the space of all spin polynomials, while the
latter is the image of all Grassmann polynomials in fermions. The
\emph{even }spin or Grassmann polynomials can be realized as either
of those, and in particular both spaces contain the submodules generated
by the identity $\ind$ and the energy $\en$. 

Let us give an overview of our proof. In a sense, we construct not
the map $\MainBij$ but its inverse. The key step was envisioned in
\cite{hongler2022conformal}, where the Clifford action on the
(enlarged as above) space of Ising fields on the square lattice, modulo null fields, was
constructed via discrete contour integrals. (This part has, by now, been implemented also for the Gaussian free field \cite{hongler2022conformal,adame2024fock} and double dimers \cite{adame2025discrete,adame2025lattice}.) Purely algebraically,
this Clifford action gives rise to a $\Vir_{\frac{1}{2}}\oplus\Vir_{\frac{1}{2}}$
module structure on the space of local fields. The main missing part,
which is supplied in the present article, was the understanding of
the structure of this module, namely, the fact that it is a direct
sum of highest weight modules generated by the primary fields of weight
$0,\frac{1}{8},1$, that is, the identity, the spin, and the energy.
In particular, there was no proof that any spin polynomial decomposes
into a linear combination of descendants of these fields. In fact,
we were not able to check the required properties of the construction
in \cite{hongler2022conformal}. We provide a modified version
thereof, which is more amenable to the computations related to the
Ising model, and also works on arbitrary isoradial graphs. In the
continuum, the fermionic creation and annihilation operators are realized
as coefficients of the OPE product with the fermions. As fermions
are (anti-)holomorphic, these coefficients can be picked out by contour
integration, after multiplying by a suitable monomial; the anti-commutation
relations then follow from the fermion-fermion OPE. In the discrete,
this operation is delicate: first, one needs to construct suitable
discrete holomorphic monomials, and, more importantly, the product
of two discrete holomorphic functions is not discrete holomorphic,
so, a contour integration procedure must be chosen so that the result
does not depend on the contour. In \cite{hongler2022conformal},
the monomials were (complex-valued) discrete holomorphic functions
on the square lattice, and the fermions were indexed by the dual lattice, so
the corresponding contour integral was seen to be independent of the contour by
the discrete Stokes' theorem. In our construction, both the monomials
and fermions are \emph{s-holomoprhic }and are\emph{ }defined on (the
same) corner lattice, and the contour integration is based on Smirnov's
definition of the imaginary part $H=\im\int F^{2}\,dz$ of the square
of s-holomorphic function \cite{Smirnov_Ising,ChelkakSmirnov2}. Although
this definition was motivated by the fact that it turns a
discrete Riemann boundary value problem for $F$ to a much more
robust Dirichlet problem for $H$, curiously, it seems to be the right
tool even when there's ostensibly no boundary present. The s-holomorphic
functions form a real linear space, so we first construct a real version
of the fermionic Fock space, and then complexify. A suitable ``discrete
residue'' property of fermionic correlations plays the role of the
fermion-fermion OPE, leading to the right anti-commutation relations.
It's been well known that these commutation relations imply Virasoro
action by the Sugawara construction; since the Sugawara generators are
quadratic in fermions, this action preserves the space of actual lattice
spin polynomials, modulo null fields. 

In this setup, it turns out quite straightforward to check that the
identity and the spin at a single vertex are vacuum vectors
in the Fock space, i.e., they are annihilated by the positive modes,
and to check that an application of two suitable Fock space creation
operators to the identity leads to energy at a single edge, which
then must be a Virasoro primary by algebraic reasons. Moreover, since
Fock representations of the Clifford algebra are unitary, so is the
Virasoro representation; in particular, it decomposes into a direct
sum of irreducibles, and the (single chirality) highest weights are known to belong to
$\{0,\frac{1}{16},\frac12\}$ by the classical work of Friedan, Qiu, and
Shenker \cite{friedan1986details}. The hardest part of the proof is to show that any spin
polynomial is a descendant of one of the three primaries listed above.
We deduce this statement from some algebraic considerations together
with the fact that any Grassmann polynomial (respectively, Grassmann
polynomial with a spin insertion) is a Clifford descendant of the
identity (respectively, of a spin at a single vertex). The latter
fact is proven by induction in the degree of the Grassmann polynomial,
which uses discrete Laurent expansions for the discrete Cauchy kernel
and its spinor version. These expansions appear to be new. This approach
is substantially different from that of \cite{adame2024fock}, which proved the
corresponding results for the bosonic Fock space and discrete Gaussian
free field by an explicit characterization of null fields, together
with dimension counting.

This, in particular, shows that any Grassmann polynomial (respectively,
a Grassmann polynomial with one spin insertion), can be represented
as a linear combination of multiple contour integrals involving fermion
and spin correlations. Considering a correlation of several such spin
polynomials in the scaling limit regime, one can deform the contours
so that they stay at a macroscopic distance from each other. The correct
scaling dimension is determined simply by the fields involved, and
the convergence follows in any context when we know the convergence
of spin-fermion correlations. The latter has been established, in
particular, in finitely connected planar domains with piecewise plus/minus/free
boundary conditions \cite{chelkak2021correlations} and on a torus
\cite{izyurov2024energy,bayrakarogluIzyurov2024}.

Our results are valid on all isoradial graphs, which is important in itself as a universality result, but also useful for future work on extending the results to Riemann surfaces. More
precisely, each such (infinite, rooted) isoradial graph tiling
the plane has its own map $\MainBij$ from the space of lattice fields on that graph to the direct sum of fermionic Fock spaces. Moreover,
even for a fixed graph, the correspondence $\MainBij$ is not unique,
but depends on certain choices in the construction. Of course, these choices do not affect the leading
homogeneous component of $\MainBij(\Op),$ as the scaling dimension
and the scaling limit are independent of them. To illustrate
this, on the square lattice
$\Z^{2}$, the most natural (symmetric) choice of the correspondence
is such that the spin $\sigma_{(0,0)}$ at the origin is mapped to
$\bspin,$ the primary field of scaling dimension $\frac{1}{8},$ but one can
also set up the correspondence so that, say, its neighboring spin
$\sigma_{(0,1)}$ is mapped to $\bspin$. The difference, $\sigma_{(0,1)}-\sigma_{(0,0)},$
is then mapped to the descendent field $(\Loper{-1}+\Loperbar{-1})\bspin,$
which has scaling dimension $1\frac{1}{8}$ and thus disappears in
the limit when rescaled by $\delta^{-\frac{1}{8}}.$ 

\section{Setup and the main results}

\subsection{The fermionic Fock spaces and the Virasoro modules\textmd{. \label{subsec: intro_Fock}}}

In order to state our results, we need to briefly describe the vector
space that encodes the possible scaling limits of lattice fields,
see Section \ref{sec: background_algebra} for more details. As a vector space,
it is a direct sum of two Fock spaces $\FullFockNS$ and $\FullFockR$,
in the physics jargon called the \emph{Neveu--Schwarz }and the \emph{Ramond}
sector respectively. The former will correspond to Grassmann polynomials
in lattice fermions (equivalently, \emph{even} spin polynomials with,
possibly, one fermion insertion), while the latter will correspond
to Grassmann polynomials with one spin insertion (equivalently, \emph{odd}
spin polynomials with, possibly, one fermion insertion).

Each of $\FullFockNS$ and $\FullFockR$ carries a highest weight
representation of its own Clifford algebra, $\TwoCliffNS$ and $\TwoCliffR$, with generators $\{\Psi_{n},\bar{\Psi}_{n}\}_{n\in\Z+\frac{1}{2}}$
and $\{\Psi_{n},\bar{\Psi}_{n}\}_{n\in\Z}$ respectively, and the
anti-commutation relations 
\begin{equation}
	\Psi_{n}\Psi_{m}+\Psi_m\Psi_n=2\delta_{n+m,0}=\bar{\Psi}_{n}\bar{\Psi}_{m}+\bar{\Psi}_{m}\bar{\Psi}_{n};\quad \Psi_{n}\bar{\Psi}_{m}+\bar{\Psi}_{m}\Psi_{n}=0.\label{eq: anticomm}
\end{equation}
Formally the defining relations~\eqref{eq: anticomm} 
look exactly the same in both cases, but the algebras $\TwoCliffNS$ and $\TwoCliffR$ are
naturally viewed as different. For example all generators $\Fgen{k}$ of $\TwoCliffNS$
satisfy $\Fgen{k}^2 = 0$, whereas the generator $\Fgen{n}$ of $\TwoCliffR$
with the index $n=0$ satisfies $\Fgen{0}^2 = \Igen$.

As Clifford modules, each of $\FullFockNS$ and $\FullFockR$
is generated by a single vector, denoted respectively by $\ind$ and
$\bspin,$ with the property that $\Psi_{n}\ind=\bar{\Psi}_{n}\ind=0$
and $\Psi_{n}\sigma=\bar{\Psi}_{n}\sigma=0$ for $n>0.$ The Neveu-Schwarz
module $\FullFockNS$ is the ``freest'' possible one with these
properties; it has a basis given by
\begin{equation}
	\Psi_{n_{1}}\dots\Psi_{n_{k}}\bar{\Psi}_{m_{1}}\dots\bar{\Psi}_{m_{\bar k}}\ind:\quad k,\bar k\geq0,n_{1}<\dots<n_{k}<0,m_{1}<\dots<m_{\bar k}<0,\label{eq: def_module_NS}
\end{equation}
the indices belonging to $\Z+\frac{1}{2}.$ The case of the Ramond
module $\FullFockR$ is slightly more involved. In this case,
its basis is given by 
\begin{equation}
	\Psi_{n_{1}}\dots\Psi_{n_{k}}\bar{\Psi}_{m_{1}}\dots\bar{\Psi}_{m_{\bar k}}\sigma:\,
	k,\bar k\geq0,n_{1}<\dots<n_{k}\leq 0,m_{1}<\dots<m_{\bar k}<0,\label{eq: def_module_R}
\end{equation}
and one has the additional relation 
\begin{equation}
\mu:=\lamb\Psi_{0}\sigma=\lambb\bar{\Psi}_{0}\sigma,\quad\lamb=e^{-\frac{i\pi}{4}}.	\label{eq: mu_psi_psibar}
\end{equation}

As we will see, $\mu$ corresponds to the lattice (Kadanoff--Ceva)
\emph{disorder operator}. We draw the reader's attention to the fact that
$\Psi_{0}\sigma$ and $\bar{\Psi}_{0}\sigma$ are in fact \emph{not}
linearly independent, and thus, unlike the Neveu--Schwarz sector,
the ``double chirality'' module $\FullFockR$ is \emph{not}
a tensor product of two chiral ones. We also note that by (\ref{eq: anticomm})
we have $2\Psi^{2}_{0}=\{\Psi_{0},\Psi_{0}\}=2$ and similarly for
$\bar{\Psi}_{0},$ so that we have $\sigma=\lambb\Psi_{0}\mu=\lamb\bar{\Psi}_{0}\mu.$
Thus, $\mu$ can as well be taken as a generator of $\FullFockR,$
a symmetry expected from the Kramers--Wannier duality. 

The \emph{Sugawara construction} (see Section \ref{sec: background_algebra}) shows that each of
$\FullFockR$ and $\FullFockNS$ carry a representation of the
direct sum of two copies of the Virasoro algebra of central charge
$c=\frac{1}{2}.$ The Virasoro generators $\Loper{n}$ (respectively,
$\Loperbar{n}$) act on any given field as even quadratic polynomials in $\Psi_{n}$ (respectively, $\bar{\Psi}_{n}$); moreover, for each
monomial $\Psi_{m}\Psi_{k}$ featuring in the expression for $\Loper{n}$,
we have $m+k=n$, see Section \ref{sec: background_algebra} 
for specific expressions. Each of $\FullFockNS$ and $\FullFockR$
has two natural gradings: first, $\Z/(2\Z)$ grading by \emph{the
	parity} of $k+\bar k$ in (\ref{eq: def_module_NS}--\ref{eq: def_module_R})
and second, by the weights $(-n_{1}-\dots-n_{k},-m_{1}-\dots-m_{\bar k}),$
which belong to $\left(\frac{1}{2}\Z\right)^{2}$ in the Neveu--Schwarz
sector and to $\Z^{2}$ in the Ramond sector. It follows from the
properties of the Sugawara construction that the Virasoro action preserves
the former grading and is homogeneous with respect to the latter one.
The structure of $\FullFockNS$ and $\FullFockR$
as $\Vir\oplus\Vir$ modules is as follows: we have 
\[
\FullFockNS=\HVM{\ind}\oplus\HVM{\bpsi}\oplus\HVM{\bpsibar}\oplus\HVM{\ben},\quad\FullFockR=\HVM{\bspin}\oplus\HVM{\bmu},
\]
where each direct summand is an irreducible, highest weight module
generated by the highest weight vectors $\ind,$ $\bpsi:=\Psi_{-\frac{1}{2}}\ind,$
$\bpsibar:=\bar{\Psi}_{-\frac{1}{2}}\ind,$ $\ben:=\frac{\i}{2}\Psi_{-\frac{1}{2}}\bar{\Psi}_{-\frac{1}{2}}\ind,$
$\sigma,$ and $\mu:=\lamb\Psi_{0}\sigma$, respectively, with the
$\Loper{0}$ and $\Loperbar{0}$ eigenvalues given by the following table:
\begin{center}
	\begin{tabular}{|c|c|c|}
		\hline 
		\noalign{\vskip1pt}
		$\Op$ & $\Loper{0}$ eigenvalue  & $\Loperbar{0}$ eigenvalue\tabularnewline[1pt]
		\hline 
		\noalign{\vskip1pt}
		\hline 
		\noalign{\vskip1pt}
		$\ind$ & 0 & 0\tabularnewline[1pt]
		\hline 
		\noalign{\vskip1pt}
		$\bpsi$  & $\frac{1}{2}$ & 0\tabularnewline[1pt]
		\hline 
		\noalign{\vskip1pt}
		$\bpsibar$ & 0 & $\frac{1}{2}$\tabularnewline[1pt]
		\hline 
		\noalign{\vskip1pt}
		$\ben$ & $\frac{1}{2}$ & $\frac{1}{2}$\tabularnewline[1pt]
		\hline 
		\noalign{\vskip1pt}
		$\bspin$ & $\frac{1}{16}$ & $\frac{1}{16}$\tabularnewline[1pt]
		\hline 
		\noalign{\vskip1pt}
		$\bmu$ & $\frac{1}{16}$ & $\frac{1}{16}$\tabularnewline[1pt]
		\hline 
	\end{tabular}
	\par\end{center}

The direct sum $M(4,3):=\HVM{\ind}\oplus\HVM{\en}\oplus\HVM{\sigma}$
consists of all the vectors in $\FullFockNS\oplus\FullFockR$ that are
\emph{even} with respect to the fermion parity grading; these fields will
correspond to \emph{even} Grassmann polynomials (in a chosen isoradial graph tiling the plane), possibly with one
spin insertion. As mentioned above, those are in correspondence with
polynomials in spins. The direct sum $\HVM{\bpsi}\oplus\HVM{\bpsibar}$
will corresponds to \emph{odd} Grassmann polynomials, or equivalently,
\emph{even} polynomials in spins with one fermion insertion. Finally,
$\HVM{\bmu}$ will correspond to \emph{odd} Grassmann polynomials
with one spin insertion, or equivalently, \emph{odd} spin polynomials
with one fermion insertion. 

Given $\Delta\in\R,$ the projection onto the subspace $\left\{ \Op\in\FullFockNS\oplus\FullFockR:(\Loper{0}+\Loperbar{0})\Op=\Delta \Op\right\} $
will be denoted by $\WeightProj{\Delta}{\cdot}.$ For $\Op\in\FullFockNS\oplus\FullFockR,$
we denote $\Delta_\Op:=\min\{\Delta:\WeightProj{\Delta}{\Op}\neq0\}$;
it is called the \emph{scaling dimension} of $\Op$. Note that $\Delta_{\Op}\in\Z_{\geq0}$
for $\Op\in\FullFockNS$ and $\Delta_{\Op}\in\frac{1}{8}+\Z_{\geq0}$
for $\Op\in\FullFockR.$

\subsection{Correlations in the continuum. \label{subsec: intro:cont_cor}}

Let $\Omega$ be 
%either a torus $\C/\Lambda$, where $\Lambda$ is a lattice, or 
a finitely-connected sub-domain of $\C$, with some
boundary conditions, e.g., given by a subdivision of the boundary
into finitely many $+,-,$ and free arcs as in \cite{chelkak2021correlations}.
Our aim is to define, given $\Op^{(1)},\dots,\text{\ensuremath{\Op}}^{(k)}\in\FullFockNS\oplus\FullFockR$
and distinct points $z_{1},\dots,z_{k}\in\Omega$, the correlation
functions in the continuum, denoted by 
\[
\ccor{\Op^{(1)}_{z_{1}}\dots\Op^{(k)}_{z_{k}}}_{\Omega}.
\]
These correlations can be expressed in terms of the \emph{spin-fermion
	correlations} and their OPE coefficients, or equivalently, contour
integrals involving the spin-fermion correlations. Since the spin
and fermions are, in CFT terminology, \emph{mutually non-local fields},
these correlations are multi-valued rather than single-valued functions
of $z_{i}'$s. On planar domains, they change signs as some points
travel around each other.
%, and on a torus, they also may change values in more complicated ways as marked points trace non-contractable loops.
Let $\Config:=\{(z_{1},\dots,z_{k})\in\Omega^{k}:z_{i}\neq z_{j}\text{ for }i\neq j\},$
and denote by $\ConfUC$ its universal cover. We take it as our starting
point that the correlation functions $\ccor{\Op^{(1)}_{z_{1}}\dots\Op^{(k)}_{z_{k}}}_{\Omega}$
are defined when $\Op^{(i)}\in\{\ind,\sigma,\psi,\psi^{\star}\}$
for each $i$; see \cite[Section 5]{chelkak2021correlations}.
% for the case of planar domains and \cite{izyurov2021energy,bayrakarogluIzyurov2024} for the case of a torus. 
These correlations have been
shown to be the scaling limits of spin-fermion correlation in the
Ising model \cite[Theorem 1.2]{chelkak2021correlations}, and they can also be written explicitly in terms of special functions
using bosonization identities, see \cite{bayraktaroglu2023bosonization}.
%for finitely connected planar domains and \cite{bayrakarogluIzyurov2024}
%for the torus case.

The spin-fermion correlations can be defined \cite{chelkak2021correlations,bayrakarogluIzyurov2024}
purely analytically as solutions to certain boundary-value problems,
without reference to the underlying model and its scaling limits.
They are uniquely identified by the following properties :
\begin{enumerate}
	\item the correlations\emph{ are (anti-)holomorphic} with respect to the
	positions of fermions \cite[Section 5]{chelkak2021correlations}
	\[
	\partial_{z_{i}}\ccor{\Op^{(1)}_{z_{1}}\dots\psi_{z_{i}}\dots\Op^{(k)}_{z_{k}}}_{\Omega}=0,\quad\bar{\partial}_{z_{i}}\ccor{\Op^{(1)}_{z_{1}}\dots\psistar{z_{i}}\dots\Op^{(k)}_{z_{k}}}_{\Omega}=0.
	\]
	They are two-valued functions that pick up a $-1$ factor every type
	a fermion goes around a spin insertion. 
	\item \emph{the (singular parts of the) OPE }\cite[Section 6]{chelkak2021correlations}
	between the fermions and the fermions and spins: 
	\begin{equation}
		\psi_{z}\psi_{w}=\frac{2}{z-w}+O(1);\quad\psistar z\psistar w=\frac{2}{\bar z-\bar w}+O(1);\quad\psi_{z}\psistar w=O(1);\label{eq: ferm_OPE_continuous}
	\end{equation}
	\begin{equation}
		\psi_{z}\sigma_{w}=\lambb(z-w)^{-\frac{1}{2}}\mu_{w}+O(z-w)^{\frac{1}{2}};\label{eq: ferm_spin_OPE}
	\end{equation}
	\begin{equation}
		\psi_{z}\mu_{w}=\lamb(z-w)^{-\frac{1}{2}}\left(\sigma_{w}+4(z-w)\partial_{w}\sigma_{w}\right)+O(z-w)^{\frac{3}{2}};\label{eq: ferm_mu_OPE}
	\end{equation}
	\begin{equation}
		\sigma_{z}\sigma_{w}=|z-w|^{-\frac{1}{4}}+o(|z-w|^{-\frac{1}{4}}).\label{eq: spins_OPE}
	\end{equation}
	Here, $\mu_{w}$ is a \emph{real }field which can be \emph{defined}
	by (\ref{eq: ferm_spin_OPE}). As usual, the OPEs are understood as
	identities that become true when inserted to any correlations with
	any fields (in this case, spins and fermions) at distant points. E.g.,
	the first equality in (\ref{eq: ferm_OPE_continuous}) above is to
	be understood as 
	\[
	\ccor{\psi_{z}\psi_{w}\Op^{(1)}_{z_{1}}\dots\Op^{(k)}_{z_{k}}}_{\Omega}=\frac{2}{z-w}\ccor{\Op^{(1)}_{z_{1}}\dots\Op^{(k)}_{z_{k}}}_{\Omega}+O(1),\quad z\to w,
	\]
	for any $\Op^{(1)}\dots\Op^{(k)}\in\{\ind,\psi,\psistar{},\sigma\}$,
	where  the implicit constant in $O(1)$ is uniform over positions of marked points away from each other and the boundary.
	\item \emph{boundary conditions} that depend on the underlying domain $\Omega$
	and the choice of the boundary conditions in the discrete model.
\end{enumerate}
The boundary conditions are important to uniquely identify the correlations
in the continuum, and also they are essential for the proofs of convergence
thereof. However, as we take these results as our starting points, we
will not dwell on this any longer.
\begin{rem}
	The spin-fermion correlations are usually only defined when the total
	number of fermion insertions is \emph{even}. Otherwise, we define
	the corresponding correlations to be zero. This will imply that also
	the correlations $\ccor{\Op^{(1)}_{z_{1}}\dots\Op^{(k)}_{z_{k}}}_{\Omega}$
	are identically zero unless the total fermion parity of $\Op^{(1)}_{z_{1}},\dots,\Op^{(k)}_{z_{k}}$
	is even. 
\end{rem}

\begin{defn}
	\label{def: correlations-functions}The correlations functions $\ccor{\Op^{(1)}_{z_{1}}\dots\Op^{(k)}_{z_{k}}}_{\Omega}$
	in a domain $\Omega$ (with boundary conditions incorporated in its
	notation) are defined as a map from $\left(\FullFockNS\oplus\FullFockR\right)^{\otimes k}$
	to the space of real-analytic functions on $\ConfUC,$ 
	\[
	\Op^{(1)}\otimes\dots\otimes\Op^{(k)}\mapsto\left((z_{1},\dots,z_{k})\mapsto\ccor{\Op^{(1)}_{z_{1}}\dots\Op^{(k)}_{z_{k}}}_{\Omega}\right)
	\]
	inductively, in terms of the spin-fermions correlations, as follows: 
	\begin{enumerate}
		\item If $\Op^{(1)},\dots,\text{\ensuremath{\Op}}^{(k)}\in\{\ind,\bspin\},$
		we define $\ccor{\Op^{(1)}_{z_{1}}\dots\Op^{(k)}_{z_{k}}}_{\Omega}:=\ccor{\prod_{i:\Op^{(i)}=
			\sigma}\sigma_{z_i}}_{\Omega}$
		\item When each $\Op^{(1)},\dots,\text{\ensuremath{\Op}}^{(k)}$ is one
		of the basis elements (\ref{eq: def_module_NS}), (\ref{eq: def_module_R}) of $\FullFockNS\oplus\FullFockR,$ we define 
		\begin{equation}
			\ccor{\Op^{(1)}_{z_{1}}\dots(\Curr{n}\Op^{(i)}_{z_{i}})\dots\Op^{(k)}_{z_{k}}}_{\Omega}=\frac{1}{2\pi\i}\oint_{(z_{i})}\ccor{\Op^{(1)}_{z_{1}}\dots\psi_{w}\Op^{(i)}_{z_{i}}\dots\Op^{(k)}_{z_{k}}}_{\Omega}(w-z_{i})^{n-\frac{1}{2}}\,dw;\label{eq: defcor1}
		\end{equation}
		\begin{equation}
			\ccor{\Op^{(1)}_{z_{1}}\dots(\Currbar{n}\Op^{(i)}_{z_{i}})\dots\Op^{(k)}_{z_{k}}}_{\Omega}=-\frac{1}{2\pi\i}\oint_{(z_{i})}\ccor{\Op^{(1)}_{z_{1}}\dots\psistar w\Op^{(i)}_{z_{i}}\dots\Op^{(k)}_{z_{k}}}_{\Omega}(\bar w-\bar z_{i})^{n-\frac{1}{2}}\,d\bar w.\label{eq: defcorbar}
		\end{equation}
		Here $n\in\Z+\frac{1}{2}$ if $\Op^{(i)}\in\FullFockNS$ and $n\in\Z$
		if $\Op^{(i)}\in\FullFockR,$ and the notation $\oint_{(z_{i})}$
		means a small\footnote{i.e., not encircling other $z_{j}$; but it should enclose any small
			contours around $z_{i}$ that may have been involved in the definition
			of the correlation function in the integrand.} counterclockwise contour around $z_{i}$. 
		\item We then extend the definition to the case $\Op^{(1)},\dots,\text{\ensuremath{\Op}}^{(k)}\in\FullFockNS\oplus\FullFockR$
		by linearity.
	\end{enumerate}
\end{defn}

By Fubini, the correlation functions thus defined do not depend on
the order in which we add $\Psi$'s at \emph{different} points at
the second step in the above construction. Moreover, it is clear that the OPE (\ref{eq: ferm_OPE_continuous}, \ref{eq: ferm_spin_OPE}, and \ref{eq: ferm_mu_OPE})
continue to hold true when the ``spectator'' fields $\Op^{(1)},\dots,\text{\ensuremath{\Op}}^{(k)}$
belong to $\FullFockNS\oplus\FullFockR$ rather than $\{\ind,\psi,\psistar{},\sigma\}.$
An easy (and standard) satellite-integral computation shows that the OPE (\ref{eq: ferm_OPE_continuous})
implies the anti-commutation relations \eqref{eq: anticomm},
understood inside the correlations, for instance, we have 
\begin{multline*}
	\ccor{\Op^{(1)}_{z_{1}}\dots\Psi_{n}\Psi_{m}\Op^{(i)}_{z_{i}}\dots\Op^{(k)}_{z_{k}}}_{\Omega}+\ccor{\Op^{(1)}_{z_{1}}\dots\Psi_{m}\Psi_{n}\Op^{(i)}_{z_{i}}\dots\Op^{(k)}_{z_{k}}}_{\Omega}\\
	\equiv2\delta_{n+m,0}\ccor{\Op^{(1)}_{z_{1}}\dots\Op^{(i)}_{z_{i}}\dots\Op^{(k)}_{z_{k}}}_{\Omega}
\end{multline*}
for any choice of the fields $\Op^{(1)}_{z_{1}},\dots,\Op^{(k)}_{z_{k}}.$
Moreover, we have $(\Psi_{n}\sigma)\equiv0$ and $(\Psi_{n}\ind)\equiv0$ for $n>0;$ thanks to
the spin-fermion OPE starting from power $(z-w)^{-\frac{1}{2}}$ and the identity-fermion OPE not having any singular part. Also, inside correlations we do have \eqref{eq: mu_psi_psibar}, as ensured by the fermion-spin OPE \cite[Theorem 6.2]{chelkak2021correlations}.
These observations imply that Definition \ref{def: correlations-functions}
indeed defines the correlation functions for any fields $\Op^{(1)},\dots,\text{\ensuremath{\Op}}^{(k)}\in\FullFockNS\oplus\FullFockR,$
in that the definition respects the module structure. As a final remark,
it follows from~(1) and~(2) and (anti)holomorphicity of spin-fermion correlations with respect to
fermion positions that
$(\Psi_{-\frac{1}{2}}\ind)_{z_{i}}=\psi_{z_{i}}$ and $(\bar{\Psi}_{-\frac{1}{2}}\ind)_{z_{i}}=\psistar{z_{i}},$
so that the spin-fermion correlations themselves fit into the above
setup. 
\begin{rem}
	Planar domains (and a torus) are special among Riemann surfaces in that
	they have a distinguished coordinate chart $z$ and a distinguished
	differential $dz$, which simplifies the above discussion a lot. On
	general Riemann surfaces, the correlations of primary fields are forms
	(of fractional degree), and they cannot really be identified with
	functions. The Clifford action in Definition \ref{def: correlations-functions}
	then depends on a choice of a coordinate chart near each of the points
	$z_{1},\dots,z_{k}$. 
\end{rem}

\subsection{Lattice fields and their correlations}

\label{subsec: local_lattice_fields_intro}Let $\isG$ be a (finite
or infinte) planar graph embedded isoradially in the plane, with a
distinguished \emph{root corner }$\root$ (a corner is a midpoint of the
segment connecting a vertex $\root^{\circ}$ of $\isG$ with a vertex
$\root^{\bullet}$ of its dual graph). Recall that the \emph{mesh size} of $\isG$, which we will denote by $\delta$, is defined as the distance between
	a vertex and an incident dual vertex, which is constant by isoradiality. The reader interested in the
simplest setup may think of $\isG$ being the square lattice $\sqrt{2}\delta\Z^{2}$ 
and $o^{\circ}$ being the origin. We denote by $\isG^{\circ},\isG^{\bullet},$
and $\Cgr(\isG)$ the set of all vertices, dual vertices, and corners
of $\isG$, respectively. We will consider the following vector spaces
associated to $\isG$:
\begin{enumerate}
	\item The space of all (formal) \emph{spin polynomials} 
	\begin{multline*}
		\LF=\LF(\isG):=\C\langle\{\sigma_{v}\}_{v\in\isG^{\circ}}\rangle/\langle\sigma^{2}_{v}=1\,\forall v,
		\sigma_{v}
		\sigma_{v'}=
		\sigma_{v'}\sigma_{v}\,\forall v,v'\rangle\\
		=\left\{ \sum_{\substack{S\subset\isG^{\circ} \\ \text{finite}}}\alpha_{S}\prod_{v\in S}\sigma_{v}:\alpha_{S}=0\text{ for all but finitely many }S\subset\isG^{\circ}\right\} ;
	\end{multline*}
	hereinafter $\C\langle X\rangle$ denotes the free associative algebra
	over $X$. We denote by $\LFe$ (respectively, $\LFo$) the space
	of all even (respectively, odd) spin polynomials. 
	\item The space $\LFPe(\isG)$ of all (formal) \emph{Grassmann polynomials
		in fermions, }which is the space of all polynomials in anti-commuting
	variables $\psi_{z}$\emph{ }indexed by \emph{corners} $z\in\Cgr(\isG):$\emph{
	}
	\[
	\LFPe=\LFPe(\isG):=\C\langle\{\psi_{z}\}_{z\in\Cgr(\isG)}\rangle/\langle\psi_{z}\psi_{w}=-\psi_{w}\psi_{z}\,\forall z,w\rangle.
	\]
	This space has a direct sum decomposition $\LFPe=\LFPee\oplus\LFPeo$
	into the spaces of even and odd Grassmann polynomials.
	\item The space of all (formal) \emph{Grassmann polynomials with one spin
		insertion}, i.e., expressions of the form $\Op\sigma_{\root^{\circ}},$
	where $\Op$ is a polynomial in anti-commuting variables $\psi_{z}$
	indexed by corners of the \emph{double cover} $\Cgr(\isG_{[\root^{\circ}]})$
	ramified at $\root^{\circ}$, subject to the relations $\psi_{z}=-\psi_{z^{\star}},$
	where $z \ne z^{\star}$ are in the fiber over the same point:
	\[
	\LFPo=\LFPo(\Gamma):=\left\{ \sigma_{\root^{\circ}}\Op:\Op\in\C\langle\{\psi_{z}\}_{z\in\Cgr(\isG_{[\root^{\circ}]})}\rangle/\langle\psi_{z}\psi_{w}=-\psi_{w}\psi_{z}\,\forall z,w,\,\psi_{z}=-\psi_{z^{\star}}\forall z\rangle\right\} .
	\]
	This space also has a direct sum decomposition $\LFPo=\LFPoe\oplus\LFPoo$,
	corresponding to $\Op$ being even or odd.
\end{enumerate}
The vectors in any of these spaces or in their direct sum will be called
\emph{lattice fields}. For $\Op\in\LF(\isG),$ define its \emph{radius
	of support} 
\[
R_{\Op}:=\max\{|v-\root|:\exists S\ni v\,\text{s.t.}\,\alpha_{S}\neq0\}+2\delta;
\]
we similarly define $R_{\Op}$ for $\Op\in\LFPe(\isG)$ or $\Op\in\LFPo(\isG).$

Fix a rooted isoradial graph $(\Gamma,\root)$ of mesh size $1$ tiling the plane. Let $\Od$
be a (finite or infinite) isoradial graph of mesh size $\delta,$
and $w\in\Cgr(\Od).$ We say that \emph{the $\delta R$--ball of }$\Od$\emph{
	at $w$ agrees with $\isG$} if $R$--neighborhoods of $\root$ in
$\isG$ and in $\root+\delta^{-1}(\Od-w)$ coincide. Given $\Op\in\LF(\isG),$
we may consider the \emph{insertion} of $\Op$ at $w$ on $\Od,$
\[
\text{\ensuremath{\Op}}_{\delta,w}=\sum\alpha_{S}\prod_{v\in S}\sigma_{\delta(v-\root)+w}\in\LF(\Od),
\]
provided that the $\delta R_{\Op}$--ball of $\Od$ at $w$ agrees
with $\isG.$ Similarly, we define $\text{\ensuremath{\Op}}_{\delta,w}$
for $\Op\in\LFPe(\isG)$ or $\Op\in\LFPo(\isG).$ In the simplest
case when $\isG$ is a square lattice $\sqrt{2}\Z^{2}$ and $\isG_{\delta}=\delta\Gamma$, then the insertion
is always possible, as long as the corner $w$ has the same orientation
as $\root$.

Let now $\Od$ be a finite isoradial graph, and consider the Ising
model on $\Od$ with (Baxter's) $Z$-invariant critical weights \cite{ChelkakSmirnov2},
with some boundary conditions, e.g. as in \cite{chelkak2021correlations}.
Given infinite rooted isoradial graphs $\isG^{(1)},\dots,\isG^{(k)}$ of mesh size $1$
tiling the plane and $\Op^{(i)}\in\LF(\isG^{(i)}),$ we may
define the correlation function 
\begin{equation}
	\E_{\Od}[\Op^{(1)}_{\delta,z_{1}}\dots\Op^{(k)}_{\delta,z_{k}}],\label{eq: def_cor_intro}
\end{equation}
provided that the $\delta R_{\Op^{(i)}}$-ball of $\Od$ at $z_{i}$ agrees
with $\isG^{(i)},$ simply by substituting the spins in the model into
the formal spin variables. We can extend this definition also to the
case $\Op^{(i)}\in(\LFPe\cup\LFPo)(\Gamma^{(i)}),$ using some additional
care to fix the signs. Let $\cvr:\Od_{\cvr}\to\Od$ be the double
cover of $\Od$ ramified at those of $z^{\circ}_{i}$ for which $\Op^{(i)}\in\LFPo.$
For every such $i$, fix one of the two embeddings of the $\delta\RadO{\Op^{(i)}}$--neighborhood
of $\Cgr(\delta\Gamma^{(i)}_{[o^{\circ}_{i}]})$ into $\Cgr(\Od_{\cvr})$,
and for every $i$ such that $\Op^{(i)}\in\LFPe,$ fix a lift of the  $\delta\RadO{\Op^{(i)}}$--neighborhood
of $\Cgr(\delta\Gamma^{(i)}_{[o^{\circ}_{i}]})$ onto $\Cgr(\Od_{\cvr})$ (in both cases, the root $\root^{(i)}$ of $\isG^{(i)}$ is mapped to $z_i$).
With this data, when each $\Op^{(i)}$ is a monomial, we can understand
the correlation (\ref{eq: def_cor_intro}) as the spin-fermion correlation
as defined in \cite[Section 2]{chelkak2021correlations}, substituting
the spins and the fermionic insertions in the model into the formal
spin and fermion variables. As in the continuum case, we define the
correlation (\ref{eq: def_cor_intro}) to be zero if the total number
of fermion insertions is odd. We furthermore extend the definition
to the case $\Op^{(i)}\in(\LFPe\oplus\LFPo\oplus\LF)(\Gamma^{(i)})$
by linearity\footnote{choosing all the embeddings and lifts required to make sense of every monomial correlation appearing in the expansion}. 

Following \cite{hongler2022conformal}, we say that a field $\Op\in(\LFPe\oplus\LFPo\oplus\LF)(\isG)$
is a \emph{null field} if it vanishes in all correlations with distant
fields, i.e., if there exists an $R>0$ such that the correlation
$\E_{\Od}[\Op_{\delta,w}\Op^{(1)}_{\delta,z_{1}}\dots\Op^{(k)}_{\delta,z_{k}}]$
vanishes identically for all possible domains $\Od$ and fields $\Op^{(i)}$,
whenever the $\delta R$-ball of $\Od$ at $w$ agrees with $\isG$
and does not intersect $\delta R_{\Op^{(i)}}$ balls centered at $z_{i}$.
The space of all null fields will be denoted by $\Null$. As we explain
in Section \ref{subsec:Local-fields}, it is easy to see that we have
linear isomorphisms
\begin{equation}
	\LFPee/\Null\cong\LFe/\Null;\quad\LFPoe/\Null\cong\LFo/\Null.\label{eq: spin-ferm-isomorphism}
\end{equation}
In other words, for every $\Op\in\LF$, one can find a field $\hat{\Op}\in\LFPoe\oplus \LFPee$
that has the same correlations with all (sufficiently distant) fields,
and the other way around. 

\subsection{Main results}

We are ready to state our main result:
\begin{thm}
	\label{thm: main2} For each infinite rooted isoradial graph $(\isG,\root)$
	of mesh size $1$, there is a surjective linear map $\MainBij_{\isG}$
	from $\LFPe(\isG)\oplus\LFPo(\isG)$ to $\FullFockNS\oplus\FullFockR$
	such that:
	\begin{enumerate}
		\item The kernel of $\MainBij_{\isG}$ consists precisely of null fields;
		\item $\MainBij_{\isG}(1)=\ind$; $\MainBij_{\isG}(\sigma_{\root^{\circ}})=\Csigma\cdot\sigma,$
		where $\Csigma=2^{\frac{1}{6}}e^{\frac{3}{2}\zeta'(-1)}$;
		\item $\MainBij_{\isG}(\LFPee(\isG))=\HVM{\ind}\oplus\HVM{\en};$ $\MainBij_{\isG}(\LFPoe(\isG))=\HVM{\sigma};$
		$\MainBij_{\isG}(\LFPeo(\isG))=\HVM{\bpsi}\oplus\HVM{\bpsibar};$ $\MainBij_{\isG}(\LFPoo)=\HVM{\mu}.$
		\item Let $\Omega$ be a finitely-connected planar domain and $\Omega^{\delta}$
		its approximation by isoradial graphs of mesh size $\delta$. Let
		$k\in\N$, $\isG_{1},\dots,\isG_{k}$ be infinite, rooted isoradial
		graphs tiling the plane, and $\Op^{(i)}\in\LFPe(\isG_{i})\oplus\LFPo(\isG_{i})$.
		Assume that the approximations $\Od$ and corners $z_{i}=z^{\delta}_{i}\in\Cgr(\Od),i=1,\dots,k$
		are such that for each $\delta>0$ small enough, the $\delta R_{\Op^{(i)}}$-ball
		of $\Od$ at $z_{i}$ coincides with $\Gamma_{i}$. Then, 
		\[
		\delta^{-\Delta_{1}-\Delta_{k}}\E_{\Od}[\Op^{(1)}_{\delta,z_{1}}\dots\Op^{(k)}_{\delta,z_{k}}]=\ccor{\WeightProj{\Delta_{1}}{\MainBij_{\isG_{1}}(\Op^{(1)})}_{z_{1}}\dots\WeightProj{\Delta_{k}}{\MainBij_{\isG_{k}}(\Op^{(k)})}_{z_{k}}}_{\Omega}+o(1),
		\]
		uniformly as $\delta\to0$ in $z_{1},\dots,z_{k}$ staying away from
		each other and $\partial \Omega$. Here the scaling dimensions $\Delta_{i}=\Delta_{\MainBij_{\isG_{i}}(\Op^{(i)})}$
		and the projectors $\WeightProj{\Delta_{i}}{\cdot}$ are as defined
		at the end of Section \ref{subsec: intro_Fock}.
	\end{enumerate}
\end{thm}

\begin{rem}
	\label{rem: generality}
	The convergence part (4) of the above result will be derived by treating the
	convergence of spin-fermion correlations, proven in \cite[Theorem 1.2]{chelkak2021correlations}, \cite{chelkak2023universality}, as a black box. Therefore, (4) immediately extends to other contexts as soon as this spin-fermion correlation convergence is proven. Besides the finitely connected planar domain setup, this has been essentially done for the square lattice discretization of a torus, see \cite{izyurov2024energy,bayrakarogluIzyurov2024}. The missing bit is the convergence of pure fermionic correlations, the difficulty
		being a non-trivial kernel of the (discrete) Cauchy-Riemann operators.
		This difficulty can be overcome by the same methods as employed in
		\cite{izyurov2024energy} for the computation of energy correlations, but to our knowledge, this has not been explicitly done in the literature.
	The case of isoradial approximations to arbitrary
	Euclidean surface with a boundary and conical singularities is a work
	in progress.
\end{rem}

\begin{rem}
	\label{rem: bij_uniqueness}
	The correspondence $\MainBij_{\isG}$ is not unique or canonical.
	However, if $\MainBij'_{\isG}$ is another such correspondence, then
	$\WeightProj{\Delta_{\Op}}{\MainBij_{\isG}(\Op)}=\WeightProj{\Delta_{\Op}}{\MainBij'_{\isG}(\Op)}$
	for any~$\Op$.
\end{rem}

The following corollary, connecting lattice spin polynomials with the states of $M(4,3)$, follows easily from the above
theorem, particularly item (3), and the isomorphisms (\ref{eq: spin-ferm-isomorphism}):
\begin{cor}
	\label{thm: main1} For each infinite rooted isoradial graph $(\isG,\root)$
	of mesh size $1$, there is a surjective linear map $\MainBijSpins_{\isG}$
	from $\LF(\isG)$ to $M(4,3)$ such that:
	\begin{enumerate}
		\item The kernel of $\MainBijSpins$ consists precisely of null fields;
		\item $\MainBijSpins(1)=\ind$, $\MainBijSpins(\sigma_{\root^{\circ}})=\Csigma\cdot\sigma$,
		$\MainBijSpins(\en_{e})=\Ceps\cdot\en$ where $e=(\root^{\circ}v)$
		is an edge of $\isG$ incident to $\root^{\circ}$, $\en_{e}=(\sin \theta_{e})^{-1}\sigma_{\root^\circ}\sigma_v-\cot{\theta_{e}}$
		is the energy observable on that edge, $\theta_{e}$ being half the rhombus angle against the edge $e$, and $\Ceps=\frac{2}{\pi}$;
		\item In the setup of item (4) of Theorem \ref{thm: main1}, we have  
		\[
		\delta^{-\Delta_1-\Delta_k}\E_{\Od}[\Op^{(1)}_{\delta,z_{1}}\dots\Op^{(k)}_{\delta,z_{k}}]=\ccor{\WeightProj{\Delta_1}{\MainBijSpins(\Op^{(1)})}_{z_{1}}\dots\WeightProj{\Delta_k}{\MainBijSpins(\Op^{(k)})}_{z_{k}}}_{\Omega}+o(1),
		\]
		as $\delta\to0$ and $\Od\to\Omega$, uniformly in $z_{1},\dots,z_{k}$
		staying away from each other and $\pa\Omega$, and $\Delta_{i}=\Delta_{\MainBijSpins_{\isG_{i}}(\Op^{(i)})}$
	\end{enumerate}
	
\end{cor}

	Remarks \ref{rem: generality} and \ref{rem: bij_uniqueness} apply to the above corollary as well.

\section{Discrete complex analysis}

\subsection{Calculus of residues for s-holomorphic functions and spinors. }

\label{subsec: s-holomorphicity}In this section, we collect the definitions and some elementary properties of s-holomoprhic functions, see e.g. \cite{ChelkakSmirnov2,chelkak2023universality}. Let $\isG$ be an isoradial graph
tiling the plane; we denote by $\isG^{\circ}$ the set of its vertices,
by $\isG^{\bullet}$ the set of its dual vertices, and by $\Cgr(\isG)$
its \emph{corner graph}. The vertices of $\Cgr(\isG)$
correspond to midpoints of edges of rhombi formed by segments connecting
vertices to adjacent dual vertices, and edges of $\Cgr(\isG)$ connect
corners that belong to neighboring edges of the same rhombus, see Figure \ref{fig: isoradial}. For
$a\in\Cgr(\isG),$ we denote by $a^{\circ}$ and $a^{\bullet}$ the
incident vertices of $\isG^{\circ}$ and $\isG^{\bullet}$ respectively.
For a corner $a\in\Cgr(\isG)$, we denote 
\[
\eta_{a}=e^{\frac{i\pi}{4}}\left(\frac{a^{\bullet}-a^{\circ}}{|a^{\bullet}-a^{\circ}|}\right)^{-\frac{1}{2}}=e^{i\frac{\pi}{4}}e^{-i\frac{\alpha_{a}}{2}},\quad\alpha_{a}=\arg(a^{\bullet}-a^{\circ}).
\]
Note that $\eta_{a}$ is defined up to sign; we will uphold the convention that if it appears several times in the same equation, then the same sign is chosen. The \emph{mesh size}
of $\isG$ is defined to be $\delta=|a^{\bullet}-a^{\circ}|,$ which is independent of $a$; we
also put $\lamb=e^{-\i\pi/4}.$ 

\begin{figure}
	\centering{}\includegraphics[width=0.8\textwidth]{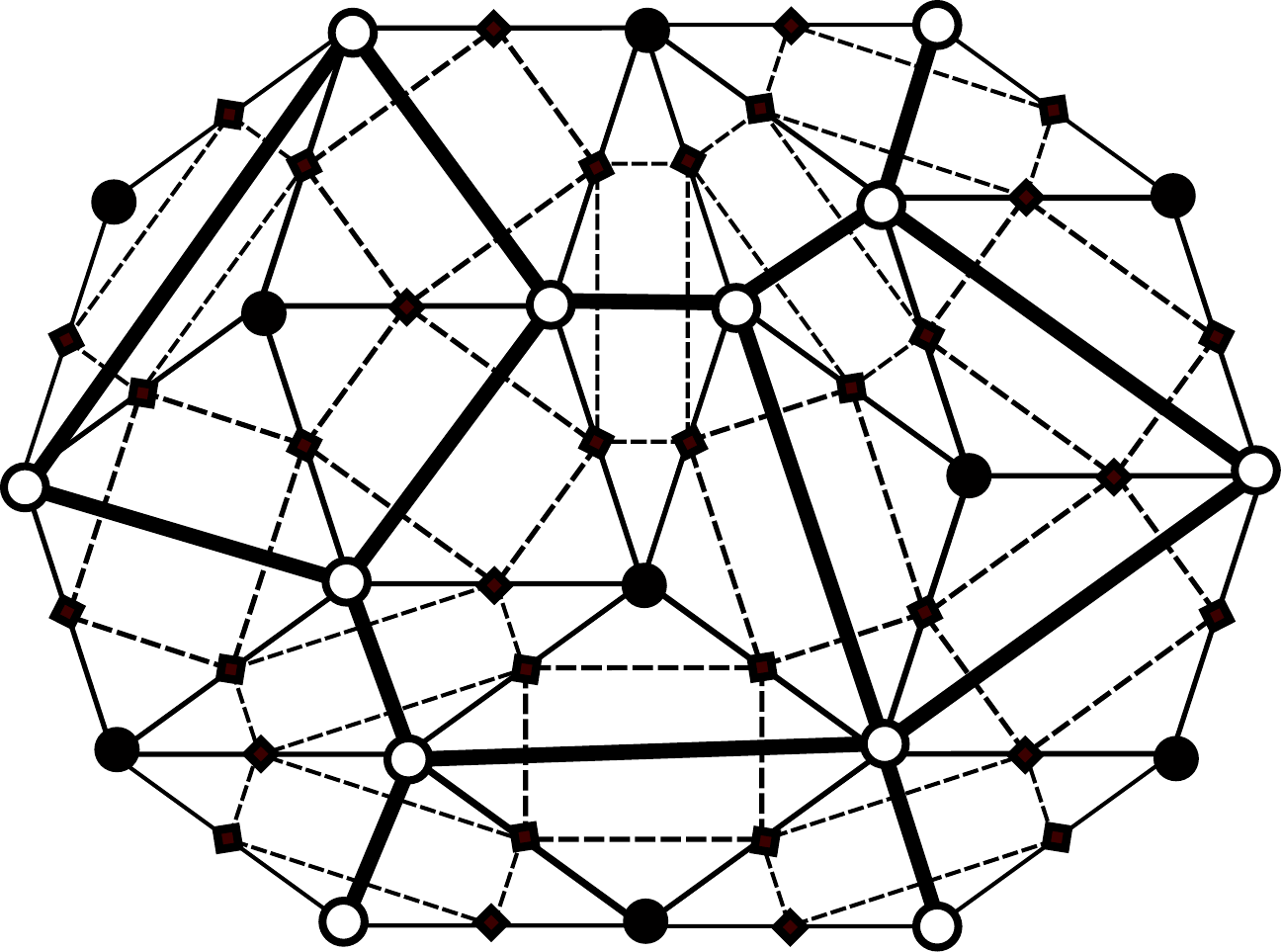}
	\label{fig: isoradial}
	\caption{A fragment of an isoradial graph $\Gamma$ tiling of the plane (vertices as $\circ$, edges as thick solid lines, dual vertices as $\bullet$), the associade rhombi tiling (sides of rhombi as thin solid lines) and the corner graph $\Cgr(\isG)$ (vertices as squares, edges as dashed lines).}
\end{figure}

We will also consider \emph{double covers} $\cvr:\Cgr(\isGd)\to\Cgr(\isG)$
of $\Cgr(\isG),$ ramified at subsets of $\isG^{\circ}\cup\isG^{\bullet},$
which are two-to one maps preserving the graph structure, see Fig. \ref{fig: isoradial} for an example. Our most
important examples will be $\Cgr(\isGoc)$ and $\Cgr(\isGob),$ ramified
at single vertices $\root^{\circ}\in\isG^{\circ}$ and $\root^{\bullet}\in\isG^{\bullet}$ (i.e., going around that vertex leads to a different sheet).
A \emph{spinor} $F$ on a double cover $\Cgr(\isG_{\cvr})$ is a function
$F:\Cgr(\isG_{\cvr})\to\C$ which changes signs between sheets. 

A spinor $F:\Cgr(\isG_{\cvr})\to\C$ is called \emph{s-holomorphic}
if, for every $z\in\Cgr(\isG_{\cvr}),$ it satisfies the \emph{phase condition}
\begin{equation}
	F(z)\in\eta_{z}\R\label{eq: phase}
\end{equation}
and the \emph{three-term relation} 
\begin{equation}
	F(z)=\proj{\eta_{z}}{F(z_{1})+F(z_{2})},\label{eq: three-term}
\end{equation}
where $z_{1},z,z_{2}$ are three consecutive edges of the same rhombus,
and $\proj{\eta}A=\frac{1}{2}(A+\eta^{2}\bar A)$ is the projection
of the complex number $A$ onto the line $\eta\R.$ 

Given \eqref{eq: phase}, the equation \eqref{eq: three-term} is equivalent to the relation  
\begin{equation}
	\bar{\eta}_zF(z)=\left(\cos\theta\right) \bar\eta_{z_1} F(z_1)+\left(\sin\theta\right) \bar\eta_{z_2} F(z_2),\label{eq: three-term-real}
\end{equation}
between the three real numbers $\bar{\eta}_zF(z),$ $ \bar\eta_{z_{1}} F(z_{1}),$  $\bar\eta_{z_2} F(z_2),$ where $\theta$ is half the angle between sides $z,z_1$ of the rhombus, and the signs of $\eta_z,$ $\eta_{z_1},$ $\eta_{z_2}$ are related by continuously rotating the segment $[w^\bullet,w^\circ]$ inside the rhombus as $w$ takes values $z,z_1,z_2$. In other words,  $\eta_{z_1}=e^{-\i \theta} \eta_z$ and $\eta_{z_2}=\i e^{-\i\theta} \eta_z$ if $z_1,z,z_2$ follow the boundary of the rhombus in the counterclockwise direction.

For two s-holomorphic functions $F,G$ on $\Cgr(\isG)$, or more generally
two s-holomorphic spinors on the same double cover $\Cgr(\isGd)$
of $\Gamma$, we have an integral 
\[
I:\isG^{\circ}\cup\isG^{\bullet}\to\R,\quad I(w)=\im\int^{w}F(z)G(z)\,dz,
\]
defined as follows. For a nearest-neighbor path $\gamma$ on $\isG^{\bullet}\cup\isG^{\circ},$
put $s_{\gamma}(c)=1$ if $c\in\Cgr(\isG)$ is oriented from $c^{\circ}$
to $c^{\bullet}$ in $\gamma$ and $s_{\gamma}(c)=-1$ else. Then,
we put 
\begin{equation}
	\im\int_{\gamma}F(z)G(z)\,dz:=\sum_{c\in\gamma}2\delta s_{\gamma}(c)\bar{\eta}^{2}_{c}F(c)G(c).\label{eq: def-Im-Int}
\end{equation}
In other words, $I(z^{\bullet})-I(z^{\circ})=-2\i F(z)G(z)(z^{\bullet}-z^{\circ})=2\delta\bar{\eta}^{2}_{z}F(z)G(z).$
For $F=G,$ we have 
\[
I(z^{\bullet})-I(z^{\circ})=-2\i\eta^{2}_{z}|F(z)|^{2}(z^{\bullet}-z^{\circ})=2\delta|F(z)|^{2}.
\]
It is elementary to check \cite{ChelkakSmirnov2} that the integral
thus defined is independent of the path.

Importantly, the above integral is a discretization of the corresponding  continuous integral. Namely, if we have a family s-holomorphic functions $F^\delta,G^\delta$ on a refining family of graphs $\Gamma^\delta$, such that as $\delta\to 0$, we have  $F^\delta(z)=\proj{\eta_{z}}{f(z)}+o(1)$
and $G^\delta(z)=\proj{\eta_{z}}{g(z)}+o(1)$ locally uniformly, where $f(z)$ and $g(z)$ are
continuous functions\footnote{in which case $f,g$ are necessarily
	holomorphic \cite{ChelkakSmirnov2}.}, then we have 
\begin{multline*}
	I(z^{\bullet})-I(z^{\circ})=\frac{1}{2}\delta\left(\bar{\eta}_{z}f(z)+\eta_{z}\bar{f(z)}\right)\left(\bar{\eta}_{z}g(z)+\eta_{z}\bar{g(z)}\right)+o(\delta)=\\
	=\delta\re\left[\bar{\eta}^{2}_{z}f(z)g(z)\right]+\delta\re\left[\bar{f(z)}g(z)\right]+o(\delta).
\end{multline*}
By summing over a path, the second term approximately telescopes,
and the first term is $\re\left[-\i(z^{\bullet}-z^{\circ})f(z)g(z)\right]=\im\left[f(z)g(z)(z^{\bullet}-z^{\circ})\right].$
Thus, we have 
\begin{equation}
	\im\int_{\gamma^{\delta}}F^\delta(z)G^\delta(z)\,dz=\im\int_{\gamma}f(z)g(z)\,dz+o(1),\label{eq: Im_int_approximation}
\end{equation}
where $\gamma^{\delta}$ is an approximation to $\gamma.$ 

Given a double cover $\cvr$ and (distinct) corners $a_{1},\dots,a_{n}\in \Cgr(\isG)$, we define a graph $\Cgr_{\{a_{1},\dots,a_{n}\}}(\isGd)$ modified at corners projecting to $a_{i}$ as follows: each vertex $a$ of $\Cgr(\isGd)$ such that  $\cvr(a)=a_{i}$ for some $i$ is cut by the rhombus edge
$[a^{\bullet},a^{\circ}]$ into two degree-two vertices $a^{+}$
and $a^{-}$, to the left (resp., to the right) as seen from $a^{\bullet}$,
see Fig. \ref{fig: modified_double}.

\begin{figure}
	\centering{}\includegraphics[width=0.8\textwidth]{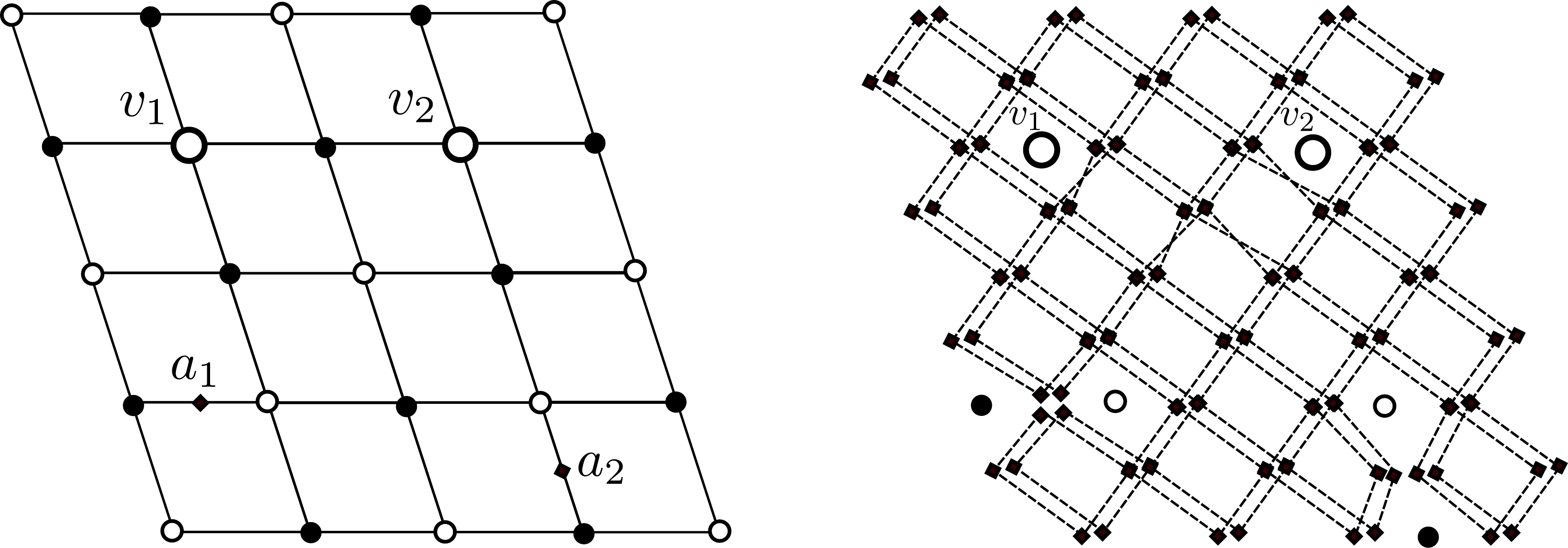}
	\label{fig: modified_double}
	\caption{Left: a fragment of a rhombi tiling of the plane, with two marked vertices $v_1,v_2\in\isG^\circ$ and two marked corners $a_1,a_2$. Right: the corresponding graph $\Cgr_{a_1,a_2}(\isGd),$ where the double cover $\cvr$ is ramified at $v_1$ and $v_2$.} 
\end{figure}

We say that a spinor $F$ is \emph{s-meromorphic with discrete poles
at $a_{1},\dots,a_{n}$} if it is defined and s-holomorphic on $\Cgr_{\{a_{1},\dots,a_{n}\}}(\isGd);$
with the three-term relation (\ref{eq: three-term}) understood as follows: if one or several of the arguments $z,z_1,z_2$ is equal to $a$ such that $\cvr(a)=a_i$ for some $i$, we replace it with $a^{+}$
or $a^{-}$, whichever lies on the side of the rhombus
in question. To shorten the notation, we will sometimes take a liberty to say that the poles are at $a_i\in \Cgr(\isG_\cvr)$ rather than at $\cvr(a_i)\in \Cgr(\isG)$, with the same meaning.

We have the following discrete analog of the residue theorem:
\begin{lem}
	If $F,G$ are two spinors on $\Cgr_{\{\mathcal{A}\}}(\isGd)$ that
	are s-meromorphic inside a simple counterclockwise contour $\gamma$ with discrete poles in
	$\mathcal{A}\subset \mathrm{inside}(\gamma)$, then we have 
	\begin{equation}
		\im\int_{\gamma}F(z)G(z)\,dz=\sum_{a\in\mathcal{A}}2\delta\bar{\eta}^{2}_{a}\left(F(a^{+})G(a^{+})-F(a^{-})G(a^{-})\right).\label{eq: Discrete residue formula}
	\end{equation}
\end{lem}

\begin{proof}
	The contour $\gamma$ is homotopic on $\Cgr_{\{a_{1},\dots,a_{n}\}}$
	to a union of two-edge contours running from $a^{\circ}$ to $a^{\bullet}$
	via $a^{+}$ and back via $a^{-},$ for each $a\in\mathcal{A}.$
\end{proof}

One example of an s-holomorphic function is a constant function $F(z)\equiv\proj{\eta_{z}}A$
for $A\in\C.$ Our main example is the \emph{spin-fermion correlations}, as defined in \cite[Section 2]{chelkak2021correlations}, see also section \ref{subsec:Local-fields}:
\begin{example}
	Given corners $a_{1},\dots,a_{n},$ the correlation 
	\[
	F(z)=\bar{\eta}_{a_{1}}\dots\bar{\eta}_{a_{n}}\E(\psi_{z}\psi_{a_{1}}\dots\psi_{a_{n}}\Op)=\E(\psi_{z}\sigma_{a^{\circ}_{1}}\mu_{a^{\bullet}_{1}}\dots\sigma_{a^{\circ}_{n}}\mu_{a^{\bullet}_{n}}\Op)
	\]
	where $\Op=\sigma_{v_{1}}\dots\sigma_{v_{m}}$, is an s-meromorphic
	spinor ramified at $v_{1},\dots,v_{m}$ with discrete residues at
	$a_{1},\dots,a_{n}.$ Specifically, the values are given by 
	\[
	F(a^{\pm}_{i})=\pm(-1)^{i+1}\eta_{a_{i}}\prod_{k\neq i}\bar{\eta}_{a_{k}}\E(\prod_{k\neq i}\psi_{k}\Op).
	\]
	On sign choices, see \cite[Definition 2.23]{chelkak2021correlations} and Section \ref{subsec:Local-fields} below.
	
	For a while from now on, we fix an infinite isoradial graph $\isG$
	with a distinguished corner $\root$ which we think of as located at
	the origin. Until the end of Section \ref{sec: good_monomials}, we also fix the mesh size to $\delta=1$.
\end{example}

\subsection{The discrete Cauchy kernel and the discrete $z^{-\frac{1}{2}}$.}

In this section, we recall the construction of discrete Cauchy kernel
(an s-holomorphic analog of $1/(z-a)$), first constructed in \cite{kenyon2002laplacian}
and the discrete analogs of $(z-o)^{-\frac{1}{2}},$ first constructed
in \cite{DubedatDimers}. More precisely, given $a\in\Cgr(\isG)$ together
with a choice of the sign of $\eta_{a},$ we define an s-meromorphic
function $\CK{\cdot}a$ with a discrete pole at $a$ satisfying $\CK{a^{\pm}}a=\pm\eta_{a}$
and the asymptotics 
\begin{equation}
	\CK za=\frac{2}{\pi}\proj{\eta_{z}}{\bar{\eta}_{a}(z-a)^{-1}}+O((z-a)^{-3}),\quad z\to\infty.\label{eq: G_a_asymp}
\end{equation}
Similarly, for each $\root\in\Cgr(\isG),$ we will define s-holomorphic
spinors $\Gv z{\root^{\bullet}}$ and $\Gu z{\root^{\circ}}$ on $\Cgr(\isGob)$
and $\Cgr(\isGoc)$ respectively, satisfying 
\begin{equation}
	\Gv z{\root^{\bullet}}=\proj{\eta_{z}}{\lamb(z-o^{\bullet})^{-\frac{1}{2}}}+O\left((z-o^{\bullet})^{-\frac{5}{2}}\right);
	\label{eq: asymp_zhalf_lamb}
\end{equation}
\begin{equation}
	\Gu z{\root^{\circ}}=\proj{\eta_{z}}{\lambb(z-o^{\circ})^{-\frac{1}{2}}}+O\left((z-o^{\circ})^{-\frac{5}{2}}\right).
	\label{eq: asymp_zhalf_lambb}
\end{equation}

We recall here that the s-holomorphicisty is only a real linear notion;
in particular, it is possible to show that $\CK za$ (respectively,
$\Gv z{\root^{\bullet}},$ $\Gu z{\root^{\circ}}$) are, up to a real
multiple, the only s-meromorphic function with the discrete pole only
at $a$ and vanishing at infinity (respectively, the only bounded
s-holomorphic spinor on $\Cgr(\isGob),$ on $\Cgr(\isGoc)$). Thus,
the coefficients $\bar{\eta}_{a},\lamb,\lambb$ in the above asymptotic
expansions are uniquely determined up to a \emph{real} multiple. The
fact that the choice of that multiple as in (\ref{eq: G_a_asymp})
corresponds to the value $\CK{a^{\pm}}a=\pm\eta_{a}$ can be derived
from the explicit construction below, or by soft methods as we do
in Lemma \ref{lem: kernels_symmetry}.

Let us briefly recall the aspects of the construction needed to prove Lemma \ref{lem: Cauchy_expansion}
below, see e.g. \cite[Section 5.3]{chelkak2023universality} for a
more detailed exposition. First, recall the definition of \emph{discrete
	exponentials}: for $\zeta\in \C$ and a nearest-neighbor path $x_{0}\sim x_{1}\sim\dots\sim x_{N}$
with $x_{i}$ belonging, alternatingly, to $\Cgr(\isG)$ and $\isG^{\bullet}\cup\isG^{\circ}$,
we put 
\[
\dexp{2\zeta}{x_{N}}{x_{0}}=\prod^{N}_{i=1}(1+2\zeta s_{i}\Delta_{i})^{s_{i}},
\]
where $\Delta_{i}=x_{i}-x_{i-1}$ and 
\[
s_{i}=\begin{cases}
	1, & x_{i-1}\in\Cgr(\isG)\text{ and }x_{i}\in\isG^{\bullet}\cup\isG^{\circ};\\
	-1, & x_{i-1}\in\isG^{\bullet}\cup\isG^{\circ}\text{ and }x_{i}\in\Cgr(\isG).
\end{cases}
\]
(Note that we drop the pre-factors compared to \cite[Eq. (5.2--5.3)]{chelkak2023universality}).
The product does not depend on the path $x_{1},\dots,x_{N-1}$ taken. Importantly, on $\Cgr(\isG)$, the discrete exponential $F(\cdot)=\dexp{2\zeta}{\cdot}{x_{0}}$ satisfies the three-term relation \eqref{eq: three-term-real}, albeit not s-holomorphic since the phase condition \eqref{eq: phase} fails.
Given $a\in\Cgr(\isG)$,
we define for $z\neq a$:
\begin{equation}
	\CKc za=-\frac{2}{\pi}\int_{(\bar a-\bar z)\R_{+}}\frac{\dexp{2\zeta}za}{1-\zeta^{2}e^{2i\alpha_{a}}}\,d\zeta\label{eq: Cauchy_kernel_integral}
\end{equation}
Similarly, for $\root\in\Cgr(\isG),$ we define 
\begin{equation}
\Gv z{\root^{\bullet}}=\frac{\lambb}{\sqrt{2\pi}}\int_{(\bar{\root^{\bullet}}-\bar z)\R_{+}}\frac{\dexp{2\zeta}z{\root^{\bullet}}}{\zeta^{\frac{1}{2}}}\:d\zeta=\frac{\lambb}{\sqrt{2\pi}}\int_{(\bar{\root^{\bullet}}-\bar z)\R_{+}}\frac{\dexp{2\zeta}z{\root}}{\zeta^{\frac{1}{2}}(1+\zeta e^{i\alpha_{a}})}\:d\zeta;
\label{eq: Gv_integral_dexp}
\end{equation}
\begin{equation}
\Gu z{\root^{\circ}}=\frac{\lamb}{\sqrt{2\pi}}\int_{(\bar{\root^{\circ}}-\bar z)\R_{+}}\frac{\dexp{2\zeta}z{\root^{\circ}}}{\zeta^{\frac{1}{2}}}\:d\zeta=\frac{\lamb}{\sqrt{2\pi}}\int_{(\bar{\root^{\circ}}-\bar z)\R_{+}}\frac{\dexp{2\zeta}z{\root}}{\zeta^{\frac{1}{2}}(1-\zeta e^{i\alpha_{a}})}\:d\zeta.
\label{eq: Gu_integral_dexp}
\end{equation}
With the mesh size $\delta=1$, we have $|\Delta_{i}|=\frac{1}{2}$
and $\Delta^{-1}_{i}=4\bar{\Delta}_{i}$. We record the symmetry,
for $z,a\in\Cgr(\isG):$
\begin{equation}
	\dexp{2\zeta}za=(-1)^{N/2}\Delta_{1}\Delta^{-1}_{N}\prod^{N}_{i=1}(1+2\zeta^{-1}s_{i}\bar{\Delta_{i}})^{s_{i}}=\eta^{2}_{z}\eta^{-2}_{a}\bar{\dexp{2\bar{\zeta}^{-1}}za}.\label{eq: exp_symmetry}
\end{equation}
which is obtained using that $\sum^{N}_{i=1}s_{i}=0,$ that $\Delta^{s_{k}}_{k}\Delta^{s_{k+1}}_{k+1}=1$
for all \emph{even} $k=2,\dots N-2,$ and that $N/2$ is odd iff $x_{1},x_{N-1}$
both belong to the same sublattice $\isG^{\bullet}$ or $\isG^{\circ}.$

We can apply the symmetry (\ref{eq: exp_symmetry}) to the integral
(\ref{eq: Cauchy_kernel_integral}). Let 
$l_{1}=\{\nu t\}_{0<t<1}$ and $l_{2}=\{\nu t\}_{1<t<\infty},$ where $\nu\in\C$ with $|\nu|=1$ is to be choosen later. Then, 
\begin{multline*}
	\int_{l_{2}}\frac{\dexp{2\zeta}za}{1-\zeta^{2}e^{2i\alpha_{a}}}\,d\zeta=\eta^{2}_{z}\eta^{-2}_{a}e^{-2\i\alpha_{a}}\int_{l_{2}}\frac{\bar{\dexp{2\bar{\zeta}^{-1}}za}}{e^{-2\i\alpha_{a}}-\zeta^{2}}\,d\zeta\\
	=-\eta^{2}_{z}\eta^{-2}_{a}e^{-2i\alpha_{a}}\bar{\left(\int_{l_{2}}\frac{\dexp{2\bar{\zeta}^{-1}}za}{1-\bar{\zeta}^{-2}e^{2i\alpha_{a}}}\,\frac{d\bar{\zeta}}{\bar{\zeta}^{2}}\right)}=\eta^{2}_{z}\eta^{2}_{a}\bar{\left(\int_{l_{1}}\frac{\dexp{2\zeta}za}{1-\zeta^{2}e^{2i\alpha_{a}}}\,d\zeta\right).}
\end{multline*}
provided that $\nu$ is such that $l_1$ and $l_2$ avoid the poles. Similarly, 
\[
\int_{l_{2}}\frac{\dexp{2\zeta}z{\root}}{\zeta^{\frac{1}{2}}(1\pm\zeta e^{i\alpha_{\root}})}\:d\zeta=\mp\i\eta^{2}_{z}\bar{\left(\int_{l_{1}}\frac{\dexp{2\zeta}z{\root}}{1\pm\zeta e^{i\alpha_{\root}}}\,\frac{d\zeta}{\zeta^{\frac{1}{2}}}\right)}.
\]
Denoting $I=\int_{l_{1}}\frac{\dexp{2\zeta}za}{1-\zeta^{2}e^{2i\alpha_{a}}}\,d\zeta$
and $J^{\bullet,\circ}=\int_{l_{1}}\frac{\dexp{2\zeta}z{\root}}{\zeta^{\frac{1}{2}}(1\pm\zeta e^{i\alpha_{\root}})}\:d\zeta$
with $\nu=\frac{\bar a-\bar z}{|a-z|}$ and $\nu=\frac{\bar o-\bar z}{|o-z|}$ respectively, we get
\[
\CK za=-\frac{2}{\pi}\left(\bar{\eta}_{a}I+\eta^{2}_{z}\bar{\left(\bar{\eta}_{a}I\right)}\right)=-\frac{4}{\pi}\proj{\eta_{z}}{\bar{\eta}_{a}I};
\]
\[
\Gv z{\root^{\bullet}}=\left(\frac{2}{\pi}\right)^{\frac{1}{2}}\proj{\eta_{z}}{\lamb J^{\bullet}};\quad\Gu z{\root^{\circ}}=\left(\frac{2}{\pi}\right)^{\frac{1}{2}}\proj{\eta_{z}}{\lambb J^{\circ}};
\]
in particular, the three functions $\CK za,\Gv z{\root^{\bullet}},\Gu z{\root^{\circ}}$
satisfy (\ref{eq: phase}). Because the exponentials satisfy \cite[Lemma 5.2]{chelkak2023universality} the three-term relation (\ref{eq: three-term-real}),
the spinors $\Gv z{\root^{\bullet}}$ and $\Gu z{\root^{\circ}}$
are s-holomorphic, while $\CK za$ is s-meromorphic with a poles only at
$a$. We note that $\CKc{\cdot}a$ is well defined, while $\CK{\cdot}a$
depends on the choice of sign of $\eta_{a};$ however, the former
is not s-meromorphic as its phases are off. 

For the asymptotics, we only need to analyse $I,J^{\bullet},J^{\circ},$ i.e., we need the
small and intermediate $\zeta$ behavior of the exponentials. For
small $\zeta,$ we have a bound 
\begin{equation}
	\dexp{2\zeta}{x_{N}}{x_{0}}=\exp\left(2\zeta(x_{N}-x_{0})+O\left(|\zeta|^{2}+|\zeta|^{3}|x_{N}-x_{0}|\right)\right),\label{eq: exp_close}
\end{equation}
which follows readily from the bounds $\sum^{N}_{i=1}\Delta_{i}=x_{N}-x_{0}$
and $\sum^{N}_{i=1}s_{i}\Delta^{2}_{i}=O(1),$ since $s_{i}\Delta^{2}_{i}+s_{i+1}\Delta^{2}_{i+1}=0$
whenever $s_{i}=1$. We also have a bound \cite[Proof of Theorem 2.5]{ChelkakSmirnov1},
with some $c>0:$
\begin{equation}
	\dexp{2\zeta}{x_{N}}{x_{0}}\leq\exp\left(-c|\zeta||x_{N}-x_{0}|\right),\quad|\zeta|<1.\label{eq: exp_far}
\end{equation}
The asymtotics of $I$, $J^{\bullet},J^{\circ}$ is now readily given
by Laplace's method, applying (\ref{eq: exp_close}) for $|\zeta|<|z-a|^{-\frac{1}{2}}$
and (\ref{eq: exp_far}) for $|\zeta|\geq|z-a|^{-\frac{1}{2}}$ and
changing the variable to $\mu=\zeta(a-z):$ 
\begin{multline*}
	\int_{l_{1}}\frac{\dexp{2\zeta}za}{1-\zeta^{2}e^{2i\alpha_{a}}}\,d\zeta\\
	=\frac{1}{(a-z)}\int^{|z-a|^{\frac{1}{2}}}_{0}e^{-2\mu}\left(1+O\left(\frac{\mu^{2}+\mu^{3}}{|z-a|^{2}}\right)\right)\:d\mu+O\left(\exp\left(-c|z-a|^{\frac{1}{2}}\right)\right)\\
	=\frac{1}{2(a-z)}+O(|a-z|^{-3}).
\end{multline*}
yielding (\ref{eq: G_a_asymp}). Similarly, we have 
\[
J^{\circ,\bullet}=\frac{\i}{(z-\root^{\circ,\bullet})^{\frac{1}{2}}}\int^{|z-a|^{\frac{1}{2}}}_{0}\mu^{-\frac{1}{2}}e^{-2\mu}\:d\mu+O\left(|z-\root^{\circ,\bullet}|^{-\frac{5}{2}}\right)=\frac{\sqrt{\pi}\i}{\sqrt{2}(z-\root^{\circ,\bullet})^{\frac{1}{2}}}+O\left(|z-\root^{\circ,\bullet}|^{-\frac{5}{2}}\right),
\]
yielding (\ref{eq: asymp_zhalf_lamb}--\ref{eq: asymp_zhalf_lambb}). 

We have the following key lemma:
\begin{lem}
	\label{lem: Cauchy_expansion}Let $a_{1},\dots,a_{n}\in\Cgr(\isG)$
	and $\alpha_{1},\dots,\alpha_{n}\in\R.$ Put $F(z)=\sum^{n}_{i=1}\alpha_{i}\CK z{a_{i}}.$
	Then, either $F\equiv0$, or there exists $m\in\N$ and $\alpha\neq0$
	such that 
	\[
	F(z)=\proj{\eta_{z}}{\alpha z^{-m}}+O(|z|^{-m-1}),\quad z\to\infty.
	\]
	Similarly, if $F(z)=\sum^{k}_{i=1}\alpha_{i}\Gv z{w_{i}}$ where
	$\alpha_{i}\in\R,$ and $w_{i}\in\isG^{\bullet}\cup\isG^{\circ},$
	then either $F\equiv0,$ or there exist $m\in\N-\frac{1}{2}$ and
	$\alpha\neq0$ such that 
	\[
	F(z)=\proj{\eta_{z}}{\alpha z^{-m}}+O(|z|^{-m-1}),\quad z\to\infty.
	\]
	
\end{lem}

\begin{proof}
	Let $|z|>4\max|a_{i}|.$ We rotate the contour of integration so that
	is is the same for all $G_{(a_{i})}(z),$ say $\nu\R_{+}$, ; we also
	fix an auxilliary corner $a_{0}$. We can write 
	\[
	F(z)=\proj{\eta_{z}}{\int_{l_{1}}R(\zeta)\dexp{2\zeta}z{a_{0}}\,d\zeta},
	\]
	where $R(\zeta)=-\sum^{n}_{i=1}\alpha_{i}\frac{4\bar{\eta}_{a_{i}}}{\pi}\frac{\dexp{2\zeta}{a_{0}}{a_{i}}}{1-\zeta^{2}e^{2\i\alpha_{a_{i}}}}$
	is a rational function in $\zeta$ which depends on $\alpha_{i}$
	and $a_{i}$ as parameters, but does not depend on $z$. If $R\equiv0$,
	then $F\equiv0$. Otherwise, we have $R(\zeta)=\beta\zeta^{m-1}+O(\zeta^{m})$
	as $\zeta\to0$ with some $m\in \Z_{\geq 0}$ and $\beta\neq0$. By rotating the contour,
	we may assume $\nu=\bar a_{0}-\bar z$; after the change of variable
	$\mu=\zeta(a_{0}-z),$ an application of Laplace's method as above
	gives 
	\begin{multline*}
		\int_{l_{1}}R(\zeta)\dexp{2\zeta}z{a_{0}}\,d\zeta=\frac{\beta}{(a_{0}-z)^{m}}\int^{|z-a_{0}|^{\frac{1}{2}}}_{0}\mu^{m-1}e^{-2\mu}\:d\mu+O((z-a_{0})^{-m-1})\\
		=\frac{\alpha}{z^{m}}+O((z-a_{0})^{-m-1}).
	\end{multline*}
	The proof of the second claim is similar: we will now have 
	\[
	F(z)=\proj{\eta_z}{\int_{l_{1}}\zeta^{-\frac{1}{2}}\hat{R}(\zeta)\dexp{2\zeta}z{w_1}\,d\zeta}
	\]
	for a rational funciton $\hat{R},$ and the application of Laplace's
	method gives the same result with $m\in\N-\frac{1}{2}.$
\end{proof}

\subsection{The spinor Cauchy kernel.}

In this section, we construct a spinor version of Cauchy kernel. As
with $\CKc{\cdot}a,$ the kernels will be well defined but require
multiplication by $\bar{\eta}_{a}$ to fix the phases, which in particular
entails a choice of the sign. Hereinafter we fix a ``root edge''
$\root\in\Cgr(\isG),$ which we view as the origin, and denote
\[
[\lamb z^{-\frac{1}{2}}]:=\Gv z{\root^{\bullet}};\quad[\lambb z^{-\frac{1}{2}}]:=\Gu z{\root^{\circ}};
\]

\begin{defn}
	\label{def: spinor-CK}Let  $a\in \Cgr(\isGoc)$. The spinor Cauchy kernel $\SCkx a{\lambb}(\cdot)$ is defined as the spinor on $\Cgr(\isGoc)$ satisfying
	the following properties: 
	\begin{enumerate}
		\item $\SCkx a{\alpha}$ is s-meromorphic with poles only at $a,a^\star$, and 
		\begin{equation}	
				\SCkx a{\alpha}(a^{\pm})=\pm\eta_{a};
				\label{eq: Q_a_value}
		\end{equation}
		\item We have, for every fixed $a$, $\SCkx a{\alpha}(z)=o(|z|^{\frac{1}{2}})$
		as $z\to\infty$. 
		\item For a large contour surrounding both $a$ and $\root$, 
		\begin{equation}
			\im\oint\SCkx a{\alpha}(z)[\lamb z^{-\frac{1}{2}}]\,dz=0.\label{eq: qalpha_behavior_at_infty}
		\end{equation}
	\end{enumerate}
	We now establish the existence of the spinor Cauchy kernel. For the uniqueness,
	which we will not need, see Remark \ref{rem: SCk_uniqueness};
	before it is established, we mean by $\SCkx a{\alpha}(\cdot)$ the
	object constructed in the next lemma. 
\end{defn}

\begin{lem}
	\label{lem: SCK_exists}
	The spinor Cauchy kernel $\SCkx a{\alpha}(\cdot)$ exists and can
	be written as 
	\begin{equation}
		\SCkx a{\alpha}(\cdot)=\sum^{k}_{i=1}\alpha_{i}\Gv z{w_{i}}\label{eq: Spinor_Cauchy_via_zmone}
	\end{equation}
	where $\alpha_{i}\in\R$ and $w_{i}\in\isG^{\bullet}\cup\isG^{\circ}.$
	In particular, it satisfies a stronger bound 
	\[
	\SCkx a{\alpha}(z)=O(|z|^{-\frac{1}{2}})\quad z\to\infty.
	\]
\end{lem}

\begin{proof}
	(cf.  \cite[Discussion after Lemma 14]{DubedatDimers}.) Let $\gamma:=(w_{1},\dots,w_{k+1})$ be a simple
	path on rhombi edges (i.e., corners of $\isGoc$) starting at $w_{1}=o^{\circ}$ and ending with
	the edge $(w_{k}w_{k+1})=a;$ denote $o_i:=(w_{i}w_{i+1})$. We prove the claim by induction on the
	number $k\geq 1$ of steps in the path. The induction hypothesis is that   $\SCkx{o_{1}}{\lambb}(z),\dots,\SCkx{o_{k-1}}{\lambb}(z)$ exist, which is vacuously true for the base case $k=1$. Let $F(z)=\Gv z{w_{k+1}}.$ By introducing
	a cut along $\gamma$ and splitting each $o_{i}$ into
	$o^{\pm}_{i},$ we may view $F$ as an s-meromorphic $\Cgr(\isGoc)$-spinor
	with poles at $o_{i}$, so that $F(o^{+}_{i})=-F(o^{-}_{i})$. We may find $\tilde{\alpha}_{1},\dots,\tilde{\alpha}_{k-1}$
	such that $G(z)=F(z)-\sum^{k-1}_{i=1}\tilde{\alpha}_{i}\SCkx{o_{i}}{\lambb}(z)$
	has no poles at $o_{1},\dots,o_{k-1}$, i.e., $G(o^{+}_{i})=G(o^{-}_{i})=0$
	for $1\leq i\leq k-1.$ Note that $G(o^{+}_{k})-G(o^{-}_{k})=F(o^{+}_{k})-F(o^{-}_{k})\neq0.$
	This is because for any $\root_i,$ 
	\[
	0\neq2\pi=\im(2\pi\i\lamb\lambb)=\im\oint G_{[\root_i^{\bullet}]}(z)G_{[\root_i^{\circ}]}(z)\,dz=\pm4\bar{\eta}^{2}_{\root}G_{[\root_i^{\bullet}]}(\root_i)G_{[\root_i^{\circ}]}(\root_i),
	\]
	by pulling the contour far away or shrinking to the one traversing the
	``two sides'' of $o_i$. Therefore, we can find a real multiple $c$
	so that $c\cdot G(o^{\pm}_{k})=\pm\eta_{o_{k}},$ and then add to
	$G$ a suitable multiple of $[\lambb z^{-\frac{1}{2}}]$
	to ensure (\ref{eq: qalpha_behavior_at_infty}), completing the induction
	step. 
\end{proof}

\iffalse
\begin{defn}
	We also define a version of the spinor Cauchy kernel 
	by 
	\begin{equation}
		\Que a(z)=\SCkx a{\lamb}(z)-\frac{2\bar{\eta}_{a}}{\pi}[\lambb a^{-\frac{1}{2}}][\lamb z^{-\frac{1}{2}}].\label{eq: def_que}
	\end{equation}
	Note that as a function of $z$, the two terms are now ramified at
	different points, $\root^{\bullet}$ and $\root^{\circ}$ respectively.
	We can, however, view $\Que a(z)$ an s-meromorphic spinor ramified
	at $\root^{\bullet}$ and with a pole at $\root$ (in addition to
	$a$). Note that we have $\Que a(a^{\pm})=\pm\eta_{a}.$
\end{defn}
\fi

\begin{lem}
	\label{lem: Cauchy_expansion_spinor}Let $a_{1},\dots,a_{n}$ be corners
	and $\alpha_{0},\dots,\alpha_{n}\in\R,$ 
	Put $F(z)=\alpha_0[\lambb z^{-\frac12}]+\sum^{n}_{i=1}\alpha_{i}\SCkx{a_{i}}{\lambb}(z).$ Then, either $F\equiv0$,
	or there exists $m\in\N-\frac{1}{2}$ and $\alpha\neq0$ such that
	\begin{equation}
		F(z)=\proj{\eta_{z}}{\alpha z^{-m}}+O(|z|^{-m-1}),\quad z\to\infty.\label{eq: posible_asymp_spinors}
	\end{equation}
	%In particular, we have $\Que a(z)=O(|z|^{-\frac{3}{2}})$ as $z\to\infty.$ 
\end{lem}

\begin{proof}
Follows immediately from
	Lemma \ref{lem: Cauchy_expansion} and (\ref{eq: Spinor_Cauchy_via_zmone}).
\end{proof}

\subsection{Further properties of $\protect\CK za$ and $\protect\SCkx a{\lambda}$,
	and the Cauchy integral formulae for s-holomorphic functions and spinors. }
\begin{lem}
	\label{lem: kernels_symmetry}We have, for any $a\neq b\in\Cgr(\isG)$ 
	\begin{equation}
		\CK{a^{\pm}}a=\pm\eta_{a};\label{eq: Cauchy_value_at_a+}
	\end{equation}
	\begin{equation}
		\CKc ba=-\CKc ab;\label{eq: Cauchy_symmetry}
	\end{equation}
	\begin{equation}
		\SCknoeta ab=-\SCknoeta ba;\label{eq: Cauchy_spinor_symmetry}
	\end{equation}
\iffalse
	\begin{equation}
		\Quec a(b)=-\Quec b(a)-\frac{2}{\pi}[\lamb a^{-\frac{1}{2}}][\lambb b^{-\frac{1}{2}}]-\frac{2}{\pi}[\lambb a^{-\frac{1}{2}}][\lamb b^{-\frac{1}{2}}].\label{eq: Que_symmetry}
	\end{equation}
\fi
	If $c_1$ and $c_2$ are opposite edges of the same rhombus, then
\begin{equation}
	\CK{c_{2}}{c_{1}}=0. \label{eq: G vanishing for energy}
\end{equation}

	If $a_1,a_2,a_3$ are three corners corresponging to consecutive counterclockwise edges of a rhombus, and $F_i(\cdot)=\CK{\cdot}{a_i},$ or $F_i(\cdot)=\SCkx{a_i}{\alpha}(\cdot)$ for each $i=1,2,3$, then we have  
		\begin{equation}
		F_2(\cdot)\equiv \cos\theta F_1(\cdot)+\sin\theta F_3(\cdot),
		\label{eq: CK_three_terms}
	\end{equation}
	where $\theta$ is the half-angle between the sides of the rhombus corresponding to $a_1$ and $a_2$, and the signs are chosen so that $\eta_{a_3}=\i e^{-i
	\theta}\eta_{a_2},$ $\eta_{a_1}=e^{-\i\theta}\eta_{a_2}$
\end{lem}

\begin{proof}
	Applying \eqref{eq: Discrete residue formula} to $F(z)=\proj{\eta_{z}}1$
	and $G(z)=\CK za$ and pulling the contour far away/scaling the lattice,
	we find 
	\[
	\bar{\eta}^{2}_{a}\left(1+\eta^{2}_{a}\right)\left(\CK{a^{+}}a-\CK{a^{-}}a\right)=\frac{2}{\pi}\im\oint\frac{\bar{\eta}_{a}}{z-a}\:dz=\frac{2}{\pi}\im(2\pi\i\bar{\eta}_{a})=2(\eta_{a}+\bar{\eta}_{a}),
	\]
	or $\CK{a^{+}}a-\CK{a^{-}}a=2\eta_{a}.$ On the other hand, applying
	\eqref{eq: Discrete residue formula} to $F(z)=G(z)=\CK za$ gives $\CK{a^{+}}a^{2}-\CK{a^{-}}a^{2}=0,$
	thus proving \eqref{eq: Cauchy_value_at_a+}. 
	
	For \eqref{eq: Cauchy_symmetry}, we apply \eqref{eq: Discrete residue formula}
	to $F(z)=\CK za$ and $G(z)=\CK zb$ and a counterclockwise contour
	surrounding both $a,b$. By pulling the contour far away, we see that
	the integral vanishes, so 
	\begin{multline*}
		0=\bar{\eta}^{2}_{a}\CK ab\left(\CK{a^{+}}a-\CK{a^{-}}a\right)\\
		+\bar{\eta}^{2}_{b}\CK ba\left(\CK{b^{+}}b-\CK{b^{-}}b\right)\\
		=2\bar{\eta}_{a}\CK ab+2\bar{\eta}_{b}\CK ba.
	\end{multline*}
	
	To prove \eqref{eq: Cauchy_spinor_symmetry}, we apply \eqref{eq: Discrete residue formula}
	to $F(z)=\SCkx a{\alpha}(z)$ and $G(z)=\SCkx b{\alpha}(z)$ and a counterclockwise
	contour surrounding $a,b$ and the origin. Since $\SCkx a{\alpha}(z)$
	and $\SCkx b{\alpha}(z)$ are ramified at the same point, this gives
	as above 
	\[
	\im\oint\SCkx a{\alpha}(z)\SCkx b{\alpha}(z)\,dz=2\bar{\eta}_{a}\SCkx b{\alpha}(a)+2\bar{\eta}_{b}\SCkx a{\alpha}(b).
	\]
	On the other hand, Lemma \ref{lem: Cauchy_expansion_spinor} and \eqref{eq: qalpha_behavior_at_infty} imply that $\SCkx x{\alpha}(z)=c_{x}[\lamb z^{-\frac{1}{2}}]+O(|z|^{-\frac{3}{2}})$
	for $c_{x}\in\R.$ By pulling the contour far away we can evaluate
	the integral as $c_{a}c_{b}\im\oint[\lamb z^{-\frac{1}{2}}][\lamb z^{-\frac{1}{2}}]\,dz=c_{a}c_{b}\im(2\pi\i\lamb^{2})=0,$
	which proves (\ref{eq: Cauchy_spinor_symmetry}). 
	
	The equation \eqref{eq: G vanishing for energy} follows from \eqref{eq: Cauchy_kernel_integral} by a direct computation.
	
	Finally, \eqref{eq: CK_three_terms} is a consequence of the s-holomorphicity of $\CK{\cdot}{a}, \SCkx{a}{\alpha}(\cdot),$ written in the form \eqref{eq: three-term-real}, and the symmetry relations \eqref{eq: Cauchy_symmetry}, \eqref{eq: Cauchy_spinor_symmetry}.
\end{proof}

We now record the version of Cauchy integral formulae that we will use later. 
\begin{lem}
	If $F$ is an s-holomorphic function (respectively, a $\isG_{[\root^\circ]}$-spinor), we have 
	\begin{equation}
		\label{eq: Disc_Cauchy_rec}
		4\dsbar{a}F(a)=\im \oint F(z)\CK{z}{a}\,dz,
	\end{equation}
	respectively,
	\begin{equation}
		\label{eq: Disc_Cauchy_rec_spinor}
		4\dsbar{a}F(a)=\im \oint F(z)\SCkx{a}{\lambb}(z)\,dz.
	\end{equation}
	where the contour of integration encloses $a$ (and in \eqref{eq: Disc_Cauchy_rec_spinor}, may or may not enclose $\root^\circ$). 
\end{lem}

\begin{proof}
	This is an immediate consequence of \eqref{eq: Discrete residue formula}, \eqref{eq: Cauchy_value_at_a+} and \eqref{eq: Q_a_value}.
\end{proof}

\subsection{A good set of monomials}
\label{sec: good_monomials}

A \emph{good set of monomials} is a collection of s-holomorphic and
s-meromorphic functions and spinors on a rooted rhombi
tiling $(\isG,\root)$ of the plane, where $\root\in\Cgr(\isG)$,
which serve as discrete analogs of the monomials $z^{n},$ $n\in\Z\cup\left(\Z+\frac{1}{2}\right)$ 
\begin{defn}
	A good set of monomials is a collection of functions $[z^{n}],[\i z^{n}]:\Cgr(\isG)\to\C$,
	$n\in\Z$, respectively, spinors $[z^{n}],[\i z^{n}]:\Cgr(\isGoc)\to\C$ ramified at $o^{\circ}$ for $n\in\Z+\frac{1}{2}$ satisfying
	the following properties: 
	\begin{enumerate}
		\item each $[z^{n}]$ and $[\i z^{n}]$ is s-holomorphic on $\Cgr$ for
		$n\geq0$ and s-meromorphic with poles in a finite set $\mathcal{A}_{n}$ for $n<0$;
		\item For every fixed $z$, $[z^{n}]=[\i z^{n}]=0$ for $n$ large enough,
		and moreover $[o^{n}]=[\i o^{n}]=0$ for all $n\in\N-\frac{1}{2}$;
		\item We have $[z^{n}]=\proj{\eta_{z}}{z^{n}}+O(z^{n-1})$ and $[\i z^{n}]=\proj{\eta_{z}}{\i z^{n}}+O(z^{n-1})$ as $z\to\infty$
		\item We have, for $n,m\in\Z$ or for $n,m\in\Z+\frac{1}{2}$, and for large enough counterclockwise contours,
		\begin{equation}
			\im\oint[\alpha z^{n}][\beta z^{m}]\,dz=\im(2\pi\i\alpha\beta)\delta_{n+m,-1},\quad\alpha,\beta\in\{1,\i\}.\label{eq: orthogonality}
		\end{equation}
	\end{enumerate}
\end{defn}

\begin{rem}
	We extend the above definition to $\alpha\in\C$ instead of $\alpha\in\{1,\i\}$
	by linearity: 
	\[
	[\alpha z^{m}]=(\re\alpha)[z^{m}]+(\im\alpha)[\i z^{m}].
	\]
	It is straightforward to check that (\ref{eq: orthogonality}) holds
	true in the same form, and so does the asymptotics 
	\[
	[\alpha z^{n}]=\proj{\eta_{z}}{\alpha z^{n}}+O(z^{n-1}).
	\]
\end{rem}

The proof of the existence of a good set of monomials will proceed by several lemmas. Let $\isG$ be an isoradial graph tiling the plane. We consider the set $\rhomb{\isG}$ of all \emph{rhombi} in the rhombi tiling associated to $\isG$. The boundary $\pa e$ of each rhombus $e\in \rhomb{\isG}$ consists of four sides, which we identify with the corresponding corners of $\isG$. If $E\subset \rhomb{\isG},$ we denote  $\Cgr(E):=\cup_{e\in E}\pa e$ and by $\VertSet(E)$ the set of all vertices of rhombi in $E$, which is a subset of $\isG^\circ \cup \isG^\bullet$. Recall that a \emph{train track} on a rhombi tiling of $\C$ is a (doubly infinite) sequence of rhombi adjacent to each other by parallel sides (modulo shifts and reversal of indexing). If a train track $\ttr$ is oriented, we orient any of these parallel sides from the left of $\ttr$ to the right of $\ttr$ and call it the \emph{transversal direction} of $\ttr$. We will need the following lemma:

\begin{figure}
	\centering{}\includegraphics[width=0.8\textwidth]{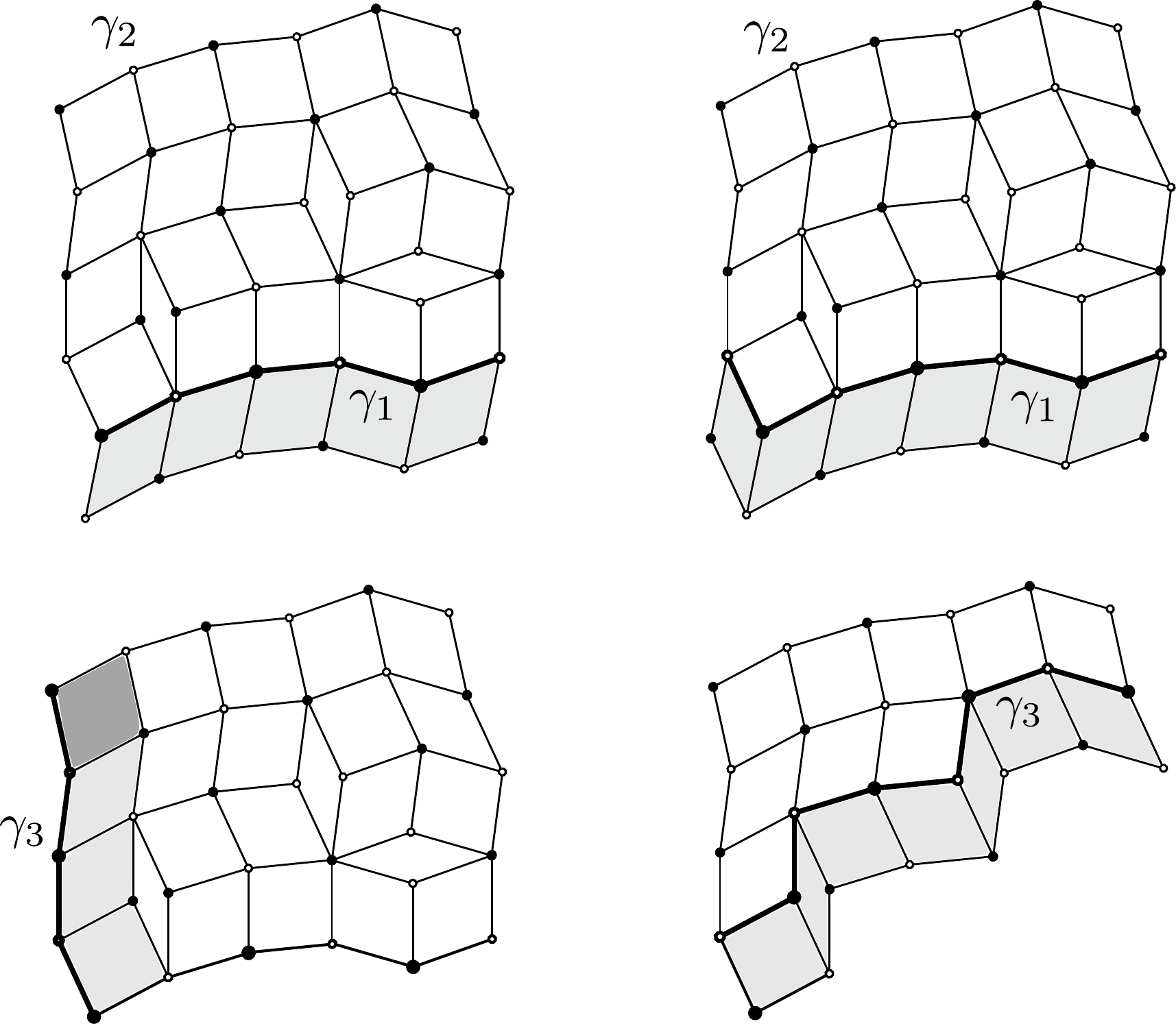}
	\caption{The two cases of the induction step in Lemma \ref{lem: clos_bdry}. Top: the domain $E$ consists of  white rhombi, the boundary arc $\gamma_1$ (bold) runs along the grey train track. Left: in the first case, $E'=\emptyset$, the train track $\ttr$ runs along $\gamma_2$ and exits $E$ at the dark grey rhombus, which has two edges on $\gamma_2$, as required. Right: in the second case, the boundary $\gamma_3$ of $\ttr$ is the new $\gamma_1$.}
	\label{fig: topol}
\end{figure}

\begin{lem}
	\label{lem: clos_bdry}
	Let $E\subset \rhomb{\isG}$ be finite and simply connected. Assume that the boundary $\pa E$ is divided into two arcs, $\gamma_1$ and $\gamma_2$ (each arc is a sequence of boundary edges of rhombi), such that $\gamma_1$ is a subset of the boundary of a train track disjoint from $E$.  Then, $\gamma_2$ contains two consecutive edges that belong to the same rhombus of $E$.
\end{lem}
\begin{proof}
	We proceed by induction on the number of rhombi in $E,$ see Fig. \ref{fig: topol}. For the base, one rhombus, it suffices to note that $\gamma_1$ can contain at most two edges.
	
	For the induction step, consider the rhombus $e$ of $E$ adjacent to the leftmost edge $c$ of $\gamma_1$ (as seen from outside of $E$), or, if $\gamma_1=\emptyset,$ pick any $e\in E$ and $c\in \pa e$.  Let $\ttr$ be the oriented train track containing $e$ in the transverse direction $c$, where the direction of $c$ is chosen to have $e$ on its left, and let $\gamma_3$ be the left boundary of $\ttr$. Notice that $\gamma_3$ cannot intersect $\gamma_1$ other than at its leftmost point, as two train tracks cannot intersect twice. Consider the part $E'$ of $E$ to the left of $\gamma_3.$ We have two possibilities: either $E'$ is non-empty, or it is empty. In the former case, it contains less rhombi than $E$, since $e\notin E'$, and we apply the induction hypothesis to $E'$ and to $\gamma_3$ in the place of $\gamma_1.$ In the latter case, $\gamma_3$ coincides with $\gamma_2$ up to the first point that $\ttr$ leaves $E$, and then the last edge of these two coinciding arcs and the edge through which $\ttr$ leaves $E$ are as desired.
\end{proof}

\begin{defn}
	A set $E$ of rhombi is called \emph{closed} if, for any $e\in \rhomb{\Gamma}$ such that at least two of the rhombi adjacent to $e$ belong to $E$, we also have that $e\in E$.  We denote by $\clos{E}$ the smallest closed set of rhombi that contains $E$.	
\end{defn}

\begin{lem}
	For every finite $E\subset \rhomb{\isG},$ $\clos{E}$ is finite, has no more connected components than $E$, and every connected component of $\clos{E}$ is simply connected.
	\label{lem: clos_properties}
\end{lem}

\begin{proof}
	It is clear that $\clos{E}$ can be constructed by adding, one by one, the rhombi that have two or more adjacent rhombi already in $E.$ Let $D$ be the set of those added rhombi $e$ that by adding $e$, we have decreased the number of connected components in the set. We see by induction that when adding rhombi, both train tracks that pass through each new rhombus also pass through a rhombus in $E\cup D$. Since every two train tracks intersect at most once, we see that $|\clos{E}||\leq 4(|E|+N-1)^2,$ where $N$ is the number of connected components in $E.$
	
	If a connected component of $\clos{E}$ contained a hole, then the inside of this hole would satisfy the conditions of Lemma \ref{lem: clos_bdry} with $\gamma_1=\emptyset,$ which is a contradiction.
\end{proof}

Given a subset $E\subset \rhomb{\isG},$ we denote:
\begin{eqnarray}
	\setG{E}&=&\spanned_\R\{\CK{\cdot}{a}:a\in \Cgr(E)\}\\
	\setQ{E}&=&\spanned_\R\{\Gv{\cdot}{a}:a\in \VertSet(E)\}.
\end{eqnarray}

The former is a subspace of the space of all functions that are s-holomorphic outside $E$, while the latter is a subset of the space of all spinors that are s-holomorphic outside $E$ and ramified at (any point of) $E$. Our motivation for introducing the notion of $\clos{E}$ is as follows:
\begin{lem}
	\label{lem: span_clos}
	For every $E$, we have $\setG{E}=\setG{\clos{E}}$ and $\setQ{E}=\setQ{\clos{E}}$.
\end{lem}
\begin{proof}
	The first assertion is immediate from \eqref{eq: CK_three_terms}.  For the second assertion, assume first that $E$ is connected and that $\root$ is one of the edges of $E$, and note that 
	\begin{equation}
		\label{eq: setQ_alternative_span}
		\setQ{E}=\spanned_\R\{\Gv{\cdot}{\root^\circ},\SCkx{a}{\lambb}(\cdot):a\in\Cgr(E)\}.
	\end{equation}
	Indeed, the $\supset$ inclusion follows readily from \eqref{eq: Spinor_Cauchy_via_zmone}, and for $\subset$, note that the roles of  $\Gv{\cdot}{w_{k+1}}$  and $\SCkx{a}{\lambb}(\cdot)$ in the proof of Lemma \ref{lem: SCK_exists} can be exchanged. The result now follows, once again, from \eqref{eq: CK_three_terms}. In the general case, apply the above reasoning to any connected component of $E$ to see that adding a rhombus with two adjacent sides in that component does not increase $\setQ{E}$, and adding a rhombus with opposite sides in $E$ does not even change $\VertSet(E)$.
\end{proof}

For the next few lemmas, we fix a sequence $\emptyset =E_0\subset E_1\subset E_2\subset \dots \subset \rhomb{\isG}$ such that $E_1=\{e_1\}$ with $\root \in \pa e_1$, and $E_{n+1}=\clos{E_n\cup e_{n+1}},$ where a rhombus $e_{n+1}$ is adjacent to $E_n$. Define $a_n$ be the sequence of corners such that $a_0=
\root$, $a_1$ is the edge of $e_1$ opposite to $a_0$, and $a_n$ is the edge of $e_n$ opposite to the one shared with $E_{n-1}$, for  $n\geq 2.$ By Lemma \ref{lem: span_clos}, each $E_n$ is finite and simply connected.

\begin{lem} \label{lem: negative_monomial_existence}
For each $n\geq 1$  
\begin{enumerate}
	\item the functions $\CK{\cdot}{a_i}$, $0\leq i\leq n,$ form a basis of $\setG{E_n}$. 
	\item the spinors $\Gv{\cdot}{\root^\circ},\Gv{\cdot}{\root^\bullet},$ and $\Gv{\cdot}{a_i^\circ},$ $1\leq i\leq n,$ form a basis of $\setQ{E_n}.$
\end{enumerate}	
\end{lem}

\begin{proof}
	For $n=1$, the first assertion follows from the fact that $\CK{\cdot}{a_0}$ and $\CK{\cdot}{a_1}$ are linearly independent by \eqref{eq: G_a_asymp}, taking into account  $\eta_{a_1}=\pm\i \eta_{a_0}$. In view of \eqref{eq: CK_three_terms} and Lemma \ref{lem: span_clos}, we have \[\setG{E_{n+1}}=\setG{E_n\cup\{e_{n+1}\}\}}=\spanned_\R\{\setG{E_{n}},\CK{\cdot}{a_{n+1}}\},\] hence, to prove the first assertion by induction in $n$, it suffices to show that $\CK{\cdot}{a_{n+1}}\notin \setG{E_n}.$ Let $\ttr$ be the oriented train track passing through $e_{n+1}$ and having $a_{n+1}$ on its right boundary. We claim that $\ttr$ separates $E_n$ from $a_{n+1}.$   Indeed, otherwise, $\ttr$ must enter $E_n$ (in one of the directions); consider the last rhombus $e\in\ttr$ before it does so. If the left boundary of $\ttr$ between $e_{n+1}$ and $e$ runs along $\pa E_n$, then $e$ has two edges adjacent to $E_n$, the one on the left boundary of $\ttr$ and the one trough which $\ttr$ enters $E_n$, which contradicts the closedness of $E_n$. Else, if the left boundary of $\ttr$ deviates from $\pa E_n$, then the set of rhombi enclosed between the two satsifies the conditions of Lemma \ref{lem: clos_bdry} with $\gamma_2\subset \pa E_n$,  which again contradicts the closedness of $E_n$.
	
	Given a linear combination $F(\cdot)=\CK{\cdot}{a_{n+1}}-\sum^n_{i=0} \alpha_i \CK{\cdot}{a_{i}},$ we can now write, as in the proof of Lemma \ref{lem: Cauchy_expansion},
	\[
	F(z)=\proj{\eta_z}{\int_{l_{1}}R(\zeta)\dexp{2\zeta}z{a_{0}}\,d\zeta}
	\] 
	with $R(\zeta)=R_{n+1}(\zeta)-\sum^{n}_{i=1}\alpha_{i}R_i(\zeta),$ and 	
	\begin{equation}
		\label{eq: def_R_i}
	R_i(\zeta)=-\frac{4\bar{\eta}_{a_{i}}}{\pi}\frac{\dexp{2\zeta}{a_{0}}{a_{i}}}{1-\zeta^{2}e^{2\i\alpha_{a_{i}}}}.
	\end{equation}
	We note that each $R_i$ is a rational function, and for any $\tau\in\C$, the pole of $R_i(\zeta)$  at $\zeta=\tau$ has order $m$, where $m$ is the number of edges in the direction of $\tau$ in a path from $a_0$ to $a_i$, minus the number of edges in the direction $-\tau$. This difference, which is independent from the path, is equal to the number of train tracks with transversal direction $\tau$ separating  $a_0$ from $a_i$.  Taking $\tau$ to be the direction of the train track $\ttr$, it follows trom the above consideration that $R_{n+1}$ has a pole at $\zeta=\tau$ of higher order than any of $R_i$ with $i\leq n$. Thus, $R\not\equiv 0$ and hence $F(\cdot)\not\equiv 0,$ for example, by computing its $z\to\infty$ asymptotics as in Lemma \ref{lem: Cauchy_expansion}.
	
	We proceed to establishing the second assertion, which is done similarly.  Recall that $\Gv{\cdot}{\root^\circ}$, $\Gv{\cdot}{\root^\bullet}$ are linearly independent  by \eqref{eq: asymp_zhalf_lamb}--\eqref{eq: asymp_zhalf_lambb}. Also, by \eqref{eq: Spinor_Cauchy_via_zmone}, we have 
	\[
	\SCkx{\hat c_n}{\lambb}(\cdot){}(\cdot)\in\spanned_{\R}\{\setQ{E_{n}},\Gv{\cdot}{a_{n+1}^\circ}, \} \quad n\geq 1
	\]
	where $\hat c_n$ is the edge of $e_{n+1}$ connecting  $a_{n+1}^\circ$ to $\VertSet(E_n)$. By \eqref{eq: setQ_alternative_span}, $\SCkx{\check c_n}{\lambb}(\cdot)\in\setQ{E_{n}}$, where $\check c_n$ is the edge of $e_{n+1}$ shared with $E_n$. By \eqref{eq: CK_three_terms} and \eqref{eq: setQ_alternative_span} and Lemma \ref{lem: span_clos}, this means that 
	\[
	\spanned_{\R}\{\setQ{E_{n}},\Gv{\cdot}{a_{n+1}^\circ}\}=\setQ{E_{n}\cup\{e_{n+1}\}}=\setQ{E_{n+1}}.
	\]
	The same holds true for $n=0,$ with $\VertSet(E_n)$ and $\setQ{E_n}$ replaced by $\{\root^\bullet,\root^\circ\}$ and $\setQ{\{\root^\bullet,\root^\circ\}}$ respectively. Therefore, it remains to establish that for each $n\geq 0$,  $\Gv{\cdot}{a_{n+1}^\circ}$ is linearly independent of $\Gv{\cdot}{\root^\circ},\Gv{\cdot}{\root^\bullet},$ and $\Gv{\cdot}{a_i^\circ},$ $1\leq i\leq n.$
	
	Given a linear combination $G(z)=\Gv{\cdot}{a_{n+1}^\circ}-\sum_{i=1}^n\alpha_i\Gv{\cdot}{a_i^\circ}-\alpha \Gv{\cdot}{\root^\circ}-\beta \Gv{\cdot}{\root^\bullet},$ write 
	\begin{equation}
		\label{eq: G_integral_formula}
		G(z)=\proj{\eta_z}{\int_{l_{1}}\zeta^{-\frac12}\hat{R}(\zeta)\dexp{2\zeta}z{\root^\circ}\,d\zeta}
	\end{equation}
	with $\hat{R}$ a linear combination of $\dexp{2\zeta}{\root^\circ}{w}$ for $w\in \{a_i^{\circ}\}_{i=0,\dots,n+1}.$ As above, $\dexp{2\zeta}{\root^\circ}{a_{n+1}}$ (which enters the linear combination with coefficient $1$) has higher order of pole at $\tau$ corresponding to the transverse direction of $\ttr$ than other terms. Therefore, $\hat{R}$ does not vanish identically and neither does $G$.
	 \end{proof}

\begin{rem}
	\label{rem: various_bases}
	It is clear from the proof that $G_{[\root^\circ]}$ and $\SCkx{a_i}{}(\cdot),$  $i=0,\dots,n,$ also forms a basis of $\setQ{E_n}$. Moreover, due to \eqref{eq: CK_three_terms}, the same results hold true if we choose $a_n$ to be \emph{any} edge of $e_n$  not shared with $\Cgr(E_{n-1})$. 
\end{rem}

\begin{lem}
	\label{lem: no_degree_jumps}
	For  each $k\in\N$ and each $\alpha\in \C,$ 
	\begin{enumerate}
		\item there exists a unique function $[\alpha z^{-k}]\in \setG{E_{2k-1}}$ such that \[[\alpha z^{-k}]=\proj{\eta_z}{\alpha z^{-k}}+O(|z|^{-k-1}).\] 
		\item there exists a unique spinor $[\alpha z^{-k-\frac12}]\in \setQ{E_{2k}}$ such that 
		\[[\alpha z^{-k-\frac12}]=\proj{\eta_z}{\alpha z^{-k-\frac12}}+O(|z|^{-k-\frac32}).\] 
	\end{enumerate}
		The functions $[ z^{-i}],[\i z^{-i}]$, $1\leq i\leq k$, form a basis of $\setG{E_{2k-1}},$ and the spinors $[ z^{-i-\frac12}],[\i z^{-i-\frac12}]$, $0\leq i\leq k$, form a basis of $\setQ{E_{2k}}.$

\end{lem}
\begin{proof}
	We proceed by induction in $k$. Thanks to Lemma \ref{lem: Cauchy_expansion} and the induction hypothesis, we can inductively find
	$\alpha_{i},\beta_{i}$ and $\hat{\alpha}_{i},\hat{\beta}_{i}$ such that 
	\begin{eqnarray*}
		\CK{z}{a_{2k}}&-&\sum^{k}_{i=1}\left(\alpha_{i}[z^{-i}]+\beta_{i}[\i z^{-i}]\right)=O(z^{-k-1}),\\
		\CK{z}{a_{2k+1}}&-&\sum^{k}_{i=1}\left(\hat{\alpha}_{i}[z^{-i}]+\hat{\beta}_{i}[\i z^{-i}]\right)=O(z^{-k-1});
	\end{eqnarray*}
	and it follows from Lemma \ref{lem: negative_monomial_existence} that these functions are linearly independent. Applying Lemma \ref{lem: Cauchy_expansion} again,  we see that either these functions have asymptotics $\proj{\eta_z}{\alpha z^{-k-1}}+O(z^{-k-2})$ and $\proj{\eta_z}{\hat{\alpha} z^{-k-1}}+O(z^{-k-2})$ with some $\alpha,\hat{\alpha}\in \C$ linearly independent over $\R$, in which case the induction step is complete, or $\setG{E_{2k+1}}$ contains a function $F(\cdot)=\sum_{i=0}^{2k+1}\gamma_i \CK{\cdot}{a_i}$ with asymptotics $F(z)=\proj{\eta_z}{\gamma z^{-m}}+O(z^{-m-1})$, where $\gamma\neq 0$ and $m\geq k+2$. Laplace's method computation as in the proof of Lemma \ref{lem: Cauchy_expansion} shows that then, the rational function $R(\zeta)=\sum_{i=0}^{2k+1}\gamma_i R_i(\zeta)$ has a zero of order at least $k+1$ at the origin. The symmetry \eqref{eq: exp_symmetry} implies that $R(\zeta) d\zeta=R(\bar{\zeta}^{-1})d(\bar{\zeta}^{-1}),$ that is, $R(\zeta)$ also has a zero of order at least $k+3$ at infinity. On the other hand, recall that the poles of $R_i$ are located at the points of the unit circle corresponding to the transversal directions of train tracks passing through rhombi of $E_{2k+1}.$ There are in total $2k+2$ such train tracks: two passing through $e_1$ and one added when adjoining each $e_i$, $2\leq i\leq 2k+1,$ since as noted in the proof of Lemma \ref{lem: clos_properties}, taking a closure does not add any new train tracks. Therefore, $R(\zeta)$ has at most $2k+2$ poles and at least $2k+4$ zeros, counted with multiplicities, which is impossible.
	
	For the second item, we similarly observe that if the assertion is not true, we can find a spinor $
	G\in\setQ{E_{2k}}$ asymptotic to $\proj{\eta_z}{\gamma z^{-m-\frac12}}$ for $m>k$. This spinor can then be written in the form \eqref{eq: G_integral_formula} with $\hat{R}=\sum_{w}\gamma_{i}\hat{R}_{w_i},$ where $\hat{R}_{w}=\dexp{2\zeta}{\root^\circ}{w},$ and the sum is over (a subset of) the vertices in $\VertSet(E_{2k}).$ The asymptotics of $G$ implies that $\hat{R}$ must have a zero of order at least $k+1$ at the origin, and the symmetry \eqref{eq: exp_symmetry} ensures that it also has a zero of order  $k+1$ at infinity. On the other hand, the total number of poles of $\hat{R},$ counted with multiplicity, is at most the number of train tracks intersecting $E_{2k},$ which is $2k+1$,  leading to a contradiction.

	For the last assertion, note that the $2k$ functions $[ z^{-i}],[\i z^{-i}]$ (respectively, $2k+2$ spinors $[ z^{-i-\frac12}],[\i z^{-i-\frac12}]$) belong to $\setG{E_{2k-1}}$ (respectively, $\setQ{E_{2k}}$) and are clearly linearly independent.
\end{proof}

\iffalse
\begin{lem}
	\label{lem: approx_CK_E_n}
For each $a\in\Cgr(\isG)$ and each $k$, there exists $H\in\setG{E_{2k-1}}$ such that
\[
\CK{z}{a}=H(z)+O(|z|^{-k-1}).
\]	
\end{lem}
\begin{proof}
Find $H(z)=\sum^{k}_{i=1}\left(\alpha_{i}[z^{-i}]+\beta_{i}[\i z^{-i}]\right)$ inductively by applying Lemma \ref{lem: Cauchy_expansion}.
\end{proof}	

\fi

\iffalse

\begin{lem}
	For every $n\geq 1$ and every $\alpha \in \C$ 
	\begin{enumerate}
		\item the (real) vector space of s-holomorphic functions $F$  on $\isG$ satisfying $|F(z)|=o(|z|^n)$ as $z\to\infty$ has dimension $2n$, and contains a function satisfying \[F(z)=\proj{\eta_z}{\alpha z^{n-1}}+O(z^{n-2});\]	
		
		\item the (real) vector space of s-holomorphic spinors $F$  on $\isG_{[\root^\circ]}$ satisfying $|F(z)|=o(|z|^{n+\frac12})$ as $z\to\infty$ has dimension $2n+1$, and contains a spinor satisfying \[F(z)=\proj{\eta_z}{\alpha z^{n-\frac12}}+O(z^{n-\frac32}).\]	 
	\end{enumerate}
	
\end{lem}

\fi
	
\begin{lem}
	\label{lem: good_mono_positive}
For  every $\alpha \in \C$ and every every $n\in \Z_{\geq 0}$ (respectively, $n\in \Z_{\geq 0}+\frac12$), there exists an  s-holomorphic function $F$ on $\isG$ (respectively, an s-holomorphic spinor $F$  on $\isG_{[\root^\circ]}$) satisfying \[F(z)=\proj{\eta_z}{\alpha z^{n}}+O(z^{n-1}).\]	
	
\end{lem}

\begin{proof}
	(cf. \cite[Lemma 10]{DubedatDimers}.) 	In both cases, we put 
	\[
	F(z)=\proj{\eta_z}{\beta \int_\gamma \zeta^{-n-1}\dexp{2\zeta}{z}{\root^\circ}\,d\zeta}, 
	\] 
with suitable $\beta$, where the contour $\gamma$ traces the ray $\overline{(z-\root^\circ)}\R_{\geq 0}\cap\{|z|\geq \eps\},$ then follows counterclockwise the circle $\{|z|=\eps\},$ then traces the same ray in the reverse direction, where $\eps\ll 1.$ Compared to the construction of $\CK{\cdot}{a}$ and $\Gv{\cdot}{\root^\circ},$ only small values of $\zeta$ contribute to the leading term of the asymptotics, giving after the rescaling
	\begin{multline*}
	\int_\gamma \zeta^{-n-1}\dexp{2\zeta}{z}{\root^\circ}\,d\zeta=
	(\root^{\circ}-z)^{n} \int_{\hat\gamma_z} \mu^{-n-1}e^{-2\mu}\,d\mu  + O(|z-\root^\circ|^{n-1})\\
=	\alpha z^n+O(|z|^{n-1}),
\end{multline*}
	where $\hat\gamma_z$ consists of the segments of the real line  between $\eps$ and $|z-\root^\circ|^\frac12$ and a small counterclockwise circle aroung the origin, and $\alpha =\lim_{w\to -n}\Gamma(w)(1-e^{-2\pi \i w})\neq 0.$ 
	
	Since the discrete exponentials satisfy the relation \eqref{eq: three-term-real}, so does the above integral, even after multiplying it by any complex number $\beta$. Projecting it onto $\eta_z$ preserves that relation and hence produces an s-holomoprhic function (respectively, spinor) with the required asymptotics.
\end{proof}

\iffalse

\begin{proof}
	For the upper bound, it suffices to show that any such function $F$ vanishing at $a_0,
	\dots, a_{2n-1}$ vanishes identically. Applying \eqref{eq: Cauchy_integral_formula} to such $F$ and $G=\CK{z}{a}$, and using Lemmas \ref{lem: approx_CK_E_n} and \ref{lem: negative_monomial_existence}, we get 
	\begin{multline}
		4\bar{\eta}_aF(a)=\im \oint F(z)\CK{z}{a} \,dz= \im \oint F(z) (H(z)+O(|z|^{-n-1}))\,dz\\
		=\sum^{2n-1}_{i=0}\im\oint \alpha_i F(z)\CK{z}{a_i}+\im \oint F(z) O(|z|^{-n-1})\,dz	\\	
		=\sum^{2n-1}_{i=0}4\bar{\eta}_{a_i} F(a_i)+\im \oint F(z) O(|z|^{-n-1})\,dz=0,
		\end{multline}
	where the last integral is shown to vanish by pulling the contour to infinity.	
	\end{proof}

\fi

\begin{thm}
	\label{thm: good_monomials}
	On any isoradial graph tiling the plane, a good set of monomials exists. 
\end{thm}

\begin{proof}
	For both $n\in\Z$ and $n\in \Z+\frac12$, we take $[\alpha z^{n}],$ $n<0$ to be as constructed in Lemma \ref{lem: negative_monomial_existence}. For $n\geq 0,$ we construct $[\alpha z^n],$ $\alpha \in \{1,\i\}$ inductively. Pick any s-holomorphic function $F$ on $\isG$ (respectively, s-holomorphic spinor $F$ on $\isG_{[\root^\circ]}$) satisfying $F(z)=\proj{\eta_z}{\alpha z^n}+O(z^{n-1});$ its existence is guaranteed by Lemma \ref{lem: good_mono_positive}. We have $\im\oint F(z)[\beta z^{m}]\,dz=\im(2\pi\i\alpha\beta)\delta_{n+m,-1}$ for $m\leq -n-1,$ by pulling the contour to infinity and using \eqref{eq: Im_int_approximation} for $m=-n-1$. Put
	\[
	[\alpha z^n]:=F(z)-\sum_{0
		\leq i\leq n-1}\left(\alpha_{i}[z^{i}]+\beta_{i}[\i z^{i}]\right)
	\] 
	where the sum is over integers if $n\in\Z$ or half-integers if $n\in\Z+\frac12 $, and $\alpha_i=(2\pi)^{-1}\im\oint F(z)[z^{-i-1}]\,dz$, $\beta_i=-(2\pi)^{-1}\im\oint F(z)[\i z^{-i-1}]\,dz$. Using the induction hypothesis, we readily get that \eqref{eq: orthogonality} holds for all $n,m$.
	
	For any $a\in \Cgr(E_{2k-1}),$ we have $\CK{\cdot}{a}\in \setG{E_{2k-1}}$ and hence, by Lemma \ref{lem: no_degree_jumps},  $\CK{\cdot}{a}=\sum^{k}_{i=1}\left(\alpha_{i}[z^{-i}]+\beta_{i}[\i z^{-i}]\right)$. Therefore,  by \eqref{eq: Disc_Cauchy_rec},
	\[
	4\bar{\eta}_a[\alpha a^n]= \im \oint[\alpha z^n]\CK{z}{a} \,dz = 0,\quad n\geq k,\,n\in\Z
	\] 
	Similarly, for any $a\in \Cgr(E_{2k}),$ we have $\SCkx{c}{\lamb}\in \setQ{E_{2k}},$ hence, Lemma \ref{lem: no_degree_jumps}, $\SCkx{c}{\lamb}=\sum^{k}_{i=0}\left(\alpha_{i}[z^{-i-\frac12}]+\beta_{i}[\i z^{-i-\frac12}]\right)$, and we arrive at  
	\[
	4\bar{\eta}_a[\alpha a^n]= \im \oint[\alpha z^n]\SCkx{c}{\lamb} \,dz = 0,\quad n\geq k+\frac12,\,n\in\Z+\frac12. 
	\]
	Choosing the sequence $E_n$ to exhaust $\rhomb{\isG}$, and taking into account that we can take $k=0$ in the above for $c=\root$, ensures the second property in the definition of a good set of monomials and thus concludes the proof.
\end{proof}

\begin{rem}
	\label{rem: general_good_sets}
The good set of monomials constructed above is uniquely defined by the sequence $E_n$, and different sequences give rise to different good sets. The above procedure gives a good set of monomials with the additional property \[\spanned\{[z^{-i}],[\i z^{-i}]:i=1,\dots,n \}=\spanned\{\CK{\cdot}{a}:a\in\Cgr(E_{2n-1})\}\] for a connected set $E_{2n-1},$ which is not true in general. However, for \emph{any} good set of monomials and any $a\in \Cgr(\isG),$ we have $\CK{\cdot}{a}\in \spanned\{[z^{-i}],[\i z^{-i}]:i=1,\dots,n \}$ (respectively, $\Que{c}(\cdot)\in \spanned\{[z^{-i}],[\i z^{-i}]:1\leq i\leq n,\,n\in\Z+\frac12 \}$) provided that $[\alpha a^m]=0$ for all $m\geq n,$ $m\in \Z$ (respectively, $m\in\Z+\frac12$.) Indeed, by Lemma \ref{lem: Cauchy_expansion} we can inductively find
$\alpha_{n},\beta_{n}$ such that 
$
	H(z)=\CK za-\sum^{n}_{i=1}\left(\alpha_{i}[z^{-i}]+\beta_{i}[\i z^{-i}]\right)
$
is either identically zero, or asymptotic to $\proj{\eta_z}{\alpha z^{-m}}$ for $m>n$. In the latter case, we have a contradiction
\[
0=4\bar{\eta}_a[\bar\alpha a^{m-1}]= \im \oint[\bar\alpha z^{m-1}]\CK{z}{a} \,dz = \im \oint[\bar\alpha z^{m-1}]H(z) \,dz = 2\pi.
\]    
A similar discussion applies to the spinor case.
\end{rem}

\begin{rem}
	The map $(F,G)\mapsto \im \oint F(z)G(z)\,dz$ defines a pairing of $\mathcal{G}^\star(E_{2n}):=\spanned_{\R}\{[z^i],[\i z^i]:0\leq i\leq n-1\}$ and $\setG{E_{2n-1}}=\spanned_{\R}\{[z^i],[\i z^i]:-n-1\leq i\leq -1\},$ which is \emph{non-degenerate}, since good monomials form dual bases. The basis $P_i(\cdot)$ of $\mathcal{G}^\star(E_{2n})$ which is dual to the basis $\CK{\cdot}{a_i}$ of $\setG{E_{2n-1}}$ is an \emph{interpolation basis}, satisfying $P_j(a_i)=\eta_{a_i}\delta_{i,j}$; in particular, this interpolation problem is always well posed. 
	
	We moreover note that the space of s-holomorphic functions $F$ satisfying $F(z)=o(|z|^{n})$ is $\mathcal{G}^\star(E_{2n})$. Indeed, for such $F$, $(F,[\alpha z^{-m}])=\im \int F(z)[\alpha z^{-m}]\,dz=0$ for all $m>n+1,$ by pulling the contour to infinity. Therefore, we can find $P\in \mathcal{G}^\star(E_{2n})$ such that $(F-P,[\alpha z^{-m}])=0$ for all $m=1,2\dots$. Then, for any $a$, we have $(F-P,\CK{\cdot}{a})=0,$ i.e., $F(a)-P(a)=0.$ A similar discussion applies to spinors.
\end{rem}

\subsection{Discrete Laurent and fractional Laurent expansion.}
The following discrete Laurent expansions will be crucial for our arguments, hence we formulate them separately.

\begin{thm}
	\label{prop: discrete_Laurent}Assume that $[z^{n}],[\i z^{n}]$ is
	a good set of monomials. Then, for every $a\in\Cgr(\isG)$, there
	exist $N$ and $\alpha_{1},\beta_{1},\dots\alpha_{N},\beta_{N}\in\R$
	such that 
	\begin{equation}
		\CK za=\sum^{N}_{n=1}\alpha_{n}[z^{-n}]+\beta_{n}[\i z^{-n}].\label{eq: abstract_Laurent}
	\end{equation}
	Similarly, for every $a\in\Cgr(\isGoc)$, there exist $N\in\N+\frac{1}{2}$
	and $\mu,\alpha_{\frac32},\beta_{\frac32},\dots\alpha_{N},\beta_{N}\in\R$ such that
	\begin{equation}
		\SCkx{a}{\lambb}(z)=\mu[\lamb z^{-\frac12}]+\sum^{N}_{n=\frac{3}{2}}\alpha_{n}[z^{-n}]+\beta_{n}[\i z^{-n}]\label{eq: abstract_Laurent_spinor}
	\end{equation}
	In both cases, the coefficients $\alpha_{i},\beta_{i}$ are
	in fact given by 
	\[
	\alpha_{n}=\frac{2}{\pi}\bar{\eta}_{a}[a^{n-1}];\quad\beta_{n}=-\frac{2}{\pi}\bar{\eta}_{a}[\i a^{n-1}],
	\]
	and in the spinor case, we have $\mu=\frac{2}{\pi}\dsbar{a}[\lambb a^{-\frac12}].$ 
\end{thm}

\begin{proof}
	The existence of the expansions follows from Lemma \ref{lem: negative_monomial_existence} (for monomials as constructed in the proof of Theorem \ref{thm: good_monomials}) and Remark \ref{rem: general_good_sets} (for general good sets of monomials). 
	
	To get the explicit formula for $\alpha_{n}$ and $\beta_{n}$, note that by
	the discrte Cauchy formula (\ref{eq: Disc_Cauchy_rec}) and by  \eqref{eq: orthogonality}, we have 
	\begin{equation*}
		4\bar{\eta}_{a}[\alpha a^{n-1}]
		=\im\int_{\gamma}[\alpha z^{n-1}]\CK za\,dz=\begin{cases}
			2\pi\alpha_{n}, & \alpha=1\\
			-2\pi \beta_{n}, & \alpha=\i.
		\end{cases}
	\end{equation*}
	Similarly, in the spinor case,  we use \eqref{eq: Disc_Cauchy_rec_spinor} with $F(z)=[\alpha z^{n-1}]$ for $n=\frac{3}{2},\frac{5}{2},\dots$ to find $\alpha_n$ and $\beta_n$, and with  $F(z)=[\lambb z^{-\frac12}]$ to recover $\mu$.
\end{proof}

We now state a corollary that we will not use, but that is of independent interest.

\begin{cor}
	A finite discrete Laurent expansion (\ref{eq: abstract_Laurent})
	holds with $\CK za$ replaced with any function $F$ which is s-holomorphic
	outside a finite set and vanishes at $\infty$. Similarly, an expansion $G(z)=\sum_{n=1}^{N}\alpha_i[z^{\frac12-n}]+\beta_n[\i z^{\frac12-n}]$ holds for any spinor $G$ on $\Cgr(\isGoc)$ which is s-holomorphic outside a finite set and satisfies the bound $G(z)=o(z^\frac12)$ at infinity.
\end{cor}

\begin{proof}
	Fix $w$ and represent $F(w)$ by Cauchy integral formula \eqref{eq: Disc_Cauchy_rec}, where the integration is over the union of a clockwise contour $\gamma_{\mathrm{in}}$ which has $w$ on its outside and all discrete poles inside, and a counterclockwise countour  $\gamma_{\mathrm{out}}$ which encloses $w$ and $\gamma_{\mathrm{in}}$. The contribution of $\gamma_{\mathrm{out}}$ is readily seen to vanish by pulling it to infinity, and using \eqref{eq: def-Im-Int} and the symmetry \eqref{eq: Cauchy_symmetry}, we can rewrite the sum over $\gamma_{\mathrm{in}}$ as 
		\[F(w)=\frac{1}{2}\sum_{c\in\gamma_{\mathrm{in}}}s_{\gamma}(c)\bar{\eta}_{c}F(c)\CK wc.
	\]
	Applying \eqref{eq: abstract_Laurent} to each term yields the result.
	
For the spinor case, we can similarly apply \eqref{eq: Disc_Cauchy_rec_spinor}, but we replace $\SCkx{w}{\lambb}(z)$ with $\Que{w}(z):=\SCkx{w}{\lambb}(z)-\frac{2}{\pi}\dsbar{w}[\lambb w^{-\frac12}][\lamb z^{-\frac12}],$ which we are free to do as the last term is s-holomorphic in $z$. Since we have, by the previous Lemma,  $\Que{w}(z)=O(z^{-\frac32})$ at infinity, we again see that the contribution of $\gamma_{\mathrm{out}}$ vanishes. The contribution of $\gamma_{\mathrm{in}}$ is handled as above using \eqref{eq: def-Im-Int}  and \eqref{eq: Cauchy_spinor_symmetry}, and expanding each $\SCkx{c}{\lambb}(z)$ in the sum by \eqref{eq: abstract_Laurent_spinor} yields the result.
\end{proof}

\begin{rem}
	\label{rem: SCk_uniqueness} If the spinor $G$ in the above argument is s-holomorophic everywhere, then $\gamma_{\mathrm{in}}$ can be taken to be empty, which shows that $G$ must be proportional to $[\lambb z^{-\frac12}].$ This shows the
	uniqueness of $\SCkx a{\lambb}$ as in Definition \ref{def: spinor-CK}, by taking $G$ to be the difference of two candidates. 
\end{rem}

\section{Fock spaces, Virasoro modules, and the Sugawara construction.}
\label{sec: background_algebra}
\if{false}
\begin{itemize}
 \item the Lie algebra $\vir$
 \item its universal enveloping algebra $\UEA{\vir}$
 \item the quotient $\UvirHalf$
 \item the two-chiral version $\TwoVir$
 \item Neveu-Schwarz Clifford algebra $\CliffNS$
 \item Ramond Clifford algebra $\CliffR$
 \item the two-chiral versions $\TwoCliffNS$, $\TwoCliffR$
 \item highest weight modules for $\vir$ (at $c=\cIsing$)
 \item modules for $\TwoVir$ (just tensor products of two)
 \item fermionic Fock space for $\CliffNS$ (irreducible, unitary, universal) and its two-chiral version
 \item (no need to define twisted fermionic Fock space for $\CliffR$?)
 \item twisted two-chiral fermionic Fock space for $\TwoCliffR$ (\ldots)
 \item \ldots
\end{itemize}
\fi

For analyzing the structure of local fields, we need to employ three (associative) algebras
as well as certain modules over them,
and in fact it is natural to consider both chiral and two-chiral versions in each case.
The first of the algebras is the universal enveloping algebra of the Virasoro algebra.
The other two are suitable Clifford algebras, the Neveu-Schwarz type being
relevant to spin-flip-even local fields and the Ramond type to spin-flip-odd local fields.
It is standard and clearest to describe the Virasoro algebra and its
modules in one chirality first, and then
construct the two-chiral versions from these. For the Clifford algebras, we instead describe the
two-chiral algebra and modules directly, partly to avoid the use of graded tensor products,
but also because in the Ramond sector the relevant module
for the two-chiral algebra is not simply obtained by tensoring modules in the separate
chiralities.

\subsection{Virasoro algebra and its modules}

The \term{Virasoro algebra} $\vir$ is the infinite-dimensional Lie algebra with basis consisting of
$\Cgen$ and $\Lgen{n}$ for $n \in \Z$, and with Lie brackets determined by
\begin{align}\label{eq:virasoro_bracket}
[\Lgen{n}, \Lgen{m}] \; = \; & (n-m) \, \Lgen{n+m} + \deltaZ{n+m} \frac{n^3-n}{12} \, \Cgen , &
[\Cgen, \Lgen{n}] \; = \; & 0 ,
    & \text{for $n, m \in \Z$.}
\end{align}
For our purposes, the most convenient associative algebra (over $\C$) related to the Virasoro algebra
is the quotient of the universal enveloping algebra $\UEA{\vir}$ by the relation $\Cgen = \cIsing$
(more precisely by the two-sided ideal generated by $\Cgen - \cIsing \, \Igen$ with $\Igen$ denoting
the algebra unit), which we denote by
\begin{align}\label{eq:virasoro_associative}
\UvirHalf \; := \; \UEA{\vir} \Big/ \big\langle \Cgen - \cIsing \, \Igen \big\rangle .
\end{align}
The full CFT has copies of the Virasoro algebra for both chiralities, so we
define
\begin{align}\label{eq:full_virasoro}
\TwoVir \; := \; \UvirHalf \tens \UvirHalfbar ,
\end{align}
where $\UvirHalfbar$ is a copy of $\UvirHalf$ with generators denoted by
$\Cbar$ and $\Lbar{n}$ for $n \in \Z$, and the algebra structure on~\eqref{eq:full_virasoro}
is such that the two tensor factors commute.

The building blocks of minimal models of CFT are irreducible highest weight modules
of the Virasoro algebra. These, in turn, are constructed as quotients of Verma modules.
%by their maximal proper submodules.
We recall the relevant definitions and constructions %of these
and a key classification of unitary highest weight modules.

\begin{defn}%[Singular vector]
Let $\Module$ be a $\UEA{\vir}$-module. A nonzero vector $v_0 \in \Module \setminus \set{0}$
is called a \term{singular vector} with central charge $c \in \C$ and conformal weight $h \in \C$
if it satisfies $\Cgen \, v_0 = c \, v_0$ and $\Lgen{0} \, v_0 = h \, v_0$ and
$\Lgen{n} \, v_0 = 0$ for $n > 0$.
\end{defn}

\begin{defn}%[Highest weight module]
A $\UEA{\vir}$-module $\Module$ is called a \term{highest weight module}
with central charge $c \in \C$ and highest weight $h \in \C$
if there exists a singular vector $v_0 \in \Module \setminus \set{0}$
of central charge $c \in \C$ and conformal weight $h \in \C$ such that
$\Module$ is generated by $v_0$ in the sense that $\UEA{\vir} \, v_0 = \Module$.
% A $\UEA{\vir}$-module $\Module$ is called a \term{highest weight module}
% with central charge $c \in \C$ and highest weight $h \in \C$
% if there exists a nonzero vector $v_0 \in \Module \setminus \set{0}$
% which generates $\Module$ in the sense that $\UEA{\vir} \, v_0 = \Module$ and
% which satisfies
% $\Cgen \, v_0 = c \, v_0$ and $\Lgen{0} \, v_0 = h \, v_0$ and
% $\Lgen{n} \, v_0 = 0$ for $n > 0$.
Such a vector $v_0$ is called a \term{highest weight vector} in $\Module$
and it is unique up to nonzero scalar multiples.
\end{defn}

\begin{defn}%[Verma module]
% A \term{Verma module}  with central charge $c \in \C$ and highest weight $h \in \C$
% is a universal (initial) object in the category of
% highest weight modules with these $c$ and $h$ (unique up to isomorphism).
% As an explicit construction, we may
The \term{Verma module} with central charge $c \in \C$ and highest weight $h \in \C$
is the $\UEA{\vir}$-module
\begin{align}\label{eq:verma}
\Verma{c}{h} \; := \;
    \UEA{\vir} \big/ \mathcal{J}_{c,h}
        %\big\langle \Cgen - c \, \Igen, \, \Lgen{0} - h \, \Igen, \, \Lgen{n} : n > 0 \big\rangle
\end{align}
obtained as the quotient of the universal enveloping algebra $\UEA{\vir}$ (seen as a left module
for itself) by its submodule $\mathcal{J}_{c,h}$ generated by the elements
$\Cgen - c \, \Igen$, $\Lgen{0} - h \, \Igen$, and $\Lgen{n}$ for $n > 0$.
The Verma module $\Verma{c}{h}$ is a highest weight module
with central charge $c \in \C$ and highest weight $h \in \C$, and with
$\Igen + \mathcal{J}_{c,h} \in \Verma{c}{h}$ as its highest weight vector.
\end{defn}

By construction, the Verma module has the following universal property, which can
equivalently be seen to characterize the Verma module as an induced $\UEA{\vir}$-module
from a $1$-dimensional module over the subalgebra generated by $\Cgen$ and $\Lgen{n}$
with $n \ge 0$.
\begin{lem}%[Universal property of the Verma module]
Let $\Module$ be a $\UEA{\vir}$-module with a singular vector $v_0 \in \Module$
of central charge $c \in \C$ and conformal weight $h \in \C$.
Then there exists a unique $\UEA{\vir}$-module morphism
$\mphism \colon \Verma{c}{h} \to \Module$ such that $\mphism(\Igen + \mathcal{J}_{c,h}) = v_0$.
In particular, if $\Module$ is a highest weight module with
central charge $c \in \C$ and highest weight $h \in \C$ and $v_0$ is
its highest weight vector, then this morphism $\mphism$
is surjective and $\Module \isom \Verma{c}{h} / \ker(\mphism)$.
\end{lem}

By the last observation, any highest weight module is a quotient of a Verma module.
One can show that among submodules of a Verma module, there is a unique maximal proper one.
The quotient by this is the unique irreducible highest weight module with the given
central charge and highest weight.
% then the unique irreducible highest weight module
% with the given central charge and highest weight.
% %The irreducible highest weight modules are therefore obtained as follows.
\begin{defn}%[Irreducible highest weight module]
Let $\maxsubmod{c}{h} \subsetneq \Verma{c}{h}$ be the maximal proper submodule
of the Verma module $\Verma{c}{h}$. The quotient
\begin{align}\label{eq:irreducible_hwmod}
\Irrep{c}{h} \; := \; \Verma{c}{h} \big/ \maxsubmod{c}{h}
\end{align}
is the \term{irreducible highest weight module} with central charge $c \in \C$ and highest
weight $h \in \C$.
\end{defn}

The CFT for the Ising model is one of a discrete series of \emph{unitary} CFTs, and the
unitarity of the corresponding modules simplifies the analysis of the Virasoro modules of
local fields substantially. We note that unitarity could in principle be avoided in the analysis,
and such alternative arguments will in any case be needed for some nonunitary
CFTs such as symplectic fermions \cite{adame2025discrete}.
But as the resulting simplifications are substantial, we embrace unitarity as a part of our
proof strategy.

We take the convention that (complex) inner products $\iprod{\cdot}{\cdot}$ are linear in their second
argument and conjugate-linear in their first argument.

\begin{defn}%[Unitary Virasoro module]
A $\UEA{\vir}$-module $\Module$ is said to be \term{unitary} if
there exists an inner product $\iprod{\cdot}{\cdot}$ on $\Module$ such that for all $n \in \Z$
the adjoint of $\Lgen{n}$ is $\Lgen{n}^{\dagger} = \Lgen{-n}$ and
the adjoint of $\Cgen$ is $\Cgen^{\dagger} = \Cgen$.
% , and if in addition
% the $\Lgen{0}$-eigenspaces are finite-dimensional and the
% $\Lgen{0}$-eigenvalues are bounded from below.
\end{defn}

In a unitary $\UEA{\vir}$-module $\Module$, the elements $\Lgen{0}$ and $\Cgen$
are self-adjoint, %(with respect to the given inner product)
so their eigenvalues are necessarily real. In particular
in a unitary highest weight representation, the central charge $c$ and the
highest weight $h$ must be real.
Another basic observation is that in a unitary module, any submodule has
a complementary submodule (the orthogonal complement to a submodule is a submodule).
Combined with the standard fact that submodules respect grading by $\Loper0$-eigenvalues and
the fact that the highest weight vector generates a highest weight module and
spans the one-dimensional eigenspace of the minimal $\Loper0$-eigenvalue, it
follows that a unitary highest weight representation is necessarily irreducible.

One can fully classify which (irreducible) highest weight modules are
unitary \cite{friedan1986details, langlands1988unitary, gabriel2024lattice}.
This will drastically simplify our analysis of the structure of the
space of local fields for the Ising model. We state the classification result
in the special case that we will use.\footnote{We moreover point out that we only use the
nonunitarity implication of the classification; the unitarity of the
irreducible highest weight modules at the three particular highest weight values is
obtained from our analysis of the Clifford modules via the Sugawara construction.
In particular, our proof does not rely on the \cite{goddard1986unitary} coset construction, which in the
general case is needed to supplement the \cite{friedan1986details} result.}
%The classification result that simplifies our analysis drastically is the following.
\begin{thm}\cite{friedan1986details, langlands1988unitary, gabriel2024lattice} %\cite{FQS}
\label{thm:fqs_nonunitarity}
At central charge $c = \cIsing$, the only unitary highest weight Virasoro modules
(up to isomorphism) are $\IrrepI{h} := \Irrep{\cIsing}{h}$ with $h \in \set{\hI, \hS, \hE}$.
\end{thm}

\subsection{The Clifford algebras and the Sugawara construction}

In this subsection, we recall how the commutation relations \eqref{eq: anticomm} give rise to representations of the Virasoro algebra with central charge $c=\frac12$. It suffices to do it in one chirality only. Moreover, we keep our treatment elementary  and avoid the discussion of Lie superalgebras. The construction is the same as in e.g. \cite{hongler2022conformal}, however, we re-state the lemma again here to explicitly include the unitarity part.
%In what follows, we omit the symbol $\circ$ when writing the composition of linear operators.

\iffalse
In both Sugawara constructions, we use the
following notation for normal ordered pairs of operators: given operators
$\Foper{a} \colon V \to V$ on a vector space $V$, indexed by either
half-integers $a \in \Z + \half$ (in the Neveu-Schwarz case) or integers $a \in \Z$
(in the Ramond case),
%indices $a,b$ are either
we denote
\begin{align}\label{eq:normal_ordering}
\nord{\Foper{a} \Foper{b}} \; = \begin{cases}
    \phantom{-} \Foper{a} \circ \Foper{b} & \text{ if $a \le b$} \\
            -   \Foper{b} \circ \Foper{a} & \text{ if $b < a$} .
    \end{cases}
\end{align}

\fi

\begin{lem}%[Fermionic Sugawara construction in the Neveu-Schwarz sector]
\label{lem:fermionic_sugawara_ns}
Let $V$ be a (complex) vector space, and assume that $\Foper{n} \colon V \to V$ is a familiy of linear
operators indexed by $n \in \Z + \half$ (Neveu-Schwarz case), respectively $n\in \Z$ (Ramond case), satisfying the following properties: 
\begin{itemize}
\item for any $v \in V$ there exists a constant $N_v$ such that $\Foper{n} \, v = 0$ for all $n > N_v$;
\item for any $n,m$ we have
$\Foper{n}  \Foper{m} + \Foper{m} \Foper{n} \, = \, \FockNorm \deltaZ{m+n} \, \idOf{V}$.
\end{itemize}
Then, in the Neveu-Schwarz case, the formulae
\begin{equation}
\label{eq:fermionic_sugawara_ns}
\Loper{n} v  := 
    \begin{cases}\SugawaraNorm \sum_{k \in \Z + \half} k  \nord{\Foper{n-k} \Foper{k}} v,&n\neq 0, \\
    \half\sum_{k \in \Z_
    	{\geq 0} + \half} k  \nord{\Foper{-k} \Foper{k}} v, & n=0,
    \end{cases}
\end{equation}
respectively, in the Ramond case, the formulae
\begin{equation}
	\label{eq:fermionic_sugawara_r}
	\Loper{n} v :=  
	\begin{cases}\frac14\sum_{k\in \Z} k\Curr{n-k}\Curr{k}v,& n\neq 0,\\ 
		\frac{1}{16} v+ \frac12 \sum_{k
			\in \Z_{\geq 0}}k\Curr{-k}\Curr{k}v,&n=0,\end{cases}
\end{equation}
give rise to well-defined linear operators $\Loper{n} \colon V \to V$, $n \in \Z$, which
satisfy
\begin{align*}
\Loper{n} \Loper{m} - \Loper{m} \Loper{n}
    \; = \; (n-m) \, \Loper{n+m} + \deltaZ{n+m} \frac{n^3 - n}{24} \, \idOf{V} .
\end{align*}
In other words, these fermionic Sugawara operators equip $V$ with a representation of
the Virasoro algebra with central charge $c = \half$ (i.e., $\Cgen$ acting by $\half \, \idOf{V}$).
If, moreover, $V$ has an inner product $\iprod{\cdot}{\cdot}$ such that
$\Foper{k}^\dagger = \Foper{-k}$ for all $k$, then for the operators
defined by~\eqref{eq:fermionic_sugawara_ns} we have
$\Loper{n}^\dagger = \Loper{-n}$, i.e., the representation of the Virasoro algebra is
unitary.
\end{lem}

\begin{proof}
For each $v$, the right-hand sides of \eqref{eq:fermionic_sugawara_ns}, \eqref{eq:fermionic_sugawara_r} are well defined since the sums only contain finitely many non-zero terms. Indeed, the terms with $k>N_v$ vanish and so do those with $n-k>N_v$ for $n\neq 0$, since $\Curr{n-k}\Curr{k}=-\Curr{k}\Curr{n-k}$ in this case. The commutation relations for $\Loper{n}$ follow by a straightforward and standard computation, see e.g. \cite{hongler2022conformal}. To see that \eqref{eq:fermionic_sugawara_ns}, \eqref{eq:fermionic_sugawara_r} are equivalent to expressions given in \cite{hongler2022conformal}, note that for $n\in \Z$, $n\neq 0$, we have 
\begin{equation}
	\label{eq: Sugawara_shift}
	\sum_k \Curr{n-k}\Curr{k}\;\stackrel{k':=n-k}{=}\;\sum_{k'} \Curr{k'}\Curr{n-k'}=-\sum_{k'} \Curr{n-k'}\Curr{k'}=0;
\end{equation}
hereinafter the sums are over $\Z+\frac12$ in the Neveu-Schwarz case and over $\Z$ in the Ramond case, and are well defined since they are finite when acting on any $v\in V$. In particular, in the $n\neq 0$ case, the prefactor of $k$ can be replaced by  $k+\frac12$, recovering the formulae in  \cite[Remark 4.19]{hongler2022conformal}. 
To prove the claim about the adjoints, we note that $(\Curr{n-k}\Curr{k})^\dagger=\Curr{k}^\dagger\Curr{n-k}^\dagger=\Curr{-k}\Curr{-n+k}.$ This immediately gives $\Loper{0}^\dagger=\Loper{0}$, while for $n\neq 0$, we have  
\[
\Loper{n}^\dagger=\sum_k k\Curr{-k}\Curr{-n+k}\;\stackrel{k':=-n+k}{=}\;\sum_{k'} (k'+n)\Curr{-n-k'}\Curr{k'}=\sum_{k'} k'\Curr{-n-k'}\Curr{k'}=\Loper{-n}.
\]
\end{proof}

\subsection{Fermionic Fock spaces and their structure}
Recall the definition of Clifford algebras $\TwoCliffNS$ and $\TwoCliffR$ given in the introduction.

\begin{defn}%[Fermionic Fock space]
% A \term{Verma module}  with central charge $c \in \C$ and highest weight $h \in \C$
% is a universal (initial) object in the category of
% highest weight modules with these $c$ and $h$ (unique up to isomorphism).
% As an explicit construction, we may
The \term{Neveu-Schwartz}, respectively, \term{Ramond} (two-chiral) \term{fermionic Fock spaces} are defined as the left $\TwoCliffNS$-module (respectively left $\TwoCliffR$-module)
\begin{align}\label{eq:chiral_fock}
\FullFockNS \; := \;
    \TwoCliffNS \big/ \mathcal{J}^{\NeveuSchwarz} \\
 \FullFockR \; := \;
 \TwoCliffR \big/ \mathcal{J}^{\Ramond}
\end{align}
where $\TwoCliffNS$ (respectively, $\TwoCliffR$) is seen as a left module
over itself, and
\begin{itemize}  \item $\mathcal{J}^{\NeveuSchwarz}$ is the submodule of $\TwoCliffNS$ generated by the elements $\Fgen{k}, \Fbar{k}$ for $k \in \set{\frac{1}{2}, \frac{3}{2}, \frac{5}{2}, \ldots}$,
	\item $\mathcal{J}^{\Ramond}$ is the submodule of $\TwoCliffR$ generated by the elements $\Fgen{k}, \Fbar{k}$ for $k=1,2\dots,$, and  the element $\lamb \Fgen{0} - \lambb \Fbar{0}$.
\end{itemize}
%called the \term{vacuum vector} of $\FockNS$.
We introduce the notation $\vacNS := \Igen + \mathcal{J}^{\NeveuSchwarz}$ for the \emph{vacuum vector} of $\FullFockNS$, and we denote by $\bspin=\Igen+\mathcal{J}^{\Ramond}$ and $\bmu=\lamb\Fgen{0}\bspin=\lambb\Fbar{0}\bspin$  for what may be called the \emph{vacuum vectors} of $\FullFockR$.
\end{defn}

\iffalse
By construction, the Fock space has a universal property.
We can state it in a more useful form after quick remarks on unitarity and
irreducibility.
We start, however, by the fermionic Sugawara construction which equips suitable
$\CliffNS$-modules (including the Fock space and some of our spaces of local fields)
with the additional structure of $\UvirHalf$-modules.

\fi

\iffalse
\emph{[TODO: Fermionic Sugawara construction in the NS sector.]}
\fi

\begin{lem}
	\label{lem: Fock_basis}
The modules $\FullFockNS$ and $\FullFockR$ are irreducible. The vectors
\begin{equation}
	\Psi_{n_{1}}\dots\Psi_{n_{k}}\bar{\Psi}_{m_{1}}\dots\bar{\Psi}_{m_{\bar k}}\ind:\quad k,\bar k\geq0,n_{1}<\dots<n_{k}<0,m_{1}<\dots<m_{\bar k}<0,\label{eq: basis_module_NS}
\end{equation}
the indices belonging to $\Z+\frac{1}{2},$ form a basis of $\FullFockNS$. The vectors
\begin{equation}
	\Psi_{n_{1}}\dots\Psi_{n_{k}}\bar{\Psi}_{m_{1}}\dots\bar{\Psi}_{m_{\bar k}}\sigma,
	k,\bar k\geq0,n_{1}<\dots<n_{k}\leq 0,m_{1}<\dots<m_{\bar k}<0,\label{eq: basis_module_R}
\end{equation}
indices belonging to $\Z$, form a basis of $\FullFockR$. There exists a unique inner product $\iprod{\cdot}{\cdot}$ on $\FullFockNS$ (respectively, on $\FullFockR$), such that $\iprod{
\ind}{\ind}=1$ (repsectivley, $\iprod{\bspin}{\bspin}=\iprod{\bmu}{\bmu}=1$ and $\iprod{\bspin}{\bmu}=0$) and such that for all $k \in \Z + \half$ (respectively, $k
\in\Z$) the adjoints of $\Fgen{k}$ and $\Fbar{k}$ are given by $\Fgen{k}^{\dagger} = \Fgen{-k}$ and $\Fbar{k}^{\dagger}=\Fbar{-k}$
\end{lem}
It is natural to call a Clifford module supporting such an inner product \emph{unitary}. 
\begin{proof}
%The fact that \eqref{eq: basis_module_NS} and \eqref{eq: basis_module_R} are bases of $\FullFockNS$ and $\FullFockR$ respectively can be established purely algebraically by a PBW - type argument, but we prefer the following evelentary one. 
It is clear that any element $\Psi_{n_{1}}\dots\Psi_{n_{k}}\bar{\Psi}_{m_{1}}\dots\bar{\Psi}_{m_{\bar k}}\ind$ with $n_i,m_i
\in\Z+\half$, (respectively, $\Psi_{n_{1}}\dots\Psi_{n_{k}}\bar{\Psi}_{m_{1}}\dots\bar{\Psi}_{m_{\bar k}}\sigma$ with $n_i,m_i
\in\Z$) can be written as a linear combination of  \eqref{eq: basis_module_NS} (respectively, \eqref{eq: basis_module_R}), by moving all the $\Fgen{n_i}$ and $\Fbar{m_i}$ with $n_i,m_i>0$ to the right using the anti-commutation relations \eqref{eq: anticomm}.
Moreover, a straightforward computation shows that, for $n_i,m_i$ and $n_i',m_i'$ % \in \Z+\frac{1}{2}
as in \eqref{eq: basis_module_NS}, we have
\begin{multline}
	\label{eq: anticomm_all}
\Fbar{-m_{\bar{k}'}}\dots\Fbar{-m'_1}\Fgen{-n'_{k'}}\dots \Fgen{-n'_1}\Fgen{n_{1}}\dots\Fgen{n_{k}}\Fbar{m_{1}}\dots\Fbar{m_{\bar k}}\ind \\ = \begin{cases} 2^{k+\bar{k}}\ind,& k=k', \bar{k}=\bar{k}', n_i=n_i',m_i=m_i'\text{ for all } i\\
	0, & \text{otherwise}.
	\end{cases}
\end{multline}
Similarly, for $n_i,m_i$ and $n_i',m_i'$ as in \eqref{eq: basis_module_R}, further assuming $n_k, n'_{k'} \ne 0$
to isolate the zero modes $\Fgen{0}$, we have
\begin{multline}
	\label{eq: anticomm_all_R}
\Fbar{-m_{\bar{k}'}}\dots\Fbar{-m'_1}\Fgen{-n'_{k'}}\dots \Fgen{-n'_1}\Fgen{0}^{b'}\Fgen{0}^{b}\Fgen{n_{1}}\dots\Fgen{n_{k}}\Fbar{m_{1}}\dots\Fbar{m_{\bar k}}\sigma \\ = \begin{cases}
	2^{k+\bar{k}}\sigma,& b+b' \text{ even}, k=k', \bar{k}=\bar{k}', n_i=n_i',m_i=m_i'\text{ for all } i\\
	\bar{\lambda} (-1)^{k+\bar{k}} 2^{k+\bar{k}}\mu,& b+b' \text{ odd}, k=k', \bar{k}=\bar{k}', n_i=n_i',m_i=m_i'\text{ for all } i\\
	0, & \text{otherwise}.
	\end{cases}
\end{multline}
% Moreover, a straightforward computation shows that, for $n_i,m_i$ and $n_i',m_i'$ as in \eqref{eq: basis_module_NS} and $\Op=\ind$  (respectively, \eqref{eq: basis_module_R}, and $\Op=\sigma$), we have
% \begin{multline}
% 	\label{eq: anticomm_all}
% \Fbar{-m_{\bar{k}'}}\dots\Fbar{-m'_1}\Fgen{-n'_{k'}}\dots \Fgen{-n'_1}\Fgen{n_{1}}\dots\Fgen{n_{k}}\Fbar{m_{1}}\dots\Fbar{m_{\bar k}}\Op \\ = \begin{cases} 2^{k+\bar{k}}\Op,& k=k', \bar{k}=\bar{k}', n_i=n_i',m_i=m_i'\text{ for all } i\\
% 	0, & \text{else if} \sum |n'_i|\geq \sum |n_i|\text{ or } \sum |m'_i|\geq \sum |m_i|.
% 	\end{cases}
% \end{multline}
Therefore, given a linear combination of elements \eqref{eq: basis_module_NS} (respectively \eqref{eq: basis_module_R}) with not all coefficients zero, one can multiply it on the left by a suitable string $$\Fbar{-m_{\bar{k}'}}\dots\Fbar{-m'_1}\Fgen{-n'_{k'}}\dots \Fgen{-n'_1}$$ to obtain a non-zero multiple of $\ind$ (respectively, $\bspin$). This easily shows that any nonzero submodule must contain $\ind$ (resp. $\sigma$), and irreducibility of $\FullFockNS$ (resp. $\FullFockR$) follows. The assertion that the collections  \eqref{eq: basis_module_NS} and  \eqref{eq: basis_module_R} are linearly independent now follows once we know that $\ind\neq 0$ and $\sigma \neq 0$, which is easy to see\footnote{it is also immediate from Corollary \ref{cor: Virasoro_structure} below, since $1$ and $\sigma_{\root^\circ}$ are obviously not null fields.}.

The same computations \eqref{eq: anticomm_all} and \eqref{eq: anticomm_all_R} shows that the inner product $\iprod{\cdot}{\cdot}$ under which the basis vectors \eqref{eq: basis_module_NS} (respectively \eqref{eq: basis_module_R}) are orthogonal with norms $2^{\frac{k+\bar{k}}{2}}$ is as required, and is uniquely defined by these properties.
\end{proof}

If follows from Lemma \ref{lem:fermionic_sugawara_ns} that $\FullFockNS$ and $\FullFockR$ are also \emph{unitary} $\TwoVir$ modules. The unitarity allows us to deduce their structure. We make a simple preliminary observation. Denote, for $\Op\in\{\ind,\sigma\}$, 
\[
\FockSym_\Op^{(n,m)}=\spanned\{\Fgen{n_{1}}\dots\Fgen{n_{k}}\Fbar{m_{1}}\dots\Fbar{m_{\bar k}}\Op:n_1+\dots+n_k=-n,m_1+\dots+m_{\bar{k}}=-m\}.\]
\begin{lem}
	\label{lem: grading}
	The direct sum decompositions 
	\[
	\FullFockNS=\bigoplus_{(n,m)\in\left(\frac12\Z_{\geq 0}\right)^2} \FockNSG{n}{m},\quad \FullFockR=\bigoplus_{(n,m)\in\left(\Z_{\geq 0}\right)^2} \FockRG{n}{m},
	\]	
give gradings of $\FullFockNS$ and $\FullFockR$ both as Cliffford and 	$\TwoVir$ modules, i.e., acting by $\Fgen{n}$ or $\Loper{n}$ (respectively, $\Fbar{n}$ or $\Loperbar{n}$) decreases the first (respectively, the second) grading index by $n$. The vectors in the component with grade $(n,m)$ of $\FullFockNS$ are eigenvectors of $\Loper{0}$ and $\Loperbar{0}$, with eigenvalues $n$ and $m$ respectively, while the vectors in the component of $\FullFockR$ with grades $(n,m)$ are eigenvectors of $\Loper{0}$ and $\Loperbar{0}$ with eigenvalues $\frac{1}{16}+n$ and $\frac{1}{16}+m$, respectively.
\end{lem}
\begin{proof}
The first assertion of the lemma is clear from the anti-commutation relations \eqref{eq: anticomm} and the fact that the expansion of $\Loper{n}$ in the Sugawara construction contains only monomials of the form $\Fgen{n-k}\Fgen{k}$. The second assertion is a straightforward computation.
\end{proof}

Denote $\psi:=\Fgen{-\half}\ind$, $\psi^\star:=\Fbar{-\half}\ind$ and $\ben:=\i\Fgen{-\half}\Fbar{-\half}\ind.$ Given $\Op\in\FullFockNS$ or $\Op\in\FullFockR$, we denote by $\HVM{\Op}$ the $\TwoVir$-module generated by $\Op$.

\begin{prop}
	\label{prop: Virasoro_structure}
	As a $\TwoVir$ modules, we have the following
	direct sum decompositions: 
	\[
	\FullFockNS=\HVM{\ind}\oplus\HVM{\bpsi}\oplus\HVM{\bpsibar}\oplus\HVM{\ben},
	\]
	\[
	\FullFockR=\HVM{\bspin}\oplus\HVM{\bmu}.
	\]
	where the direct summands are irreducible highest weight modules described as follows:
	\[
	\HVM{\ind} \isom \Irrep{\half}{0}\otimes\Irrepbar{\half}{0},\quad \HVM{\bpsi} \isom\Irrep{\half}{\half}\otimes\Irrepbar{\half}{0},\quad \HVM{\bpsibar} \isom\Irrep{\half}{0}\otimes\Irrepbar{\half}{\half}\quad
	\HVM{\ben} \isom\Irrep{\half}{\half}\otimes\Irrepbar{\half}{\half}.
	\]
	\[
	\HVM{\bspin} \isom \HVM{\bmu}\isom \Irrep{\half}{\frac{1}{16}}\otimes\Irrepbar{\half}{\frac{1}{16}}
	\]
\end{prop}

\begin{proof}
	It follows from Lemma \ref{lem: grading} that each of $
	\ind,\bpsi,\bpsibar,\ben,\bspin,\bmu$ is an eigenvector of both $\Loper{0}$ and $\Loperbar{0},$ with eigenvalues given as in the table in Section \ref{subsec: intro_Fock}. Moreover, those are the only joint eigenvectors with these eigenvalue pairs, by explicit count of dimensions in the corresponding gradings using bases \eqref{eq: basis_module_NS}, \eqref{eq: basis_module_R}. Applying $\Loper{n}$ or $\Loperbar{n}$ with $n\geq 1$ to any of them makes one of the grading components negative, and the only vectors in either $\FullFockNS$ or $\FullFockR$ with a negative grade component is $0$. Hence, each of the vectors in the above list is highest weight for both of the commuting actions of $\UvirHalf $. By unitarity, the modules generated by them are irreducible, and in particular, their sum is a direct sum.
	
	A standard Vandermonde determinant argument implies that the modules $\mathcal{M}=\HVM{\ind}\oplus\HVM{\bpsi}\oplus\HVM{\bpsibar}\oplus\HVM{\ben}$ and $\mathcal{M}'=\HVM{\bspin}\oplus\HVM{\bmu}$ respect the grading of Lemma \ref{lem: grading}, i.e., they are vector space direct sums of their intersections with the subspaces $\FockNSG{n}{m}$ and $\FockRG{n}{m}.$ By unitarity again, $\mathcal M$ and $\mathcal M'$ have complementary submodules in $\FullFockNS$ and $\FullFockR$ respectively, which are simply the vector space direct sums of $\FockNSG{n}{m}\cap \mathcal M^\bot$ and $\FockRG{n}{m}\cap \mathcal M'^\bot$. Let $\Op$ be a vector in one of these orthogonal complements with the smallest value of $(n+m);$ then the $\TwoVir$ submodule generated by $\Op$ is highest weight for both actions of $\UvirHalf$ and also unitary, hence the classification of Theorem \ref{thm:fqs_nonunitarity} shows that the $\Loper{0}$ eigenvalue of $\Op$ must belong to $\{0,\frac12,\frac{1}{16}\},$ and similarly for $\Loperbar{0}$. In the $\FullFockNS$ case, this forces the eigenvalue pair of $\Op$ to be one of $(0,0),$ $(0,\half)$, $(\half,0)$ and $(\half,\half)$, and hence by the above discussion $\Op$ is proportional to one of $\ind,\bpsi,\bpsibar,\ben,$  which is absurd. Similarly, in $\FullFockR$ case, the eigenvalue pair must be $(\frac{1}{16},\frac{1}{16})$, but the joint eigenspace with this eigenvalue pair is spanned by $\sigma$ and $\mu$. We conclude that the compements of $\mathcal M$ and $\mathcal M'$ are zero.
	
	It is clear that a unitary $\TwoVir$-module that it highest weight for both actions of $\UvirHalf$ is simply a tensor product of the corresponding $\UvirHalf$ modules, which concludes the proof.
\end{proof}

\iffalse
\begin{lem}\label{lem:full_fock_virasoro_structure}
As a $\TwoVir$-module, we have
\begin{align*}
\FullFockNS \; \isom \;
    \big( \IrrepSym(\hI) \tens \overline{\IrrepSym}(\hI) \big) \oplus
    \big( \IrrepSym(\hI) \tens \overline{\IrrepSym}(\hE) \big) \oplus
    \big( \IrrepSym(\hE) \tens \overline{\IrrepSym}(\hI) \big) \oplus
    \big( \IrrepSym(\hE) \tens \overline{\IrrepSym}(\hE) \big) .
\end{align*}
\end{lem}

By contrast, the relevant module in the Ramond sector cannot be constructed
simply by tensoring modules in each chirality.
Specifically, while there exists a natural
``twisted'' Fock space $\CliffR$-module (and this module appears as a $\CliffR$-submodule
of our space of local fields), the zero-modes of the two chiralities
interact in such a way that the local fields do not contain the tensor product of
two copies of this for~$\TwoCliffR$. We next describe the construction of the
relevant $\TwoCliffR$-module directly.
\fi

\iffalse
\emph{[TODO: Construct the very-twisted-full-fermionic-Fock-space-in-the-Ramond-sector
and get its universal property, irreducibility, and structure as a $\TwoVir$-module.]}
\fi

\section{Ising local fields, fermionic current action, and Virasoro action}

\subsection{Local fields and local fields with a fermion insertion.}

\label{subsec:Local-fields} All the notions in this section will
depend on a rooted rhombi tiling $(\isG,\root)$ of the plane of mesh
size $\delta$. Recall the definitions of the spaces $\LF(\isG),\LFPe(\isG),$
and $\LFPo(\isG)$ of local spin and Grassmann polynomials on $\isG,$
and of $R_{\Op}$ and $\Op_{\delta,z},$ given in Section \ref{subsec: local_lattice_fields_intro}.

Given $\isG^{(1)},\dots,\Gamma^{(k)}$ of mesh size $1$ and an isoradial
graph $\Od$ with an Ising model on it (with some boundary conditions),
let us recall how the critical Ising correlations 
\begin{equation}
	\E_{\Od}\left(\Op^{(1)}_{\delta,z_{1}}\dots\Op^{(k)}_{\delta,z_{k}}\right)\label{eq: def_cor_function}
\end{equation}
are defined, where $\Op^{(i)}\in\LF(\isG^{(i)})\oplus\LFPe(\isG^{(i)})\oplus\LFPo(\isG^{(i)}),$
and $\delta R_{\Op^{(i)}}$-balls of $\Od$ at $z_{i}$ agree with
$\isG^{(i)}$ and don't intersect each other or the boundary of $\Od.$
First, using linearity, we assume that each $\Op^{(i)}$ is a monomial.
If the total Grassmann degree (i.e., the total number of $\psi$)
of these monomials is odd, we define the correlation to be zero. Else,
we define it to be the \emph{spin-fermion correlation} as in \cite[Section 2]{chelkak2021correlations},
by substituting for the formal variables the spins and the fermion variables in the Ising model with
(Baxter's) Z-invariant critical weights \cite{baxter1985exactly, ChelkakSmirnov2} with the 
chosen boundary conditions.  Recall that the Ising Hamiltonian with Baxter's weights is \[\mathcal{H}(\sigma)=-\sum_{x\sim y}\beta_{(xy)}\sigma_x\sigma_y,\quad \beta_{(xy)}=-\frac12\log\tan\frac{\theta_{(xy)}}{2},\] where $\theta_{(xy)}$ is half the rhombus angle against the edge $(xy)$. 

More concretely, to define the correlation
$\E[\psi_{w_{1}}\dots\psi_{w_{2N}}\sigma_{v_{1}}\dots\sigma_{v_{n}}]$,
we choose an arbitrary subset $\gamma$ of the edges of the dual graph
$\isG^{\bullet}$ such that, when $\gamma$ is viewed as $1$-chain
over $\Z/(2\Z),$ we have $\partial\gamma=w^{\bullet}_{1}+\dots+w^{\bullet}_{2N}\mod2$.
We put 
\[
\mu_{\gamma}=e^{-2\beta\sum_{(vv')\cap\gamma\neq\emptyset}\sigma_{v}\sigma_{v'}},
\]
and 
\[
\E\left[\psi_{w_{1}}\dots\psi_{w_{2N}}\sigma_{v_{1}}\dots\sigma_{v_{n}}\right]:=\eta_{w_{1}}\dots\eta_{w_{2N}}\E\left[\sigma_{w^{\circ}_{1}}\dots\sigma_{w^{\circ}_{2N}}\sigma_{v_{1}}\dots\sigma_{v_{n}}\mu_{\gamma}\right].
\]
It is well known, and explained in \cite[Section 2]{chelkak2021correlations},
that only the sign of the above correlation depends on $\gamma$, and that
only on the homology class of $\gamma$ in $\Od\setminus\{v_{1},\dots,v_{n}\}.$
Moreover, to fix the sign of $\E\left[\psi_{w_{1}}\dots\psi_{w_{2N}}\sigma_{v_{1}}\dots\sigma_{v_{n}}\right]$
unambiguously, it is enough to choose lifts of $w_{1},\dots,w_{2N}$
onto the double cover $\Od_{[v_{1},\dots,v_{n}]}$ of $\Od$ ramified
at $v_{1},\dots,v_{n};$ changing the sheet of one lift flips the
sign, see the convention at \cite[Definition 2.23]{chelkak2021correlations} or discussion around Definition \ref{def: insertion} below.
Since originally, the fermion variables are defined on corners of
$\isG^{(i)}$ or $\isG^{(i)}_{[\root^{\circ}]}$, for each $i$, in
order to unambiguously define the correlation (\ref{eq: def_cor_function}),
it is enough to choose one of the two embeddings of the $R_{\Op^{(i)}}$-neighborhood
of either $\isG^{(i)}$ or $\isG^{(i)}_{[\root^{\circ}]}$ into the
double cover of $\Od$ ramified at $z^{\circ}_{j}$ for all $j$ such
that $\Op^{(j)}\in\LFo\cup\LFPo$.

We follow \cite{hongler2022conformal} to define the concept of a
\emph{null field, }extending it naturally also to Grassmann polynomials: 
\begin{defn}
	A local field $\Op\in(\LF\oplus\LFPe\oplus\LFPo)(\isG)$ is called
	a \emph{null field }if there exists an $R>0$, such that, whenever
	$\Od$ is a finite isoradial graphs with any boundary conditions,
	$\Op^{(i)}\in(\LF\oplus\LFPe\oplus\LFPo)(\isG^{(i)}),$ $\delta R_{\Op^{(i)}}$-balls
	of $\Od$ at $z_{i}$ agree with $\isG^{(i)}$ and don't intersect
	each other or the boundary of $\Od$ or $\delta R$-ball at $z,$
	and the latter agrees with $\isG$, we have 
	\[
	\E_{\Od}\left[\Op_{\delta,z}\Op^{(1)}_{\delta,z_{1}}\dots\Op^{(k)}_{\delta,z_{k}}\right]=0.
	\]
	The space of all null fields is denoted by $\Null.$ We will use the
	notation $\Op_{1}\equiv\Op_{2}$ if $\Op_{1}-\Op_{2}\in\Null.$ The class of a field $\Op$ modulo null fields will be denoted by $[\Op],$ so that $\Op_{1}\equiv\Op_{2}$ if and only if $[\Op_{1}]=[\Op_2]$.
\end{defn}

\begin{rem}
	We believe that one could impose stronger requirements, e.g., take
	$R=R_{\Op}$ in the above definition, and that will lead to the same
	notion of null fields. We will not pursue this direction here. 
\end{rem}

We proceed to explaining the isomorphisms (\ref{eq: spin-ferm-isomorphism}).
Consider a monomial $\Op=\psi_{z_{1}}\dots\psi_{z_{2n}}\in\LFPee$
or $\sigma_{\root^{\circ}}\Op=\sigma_{\root^{\circ}}\psi_{z_{1}}\dots\psi_{z_{2n}}\in\LFPoe$.
We may fix a subset $\gamma_{1}$ of edges of the dual graph $\isG^{\bullet}$
such that $\pa\gamma_{1}=z^{\bullet}_{1}+\dots+z^{\bullet}_{2n}$
as chains over $\Z/(2\Z)$. Denote $\hat{\Op}=\eta_{z_{1}}\dots\eta_{z_{2n}}\sigma_{z^{\circ}_{1}}\dots\sigma_{z^{\circ}_{2n}}\mu_{\gamma_{1}}.$
Then, by definition, $\hat{\Op}\in\LFe$ (respectively, $\sigma_{\root^{\circ}}\hat{\Op}\in\LFo$)
has the same correlations with any distant fields as $\Op$ (respectively,
$\sigma_{\root^{\circ}}\Op$), i.e., $\hat{\Op}-\Op\in\Null$ (respectively,
$\sigma_{\root^{\circ}}\hat{\Op}-\sigma_{\root^{\circ}}\Op\in\Null$).
Conversely, given a spin monomial $\sigma_{v_{1}}\dots\sigma_{v_{n}}$,
we can always choose $N$ and $z_{1},\dots,z_{2N}$ in such a way
that $z^{\bullet}_{2i-1}=z^{\bullet}_{2i}$ for all $i=1,\dots,N$,
and $\prod^{2N}_{i=1}\sigma_{z^{\circ}_{i}}=\prod^{n}_{i=1}\sigma_{v_{i}}$
(if $n$ is even) or $\sigma_{\root^{\circ}}\prod^{2N}_{i=1}\sigma_{z^{\circ}_{i}}=\prod^{n}_{i=1}\sigma_{v_{i}}$
(if $n$ is odd). In this case, when inserting $\Op=\psi_{z_{1}}\dots\psi_{z_{2N}}$
(respectively, $\sigma_{\root^{\circ}}\Op=\sigma_{\root^{\circ}}\psi_{z_{1}}\dots\psi_{z_{2N}}$)
into a correlation, we may choose $\gamma_{1}$ as above to be empty,
which shows that $\Op$ (respectively, $\sigma_{\root^{\circ}}\Op$)
has the same correlation as $\eta_{z_{1}}\dots\eta_{z_{2N}}\prod^{n}_{i=1}\sigma_{v_{i}}$
when inserted into correlations with distant fields, for a suitable
choice square roots in the definition of $\eta_{z_{i}}'s$. Together,
these observations can be summarized as (\ref{eq: spin-ferm-isomorphism}).

Next, we define the \emph{fermion insertion maps} that add one fermion
to a Grassmann polynomial $\Op$ (possibly with spin insertion at
$\sigma_{\root^{\circ}}$) at a given corner $z$. These maps are
unambiguously defined when $\psi_{z}$ does not appear in any monomial;
else, one needs to choose the ``side'' of $z$ on the modified graph,
$z_{+}$ or $z_{-}$, where one performs the insertion. 
\begin{defn}
	\label{def: insertion}
	\label{def: ferm_insertion}For a corner $z\in\Cgr(\isG),$ the \emph{fermion insertion map} is a linear map $\LFPe(\isG)\to\LFPe(\isG)$ defined on monomials
	by
	\[
	\psi_{z_{1}}\dots\psi_{z_{n}}\mapsto\psi_{z_{\pm}}\psi_{z_{1}}\dots\psi_{z_{n}}=\begin{cases}
		\psi_{z}\psi_{z_{1}}\dots\psi_{z_{n}}, & z\notin\{z_{i}\}^{n}_{i=1};\\
		\pm(-1)^{i-1}\ds{z_i}^2\psi_{z_{1}}\dots\psi_{z_{i-1}}\psi_{z_{i+1}}\dots\psi_{z_{n}}, & z=z_{i}.
	\end{cases}
	\]
	Similarly, for a corner $z\in\Cgr(\isG_{[\root^{\circ}]}),$
	the fermion insertion map $\LFPo(\isG)\to\LFPo(\isG)$ is defined on monomials by 
	\[
	\psi_{z_{1}}\dots\psi_{z_{n}}\sigma_{\root^\circ}\mapsto\psi_{z_{\pm}}\psi_{z_{1}}\dots\psi_{z_{n}}\sigma_{\root^\circ}=\begin{cases}
		\psi_{z}\psi_{z_{1}}\dots\psi_{z_{n}}\sigma_{\root^{\circ}}, & z\notin\{z_{i},z^{\star}_{i}\}^{n}_{i=1};\\
		\pm(-1)^{i-1}\ds{z_i}^2\psi_{z_{1}}\dots\psi_{z_{i-1}}\psi_{z_{i+1}}\dots\psi_{z_{n}}\sigma_{\root^{\circ}}, & z=z_{i},\\
		\mp(-1)^{i-1}\ds{z_i}^2\psi_{z_{1}}\dots\psi_{z_{i-1}}\psi_{z_{i+1}}\dots\psi_{z_{n}}\sigma_{\root^{\circ}}, & z=z^{\star}_{i}.
	\end{cases}
	\]
\end{defn}

It is clear that these maps exchange $\LFPee$ with $\LFPeo$ (respectively,
$\LFPoe$ with $\LFPoo$)
\begin{defn}
	A field $\Op\in(\LFPe\oplus\LFPo)(\isG)$ is called \emph{real }if
	all coefficients in front of any monomial $\psi_{z_{1}}\dots\psi_{z_{n}}$
	or $\sigma_{\root^{\circ}}\psi_{z_{1}}\dots\psi_{z_{n}}$ in the expansion
	of $\Op$ belongs to $\bar{\eta}_{z_{1}}\dots\bar{\eta}_{z_{n}}\R$
	
	The reader will readily check that all correlations of real fields
	are real, and $\Op\in\LFPee\oplus\LFPoe$ is real if and only if the
	corresponding class in $\LF/\mathcal{N}$ can be represented by a
	spin polynomial with real coefficients.
\end{defn}

Let us give two related, important examples of a null field. First,
 $\psi_{z}$ is \emph{s-holomorphic}, i.e., for
any three corners $z_{1},z_{2},z_{3}$ adjacent to the same edge (so
that they correspond to consecutive sides of a rhombus following in counterclockwise direction) and any $\Op$, the field 
\[
\left(\cos\theta e^{\i\theta}\psi_{z_{1}}-\i \sin\theta e^{\i\theta}\psi_{z_{3}}-\psi_{z_{2}}\right)\Op
\]
is a null filed, where $\theta$ is half the angle between the sides of the rhombus corresponding to $z_1,z_2$. We stress that this is true even when the field $\Op$ contains fermionic insertion
	at some of $z_{1},z_{2},z_{3}$, provided that one performs the insertions according to Definition \ref{def: ferm_insertion}
	``on the side of the rhombus'' in question, i.e., at $z_i^{+}$ if the rhombus is on the left of $[z_i^\bullet,z_i^{\circ}]$ and at $z_i^{-}$ else. We call this property the s-holomorphicity of $\psi_{z}$ since it implies the s-holomorphicity, as a function of $z$, of $\E[\psi_{z}\Op]$ for any \emph{real} field $\Op$. 

Second, if $\Op$ is a real field, $G$ is an s-holomoprhic function
or spinor, and $\gamma_{1}$ and $\gamma_{2}$ are two contours on
$\isG^{\bullet}\cup\isG^{\circ}$ that are homotopic to each other
in the complement of the support of $\Op$ and ramification points
of $G$, then 
\[
\im\int_{\gamma_{1}}\psi_{z}G(z)\Op\,dz-\im\int_{\gamma_{2}}\psi_{z}G(z)\Op\,dz
\]
is a null field. This is a consequence of the fact that, when inserted
into a correlation with real fields, $\psi_{z}\Op$ becomes an s-holomorphic
function. 

\iffalse
\begin{rem}
	In the next subsection, we don't really need to care about the lifts
	since $z$ and $w$ are far from $\Op$ and we can just identify the
	double cover with $\Cgr$ itself in that region (in other words, we
	let $\gamma$ stay far away from $\Op$).
\end{rem}

\fi

\subsection{Fermionic current action}

\label{subsec:Fermionic-current-action} For the next three subsections, we fix an underlying isoradial graph $(\Gamma,\root)$ of mesh size $1$ tiling the plane; we will drop it from notation.

For a \emph{real} field $\Op\in\LFPe$ and $n\in\Z+\frac{1}{2}$, respectively, $\Op\in\LFPo$ and
$n\in\Z$, and $\alpha\in\C,$ we define 
\begin{equation}
	\RCurr n{\alpha}\Op=\im\oint_{\Op}\psi_{z}[\alpha z^{n-\frac{1}{2}}]\Op\,dz,\label{eq: def_Psi_alpha}
\end{equation}
where the integral is over a counterclockwise closed loop encircling
the $R_{\Op}$ ball around the origin. As noted at the end of the
previous section, any choice of contour leads to the same result up
to a null field; ditto for adding a null field to $\Op$. Thus, $\Psi^{[\alpha]}_{n}$
are well defined as linear maps 
\[
\Psi^{[\alpha]}_{n}:\LFPe/\Null\to\LFPe/\Null,\quad n\in\frac{1}{2}+\Z;\quad\Psi^{[\alpha]}_{n}:\LFPo/\Null\to\LFPo/\Null,\quad n\in\Z.
\]
For a complexification, given a \emph{real} field $\Op$, we simply define $\Psi^{[\alpha]}_{n}\i\Op=\Psi^{[\i\alpha]}_{n}\Op;$
note that this identity is then true for all fields $\Op.$ We then
put
\begin{equation}
	\Curr n\Op:=\frac{1}{2\sqrt{2\pi}}\left(\RCurr n1-\i\RCurr n{\i}\right)\Op;\quad\Currbar n\Op:=\frac{1}{2\sqrt{2\pi}}\left(\RCurr n1+\i\RCurr n{\i}\right)\bar{\Op}.\label{eq: def_Psi_Psibar}
\end{equation}
Note that $\Curr n's$, thus defined, are \emph{complex linear}: indeed, we have 
\begin{gather*}
	\Curr n(\i\Op)=\frac{1}{2\sqrt{2\pi}}\left(\RCurr n{\i}+\i\RCurr n1\right)\Op=\i\Curr n\Op,\\
	\Currbar n(\i\Op)=\frac{1}{2\sqrt{2\pi}}\left(\RCurr n{\i}-\i\RCurr n1\right)(-\bar{\Op})=\i\Currbar n\Op.
\end{gather*}
Let us also record the inverse to \eqref{eq: def_Psi_Psibar} given by 
\begin{equation}
	\RCurr{n}{\alpha}\Op=\sqrt{2\pi}(\alpha\Curr{n}\Op+\bar\alpha \Currbar{n}\bar{\Op}),\quad \alpha \in \C.
	\label{eq: Curr_to_Rcurr}
\end{equation}

\begin{prop}
	\label{prop: anticomm_discrete}
	The operators $\Curr{n},\Currbar{n}$ satisfy the anti-commutation relations \eqref{eq: anticomm}.
\end{prop}

\begin{proof}
	Let us first compute the anti-commutation relations for $\RCurr{n}{\alpha}$: fix three large contours
	$\gamma_{1}\text{ inside }\gamma_{2}\text{ inside }\gamma_{3}$, so that $\gamma_1$ encircles $R_\Op$, then we have
	\begin{multline*}
		\left(\RCurr n{\alpha}\RCurr m{\beta}+\RCurr m{\beta}\RCurr n{\alpha}\right)\Op\stackrel{(1)}{\equiv}\im\oint_{\gamma_{3}}\psi_{z}[\alpha z^{n-\frac{1}{2}}]\im\oint_{\gamma_{2}}\psi_{w}[\beta w^{m-\frac{1}{2}}]\Op\,dw\,dz\\
		-\im\oint_{\gamma_{1}}\psi_{z}[\alpha z^{n-\frac{1}{2}}]\im\oint_{\gamma_{2}}\psi_{w}[\beta w^{m-\frac{1}{2}}]\Op\,dw\,dz\\
		\stackrel{(2)}{\equiv}\sum_{w\in\gamma_{2}}\left(\im\oint_{\gamma_{3}}-\im\oint_{\gamma_{1}}\right)\left(\psi_{z}[\alpha z^{n-\frac{1}{2}}]s_{\gamma_{2}}(w)2\bar{\eta}^{2}_{w}[\beta w^{m-\frac{1}{2}}]\psi_{w}\Op\right)\,dz\\
		\stackrel{(3)}{\equiv}\sum_{w\in\gamma_{2}}4s_{\gamma_{2}}(w)\bar{\eta}^{3}_{w}[\alpha w^{n-\frac{1}{2}}][\beta w^{m-\frac{1}{2}}](\psi_{w^{+}}\bar{\eta}_{w}\psi_{w}\Op-\psi_{w^{-}}\bar{\eta}_{w}\psi_{w}\Op)\\
		\stackrel{(4)}{\equiv}4\sum_{w\in\gamma_{2}}2s_{\gamma_{2}}(w)\bar{\eta}^{2}_{w}[\alpha w^{n-\frac{1}{2}}][\beta w^{m-\frac{1}{2}}]\Op\stackrel{(5)}{\equiv} 4\im\oint_{\gamma_{2}}[\alpha w^{n-\frac{1}{2}}][\beta w^{m-\frac{1}{2}}]\Op\,dw\\
		\stackrel{(6)}{\equiv}4\im(2\pi\i\alpha\beta)\delta_{n+m,0}\Op
		%=\delta_{n+m,0}\Op\cdot\begin{cases}
			%	8\pi, & \alpha=\beta=1;\\
			%	-8\pi, & \alpha=\beta=\i;\\
			%	0, & \alpha\beta\notin\R.
			%\end{cases}
		\end{multline*}
		Here, in (1),  we exchanged the order of summation and used the anti-symmetry of fermionic correlations; the (2) and (5) are by the definition \eqref{eq: def-Im-Int} of the discrete integration, (3) follows from \eqref{eq: Discrete residue formula}, (4) follows from Definition \ref{def: insertion}, and (6) is by \eqref{eq: orthogonality}. 
		
		The anti-commutation relations for $\Curr{n}$ and $\Currbar{m}$ are now straightforward to check using~\eqref{eq: def_Psi_Psibar}: writing $\{A;B\}:=AB+BA,$ we have 
		\begin{multline*}
			8\pi\{\Curr n;\Curr m\}=\left(\RCurr n1-\i\RCurr n{\i}\right)\left(\RCurr m1-\i\RCurr m{\i}\right)+\left(\RCurr m1-\i\RCurr m{\i}\right)\left(\RCurr n1-\i\RCurr n{\i}\right)\\
			=\{\RCurr n1;\RCurr m1\}-\{\RCurr n{\i};\RCurr m{\i}\}-\i\{\RCurr n{\i};\RCurr m1\}-\i\{\RCurr n1;\RCurr m{\i}\}=16\pi\delta_{n+m,0}.
		\end{multline*}
		\begin{multline*}
			8\pi\{\Currbar n;\Currbar m\}=(\RCurr n1+\i\RCurr n{\i})\bar{(\RCurr m1+\i\RCurr m{\i})}+(\RCurr m1+\i\RCurr m{\i})\bar{(\RCurr n1+\i\RCurr n{\i})}\\
			=\{\RCurr n1;\RCurr m1\}-\{\RCurr n{\i};\RCurr m{\i}\}+\i\{\RCurr n{\i};\RCurr m1\}+\i\{\RCurr n1;\RCurr m{\i}\}=16\pi\delta_{n+m,0}.
		\end{multline*}
		
		\begin{multline*}
			8\pi\{\Curr n;\Currbar m\}=(\RCurr n1-\i\RCurr n{\i})(\RCurr m1+\i\RCurr m{\i})+(\RCurr m1+\i\RCurr m{\i})\bar{(\RCurr n1-\i\RCurr n{\i})}\\
			=\{\RCurr n1;\RCurr m1\}+\{\RCurr n{\i};\RCurr m{\i}\}+\i\{\RCurr n{\i};\RCurr m1\}+\i\{\RCurr n1;\RCurr m{\i}\}=0.
		\end{multline*}
\end{proof}

In other words, \eqref{eq: def_Psi_alpha}, \eqref{eq: def_Psi_Psibar} define the action of $\CliffNS$ and $\CliffR$ on the spaces $\LFPe/\Null$ and $\LFPo/\Null$, respectively.
\subsection{Structure of the Clifford modules generated by $1$ and
	$\sigma$.}

\label{sec: lattice_Fock_structure}
We consider the Clifford modules  $\ClInd$ and $\Clspin$ generated under the above action by the equivalence classes $\LFone$  and $\LFspin$ of the local fields $1\in\LFPe$ and $\sigma:=\sigma_{o^{\circ}}\in\LFPo$ respectively. We have the following
simple observations:
\begin{lem}
	\label{lem: lattice_heighest_weight}
	We have \[\Curr n\LFone=0,\,\Currbar n\LFone=0\] for $n\in\Z+\frac12$, $n>0$ and \[\Curr n\LFspin=0,\,\Currbar n\LFspin=0\]
	for $n\in\Z$, $n >0$. Moreover, \[\Curr 0\LFspin=\lambb[\mu],\quad \Currbar 0\LFspin=\lamb[\mu],\]\[\Curr 0[\mu]=\lamb\LFspin,\quad \Currbar 0[\mu]=\lambb\LFspin,\]
	where we define $\mu:=\bar{\eta}_{o}\psi_{o}\sigma_{o^{\circ}}.$
\end{lem}

\begin{rem}
	It is straightforward to check that inside correlations, $\mu$ acts as the disorder insertion $\mu_{\root^\bullet}$, as defined in \cite[Section 2]{chelkak2021correlations}.
	
\end{rem}

\begin{proof}
	The fact that $\Curr n1=0,\,\Currbar n 1=0$ for $n>0$ follows
	simply from s-holomorphicity of $[\alpha z^{n}]$ for $n>0$ (shrink
	the contour). To see that $\Curr n\sigma=0,\,\Currbar n\sigma=0$
	for $n\geq1$: note that both $\psi_{z}\sigma_{o^{\circ}}$ and $[\alpha z^{n-\frac{1}{2}}]$ are s-holomorphic
	everywhere and ramified at $o^{\circ}$, so the
	contour can be shrunk to nothing.
	
	We can compute, using \eqref{eq: Curr_to_Rcurr},
	\begin{multline*}
		\Curr 0\sigma_{\root^\circ}\equiv\frac{1}{2\sqrt{2\pi}}\left(\RCurr 01-\i\RCurr 0{\i}\right)\sigma_{o^{\circ}}\\
		\equiv\frac{1}{4\sqrt{\pi}}\im\oint\left([\lambb z^{-\frac{1}{2}}]+[\lamb z^{-\frac{1}{2}}]\right)\psi_{z}\sigma_{o^{\circ}}\,dz-\frac{\i}{4\sqrt{\pi}}\im\oint\left([\lambb z^{-\frac{1}{2}}]-[\lamb z^{-\frac{1}{2}}]\right)\psi_{z}\sigma_{o^{\circ}}\:dz\\
		\equiv\frac{\lambb}{2\sqrt{2\pi}}\im\oint\left([\lamb z^{-\frac{1}{2}}]\right)\psi_{z}\sigma_{o^{\circ}}dz\equiv\frac{\lambb}{\sqrt{2\pi}}\bar{\eta}^{2}_{o}(\psi_{o^{+}}\sigma_{o}[\lamb\left(o^{+}\right)^{-\frac{1}{2}}]-\psi_{o^{-}}\sigma_{o}[\lamb\left(o^{-}\right)^{-\frac{1}{2}}])\\
		\equiv\lambb\mu_{o^{\bullet}}
	\end{multline*}
	using the fact that $[\lambb z^{-\frac{1}{2}}]$ and $\psi_{z}\sigma_{o^{\circ}}$
	are s-holomorphic as spinors ramified at $\root^\circ$, hence the contour can be shrunk to nothing,
	and that $[\lambda o^{-\frac{1}{2}}]=\left(\frac{\pi}{2}\right)^{\frac{1}{2}}\eta_{o}$. The case of $\Currbar{0}$ is similar, and the other two identities follow from the anti-commutation relation $\Curr{0}^2=\Currbar{0}^2=\Igen.$
\end{proof}

\begin{cor}
	\label{cor: Virasoro_structure}The modules $\ClInd$ and $\Clspin$
	are isomorphic to the fermionic Fock spaces $\FullFockNS$ and $\FullFockR$ respectively.
\end{cor}

\begin{proof}
	The commutation relations proved in the previous section show that the mapping $\Igen\mapsto \LFone$ (respectively, $\Igen\mapsto \LFspin$)  induces a surjective morphism $\CliffNS\to\ClInd$ as left $\CliffNS$ modules (respectively, $\CliffR\to\Clspin$ as left $\CliffR$ modules). Lemma \ref{lem: lattice_heighest_weight} ensures that $\mathcal{J}^{\NeveuSchwarz}$ (respectively, $\mathcal{J}^{\Ramond}$) is in the kernel of this morphism, thus yielding a surjective morphism $\FullFockNS\to \ClInd$ (respectively, $\FullFockR\to\Clspin$). Since by Lemma~\ref{lem: Fock_basis}, $\FullFockNS$ and $\FullFockR$ are irreducible, these morphisms have trivial kernels.
\end{proof}

Abusing the notation slightly, we will denote by $\HVM{\cdot}$ also the images in $\ClInd$ and $\Clspin$ of the $\TwoVir$ modules featuring in Proposition \ref{prop: Virasoro_structure} under the isomorphism of Corollary \ref{cor: Virasoro_structure}. We now find particular lattice fields corresponding to the vectors $\ind,\bpsi,\bpsibar,\ben,\bspin,\bmu$. For $\ind,\bspin$, this is clear, and for $\bmu$, it is done in Lemma \ref{lem: lattice_heighest_weight} above. We assume, without loss of generality, that $[z^{-1}]=\frac{\pi}{2}\CK{z}{c_1}$
and $[\i z^{-1}]=\frac{\pi}{2}\CK{z}{c_2},$ where the corners
$c_{1}$ and $c_{2}$ are midpoints to two opposite edges of a rhombus,
and that $\bar{\eta}_{c_{1}}=1$ and $\bar{\eta}_{c_{2}}=\i.$ We compute, using \eqref{eq: Disc_Cauchy_rec}

\[
\RCurr{-\frac{1}{2}}{1} 1\equiv\frac{\pi}{2}\im\oint\psi_{z}\CK z{c_{1}}\,dz\equiv2\pi\bar{\eta}_{c_{1}}\psi_{c_{1}}=2\pi\psi_{c_{1}}
\]

\[
\RCurr{-\frac{1}{2}}{\i}1\equiv\frac{\pi}{2}\im\oint\psi_{z}\CK z{c_{2}}\,dz\equiv2\pi\i\psi_{c_{2}},
\]
and using \eqref{eq: def_Psi_Psibar}, we get 
\[
\psi\isom\Curr{-\half}1\equiv\frac{\sqrt{\pi}}{\sqrt{2}}(\psi_{c_1}+\psi_{c_2});\quad \psi^\star\isom\Currbar{-\half}1\equiv\frac{\sqrt{\pi}}{\sqrt{2}}(\psi_{c_1}-\psi_{c_2}).
\]
For the energy $\ben$, we have 
\begin{multline*}
	\Curr{-\frac{1}{2}}\Currbar{-\frac{1}{2}}[1]= \frac{1}{8\pi}\left(\RCurr{-\frac{1}{2}}1-\i\RCurr{-\frac{1}{2}}{\i}\right)\left(\RCurr{-\frac{1}{2}}1+\i\RCurr{-\frac{1}{2}}{\i}\right)[1]\\
	=\frac{1}{8\pi}\left(\RCurr{-\frac{1}{2}}1\RCurr{-\frac{1}{2}}1+\RCurr{-\frac{1}{2}}{\i}\RCurr{-\frac{1}{2}}{\i}\right)[1]+\frac{1}{8\pi}\left(\i\RCurr{-\frac{1}{2}}1\RCurr{-\frac{1}{2}}{\i}-\i\RCurr{-\frac{1}{2}}{\i}\RCurr{-\frac{1}{2}}1\right)[1]=\frac{\i}{4\pi}\RCurr{-\frac{1}{2}}1\RCurr{-\frac{1}{2}}{\i}[1]
\end{multline*}
and, thanks to \eqref{eq: G vanishing for energy}, 
\begin{multline*}
	\RCurr{-\frac{1}{2}}1\RCurr{-\frac{1}{2}}{\i}1\equiv\frac{\pi}{2}\im\oint\psi_{z}\CK{c_{2}}{c_{1}}(2\pi\bar{\eta}_{c_{2}}\psi_{c_{2}})\,dz\equiv4\pi^{2}(\bar{\eta}_{c_{1}}\bar{\eta}_{c_{2}}\psi_{c_{1}}\psi_{c_{2}})\equiv4\pi^{2}\i\psi_{c_{1}}\psi_{c_{2}}\\
	\equiv-4\pi^{2}\sigma_{c^{\circ}_{1}}\sigma_{c^{\circ}_{2}}e^{-2\beta\sigma_{c^{\circ}_{1}}\sigma_{c^{\circ}_{2}}}=-4\pi^{2}(\sigma_{c^{\circ}_{1}}\sigma_{c^{\circ}_{2}}(\sin\theta)^{-1}-\cot\theta),
\end{multline*}
where $\theta$ is half the angle at $c_1^\bullet$ of the rhombus containing $c_1,c_2$, and we used the relationships $e^{-2\beta}=\tan{\half\theta}$, see \cite{ChelkakSmirnov2}, so that we finally have 
\[
\en \isom \frac{\i}{2}\Curr{-\frac{1}{2}}\Currbar{-\frac{1}{2}}1\equiv\frac{\pi}{2}(\sigma_{c^{\circ}_{1}}\sigma_{c^{\circ}_{2}}(\sin\theta)^{-1}-\cot\theta).
\]

\iffalse
(Note that here we used that $\psi_{c_{1}}\psi_{c_{2}}=\eta_{c_{1}}\eta_{c_{2}}\sigma_{c^{\circ}_{1}}\sigma_{c^{\circ}_{2}}e^{-2\beta\sigma_{c^{\circ}_{1}}\sigma_{c^{\circ}_{2}}}$
is consistent with the sign convention of Section \ref{subsec:Local-fields}).
\fi
%But the signs and the constant has to be checked carefully).

\subsection{All lattice local fields are in $\protect\ClInd \oplus\protect\Clspin$ }

\label{sec: surjectivity}
\begin{thm}
	\label{lem: surjectivity} We have 
	\[
	\LFPe/\Null = \ClInd,\quad \LFPo/\Null = \Clspin, 
	\]
	and moreover:
	\[
	\LFPee/\Null= \HVM{\ind}\oplus\HVM{\en},\quad \LFPoe/\Null= \HVM{\sigma},
	\]
	\[
	\LFPeo/\Null=\HVM{\psi}\oplus\HVM{\psi^{\star}},\quad \LFPoo/\Null= \HVM{\mu}
	\]	
\end{thm}

\begin{proof}
	For the first identity, it is sufficient to prove the inclusion
	$[\Op]\in\ClInd$ for any field of the form $\Op=\bar{\eta}_{z_{1}}\dots\bar{\eta}_{z_{k}}\psi_{z_{1}}\dots\psi_{z_{k}}.$ We proceed by induction in $k.$ The base $k=0$ and $\Op=1$ is clear.
	
	Note that $\Op(z):=\bar{\eta}_{z_{2}}\dots\bar{\eta}_{z_{k}}\psi_{z}\psi_{z_{2}}\dots\psi_{z_{k}}$
	is, inside a correlation with any real field far away, an s-meromorphic
	function with discrete poles at $z_{2},\dots,z_{2k}$, with values
	$\Op(z^{\pm}_{i})=\mp(-1)^{i}\eta_{z_{i}}\prod_{2<j\neq i\leq k}(\bar{\eta}_{j}\psi_{z_{j}}),$
 see Definition \ref{def: ferm_insertion} and the discussion thereafter. Therefore, integrating
	over a simple counterclockwise contour enclosing all the marked points,
	we can write by (\ref{eq: Discrete residue formula})
	\begin{multline*}
		\im\oint\CK z{z_{1}}\Op(z)\,dz\equiv2\prod^{k}_{j=1}(\bar{\eta}_{j}\psi_{z_{j}})+\sum^{k}_{i=2}4\bar{\eta}^{2}_{z_{i}}\CK{z_{i}}{z_{1}}(-1)^{i+1}\eta_{z_{i}}\prod_{2<j\neq i\leq k}(\bar{\eta}_{j}\psi_{z_{j}})
	\end{multline*}
	The sum in the right-hand side is in $\ClInd$ by induction
	hypothesis. For the left-hand side, we have by (\ref{eq: abstract_Laurent})
	and (\ref{eq: def_Psi_alpha}) 
	\begin{equation}
		\im\oint\CK z{z_{1}}\Op(z)\,dz\equiv\left(\sum^{N}_{i=1}\left(\alpha_{i}\RCurr{-i+\frac{1}{2}}1+\beta_{i}\RCurr{-i+\frac{1}{2}}{\i}\right)\right)\left(\bar{\eta}_{z_{2}}\dots\bar{\eta}_{z_{k}}\psi_{z_{2}}\dots\psi_{z_{k}}\right).\label{eq: LHS_induction}
	\end{equation}
	By induction hypothesis, the monomial in the right-hand side of (\ref{eq: LHS_induction}) belongs to $\ClInd$, and since it is also real, $\eqref{eq: Curr_to_Rcurr}$ implies that $\left[\im\oint\CK z{z_{1}}\Op(z)\,dz\right]$ also belongs to $\ClInd$.
	It follows that $\left[2\prod^{k}_{j=1}(\bar{\eta}_{j}\psi_{z_{j}})\right]\in\ClInd,$
	completing the induction step.
	
	Similarly, to prove  $\LFPo/\Null = \Clspin$, it suffices to prove the inclusion $[\Op]\in\Clspin$
	for any field of the form $\Op=\bar{\eta}_{z_{1}}\dots\bar{\eta}_{z_{k}}\psi_{z_{1}}\dots\psi_{z_{k}}\sigma_{\root^{\circ}}.$
	The base $k=0$ and $\Op=\sigma_{\root^{\circ}}$ is clear. For the
	induction step, we consider 
	\[
	\Op(z)=\bar{\eta}_{z_{2}}\dots\bar{\eta}_{z_{k}}\psi_{z}\psi_{z_{2}}\dots\psi_{z_{k}}\sigma_{\root^{\circ}},
	\]
	which is, inside a correlation with any real field far away, an s-meromorphic
	$\Cgr(\isGoc)$-spinor with poles at $z_{2},\dots,z_{k},$ and with
	values $\Op(z^{\pm}_{i})=\mp(-1)^{i}\eta_{z_{i}}\prod_{2<j\neq i\leq k}(\bar{\eta}_{j}\psi_{z_{j}})\sigma_{\root^{\circ}}.$
	Therefore, we have by (\ref{eq: Discrete residue formula})
	\[
	\im\oint\SCkx{z_{1}}{\lamb}(z)\Op(z)\,dz\equiv2\prod^{k}_{j=1}(\bar{\eta}_{j}\psi_{z_{j}})\sigma_{\root^{\circ}}+\sum^{k}_{i=2}4\bar{\eta}^{2}_{z_{i}}\SCkx{z_{1}}{\lamb}(z)(-1)^{i+1}\eta_{z_{i}}\prod_{2<j\neq i\leq k}(\bar{\eta}_{j}\psi_{z_{j}})\sigma_{\root^{\circ}}.
	\]
	As above, the sum in the right-hand side belongs to $\Clspin$
	by the induction hypothesis, and by combining (\ref{eq: abstract_Laurent_spinor}) and (\ref{eq: def_Psi_alpha}), we get 
	\[
	\im\oint\SCkx{z_{1}}{\lamb}\Op(z)\,dz\equiv\left(\sum^{N}_{i=\frac{1}{2}}\left(\alpha_{i}\RCurr{-i+\frac{1}{2}}1+\beta_{i}\RCurr{-i+\frac{1}{2}}{\i}\right)\right)\left(\bar{\eta}_{z_{2}}\dots\bar{\eta}_{z_{k}}\psi_{z_{2}}\dots\psi_{z_{k}}\sigma_{\root^\circ}\right),
	\]
	and applying the induction hypothesis yields the result.
	
	The ``moreover'' claims follow easily from Proposition \ref{prop: Virasoro_structure}
	by keeping track of the $\Z/(2\Z)$ grading by fermion parity. 
\end{proof}

\subsection{Concluding the proof}

We are now in the position to conclude the proof of the main theorem.
\begin{proof}[Proof of Theorem \ref{thm: main2}]
	Theorem \ref{lem: surjectivity} together with Corollary \ref{cor: Virasoro_structure} established the isomorphisms of $\LFPe(\isG)/\Null$ and  $\LFPo(\isG)/\Null$ to $\FullFockNS$ and $\FullFockR$ respectively. The map $\MainBij_{\isG}$
	is defined to be the direct sum of these isomorphisms,
	pre-composed with the projections $\Op\mapsto[\Op]$ and scaled to satisfy the second item in Theorem \ref{thm: main2}. The first two
	properties of the map $\MainBij_{\isG}$ claimed in Theorem \ref{thm: main2}
	are immediate from the construction. The ``moreover'' claim of Theorem \ref{lem: surjectivity}
	readily gives the third property. 
	
	We turn to the convergence claim. By linearity, it is enough to establish
	the case when for each $i$, we have $\Op^{(i)}\in\LFPe(\isG^{(i)})\cup\LFPo(\isG^{(i)})$.
	By Theorem \ref{lem: surjectivity} and Lemma \ref{lem: Fock_basis}, we can represent $\Op^{(i)},$ up to a null field, in the form 
	\[
	\Op^{(i)}\equiv\sum\beta^{(i)}_{\{n_{j}\},\{m_{j}\}}\Psi_{n_{1}}\dots\Psi_{n_{N}}\bar{\Psi}_{m_{1}}\dots\bar{\Psi}_{m_{N'}}\tilde{\Op}^{(i)},
	\]
	where $\tilde{\Op}^{(i)}\in\{1,\sigma_{\root^{\circ}}\},$ and the
	sum runs over all tuples of $n_{i}$
	and $m_{i}$ such that $n_{1}<\dots<n_{N}\leq0,$ $m_{1}<\dots<m_{N'}<0$,
	and almost all coefficients vanish. The scaling dimension $\Delta_{\Op^{(i)}}$
	can be calculated as 
	\[
	\Delta_{\Op^{(i)}}=\min\left\{ \tilde{\Delta}_{i}-\sum^{N}_{j=1}n_{j}-\sum^{N'}_{j=1}m_{j}:\beta^{(i)}_{\{n_{j}\},\{m_{j}\}}\neq0\right\} ,\quad\tilde{\Delta}_{i}=\begin{cases}
		0, & \tilde{\Op}^{(i)}=1\\
		\frac{1}{8}, & \tilde{\Op}^{(i)}=\sigma_{\root^{\circ}}
	\end{cases}.
	\]
	Using (\ref{eq: def_Psi_Psibar}), we can now write 
	\[
	\Op^{(i)}\equiv\sum\hat{\beta}^{(i)}_{\vectn}\Psi^{[\alpha_{1}]}_{n_{1}}\dots\Psi^{[\alpha_{N(\vectn)}]}_{n_{N(\vectn)}}\tilde{\Op}^{(i)},
	\]
	where the sum is over tuples $\vectn=(n_{1},\dots,n_{N(\vectn)},\alpha_{1},\dots,\alpha_{N(\vectn)}):$
	$n_{1}<\dots<n_{N(\vectn)}\leq0,$ $\alpha_{1},\dots,\alpha_{N(\vectn)}\in\{1,\i\},$
	almost all coefficients $\hat{\beta}^{(i)}_{\vectn}$ vanish, and
	all coefficients $\hat{\beta}^{(i)}_{\vectn}$ with $\tilde{\Delta}_{i}-\sum^{N(\vectn)}_{j=1}n_{j}<\Delta_{\Op^{(i)}}$
	vanish. In order to compute the scaling limits of the correlation
	$\delta^{-\sum^{k}_{i=1}\Delta_{\Op^{(i)}}}\E\left[\Op^{(1)}_{\delta,\z_{1}}\dots\Op^{(k)}_{\delta,\z_{k}}\right],$
	we plug in the above sum and expand it by linearity, 
	\begin{equation}
		\E\left[\Op^{(1)}_{\delta,\z_{1}}\dots\Op^{(k)}_{\delta,\z_{k}}\right]=\sum_{\vectn^{(1)},\dots,\vectn^{(k)}}\left(\prod^{k}_{i=1}\hat{\beta}^{(i)}_{\vectn^{(i)}}\right)\E\left[\Op^{(1,\vectn^{(1)})}_{\delta,\z_{1}}\dots\Op^{(1,\vectn^{(k)})}_{\delta,\z_{k}}\right],\label{eq: corr_expansion_for_sl}
	\end{equation}
	where $\Op^{(i,\vectn^{(i)})}:=\Psi^{[\alpha_{1}]}_{n_{1}}\dots\Psi^{[\alpha_{N}]}_{n_{N}}\tilde{\Op}^{(i)}$.
	The coefficients $\hat{\beta}^{(i)}_{\vectn^{(i)}}$ in (\ref{eq: corr_expansion_for_sl})
	do not depend on $\delta$, so we only need to take care of convergence
	of the expectations. Denote $N^{(i)}:=N(\vectn^{(i)})$ and $\Delta'_{i}:=\tilde{\Delta}_{i}-\sum^{N(\vectn^{(i)})}_{j=1}n^{(i)}_{j}.$
	We will prove a more general result, namely, that for any $m$, we
	have the convergence 
	\begin{multline*}
		\delta^{-\Delta'_{1}-\dots-\Delta'_{k}-\frac{m}{2}}\E_{\Od}\left[\Op^{(1,\vectn^{(1)})}_{\delta,\z_{1}}\dots\Op^{(1,\vectn^{(k)})}_{\delta,\z_{k}}\psi_{w_{1}}\dots\psi_{w_{m}}\right]\\
		=\Cpsi^{m}\prod^{m}_{i=1}\eta_{i}\left\langle \MainBij\left(\Op^{(1,\vectn^{(1)})}_{\delta,\z_{1}}\right)\dots\MainBij\left(\Op^{(1,\vectn^{(k)})}_{\delta,\z_{k}}\right)\psi^{[\eta_{w_{1}}]}_{w_{1}}\dots\psi^{[\eta_{w_{m}}]}_{w_{m}}\right\rangle _{\Omega}+o(1),
	\end{multline*}
	uniformly as $z_{i}$ and $w_{i}$ stay away from each other and the
	boundary, by induction on $N^{(1)}+\dots+N^{(k)}.$ In the base case,
	we have $\Op^{(i,\vectn^{(i)})}=\ind$ or $\Op^{(i,\vectn^{(i)})}=\sigma_{\root^{\circ}}$
	for each $i$, and the result follows from the convergence of spin-fermion
	correlation \cite[Theorem 1.2]{chelkak2021correlations}, \cite{chelkak2023universality}.
	For the induction step, pick $1\leq j\leq k$, $\alpha\in\{1,\i\}$,
	and $n\in\Z+\frac{1}{2}$ if $\tilde{\Op}_{j}=1,$ respectively, $n\in\Z$
	if $\tilde{\Op}_{j}=\sigma_{\root^{\circ}}.$ Let $\cvr$ be the double
	cover of $\Od$ ramified at those $i$ for which $\tilde{\Op}_{i}=\sigma_{\root^{\circ}}$
	and for $w\in\Cgr(\Od_{\cvr}),$ write 
	\[
	F_{\delta}(w)=\delta^{-\Delta'_{1}-\dots-\Delta'_{k}-\frac{m}{2}-\frac{1}{2}}\prod^{m}_{i=1}\bar{\eta}_{w_{i}}\E_{\Od}\left[\Op^{(1,\vectn^{(1)})}_{\delta,\z_{1}}\dots\psi_{w}\Op^{(1,\vectn^{(j)})}_{\delta,\z_{j}}\dots\Op^{(1,\vectn^{(k)})}_{\delta,\z_{k}}\psi_{w_{1}}\dots\psi_{w_{m}}\right],
	\]
	which is an $s$-holomorphic function of $w$: note that $\Op^{(i,n^{(i)})},$
	as well as $\bar{\eta}_{w_{i}}\psi_{w_{i}}$ are real fields. In a
	vicinity of $z_{j},$ write also 
	\[
	G_{\delta}(w)=\delta^{n-\frac{1}{2}}[\alpha\left(\delta^{-1}(w-z_{j})\right)^{n-\frac{1}{2}}]
	\]
and note that $G_{\delta}(w)=\proj{\eta_{w}}{g(w)}+o(1),$ where 
\[
g(w)=\alpha(w-z_{j})^{n-\frac{1}{2}}.
\]	
Then we have, by definition, 
	\begin{multline*}
		\left(\Psi^{[\alpha]}_{n}\Op^{(j,\vectn^{(j)})}\right)_{\delta,z_{j}}=\sum_{w\in\gamma}2s_{\gamma}(w)\bar{\eta}^{2}_{w}\psi_{w}[\alpha\left(\delta^{-1}(w-z_{j})\right)^{n-\frac{1}{2}}]\Op^{(j,\vectn^{(j)})}\\
		=\delta^{-1}\delta^{\frac{1}{2}-n}\im\oint_{\gamma}\psi_{w}G_{\delta}(w)\Op\,dw.
	\end{multline*}
	When inserted into correlations, we thus get 
	\begin{multline*}
		\delta^{-\Delta'_{1}-\dots-\Delta'_{k}+n-\frac{m}{2}}\E_{\Od}\left[\Op^{(1,\vectn^{(1)})}_{\delta,\z_{1}}\dots\left(\Psi^{[\alpha]}_{n}\Op^{(j,\vectn^{(j)})}\right)_{\delta,z_{j}}\dots\Op^{(1,\vectn^{(k)})}_{\delta,\z_{k}}\psi_{w_{1}}\dots\psi_{w_{m}}\right]\\
		=\prod^{m}_{i=1}\eta_{w_{i}}\im\oint_{z_{j}}F_{\delta}(w)G_{\delta}(w)\,dw.
	\end{multline*}
	We pull the contour of integration so that $w$ it stays at a small
	fixed distance from $z_{j}.$ We have by induction hypothesis $F_{\delta}(w)=\proj{\eta_{w}}{f(w)}+o(1)$, where 
	\[
	f(w)=\Cpsi^{m+1}\left\langle \MainBij(\Op^{(1,\vectn^{(1)})}_{\delta,\z_{1}})\dots\psi_{w}\MainBij\left(\Op^{(1,\vectn^{(j)})}_{\delta,\z_{j}}\right)\dots\MainBij\left(\Op^{(1,\vectn^{(k)})}_{\delta,\z_{k}}\right)\psi^{[\eta_{w_{1}}]}_{w_{1}}\dots\psi^{[\eta_{w_{m}}]}_{w_{m}}\right\rangle.
	\]
	
	Therefore, applying $(\ref{eq: Im_int_approximation}),$we conclude
	that 
	\begin{multline*}
		\delta^{-\Delta'_{1}-\dots-\Delta'_{k}+n-\frac{m}{2}}\prod^{m}_{i=1}\bar{\eta}_{w_{i}}\E_{\Od}\left[\Op^{(1,\vectn^{(1)})}_{\delta,\z_{1}}\dots\left(\Psi^{[\alpha]}_{n}\Op^{(j,\vectn^{(j)})}\right)_{\delta,z_{j}}\dots\Op^{(1,\vectn^{(k)})}_{\delta,\z_{k}}\psi_{w_{1}}\dots\psi_{w_{m}}\right]\\
		=\im\oint f(w)g(w)\,dw+o(1)\\
		=\frac{1}{2\i}\left(\alpha\oint f(w)(w-z_{j})^{n-\frac{1}{2}}dw-\bar{\alpha}\oint\bar{f(w)}(\bar w-\bar z_{j})^{n-\frac{1}{2}}d\bar w\right)+o(1)\\
		=\Cpsi^{m}\left\langle \MainBij\left(\Op^{(1,\vectn^{(1)})}\right)_{z_{1}}\dots\left(\Psi^{[\alpha]}_{n}\MainBij\left(\Op^{(j,\vectn^{(j)})}\right)\right)_{z_{j}}\dots\MainBij\left(\Op^{(1,\vectn^{(k)})}\right)_{z_{k}}\psi^{[\eta_{w_{1}}]}_{w_{1}}\dots\psi^{[\eta_{w_{m}}]}_{w_{m}}\right\rangle _{\Omega}+o(1),
	\end{multline*}
	where $\Psi^{[\alpha]}_{n}=\sqrt{2\pi}\left(\alpha\Psi_{n}+\bar{\alpha}\bar{\Psi}_{n}\right)$,
	and we recall that $\Cpsi=\frac{\sqrt{2}}{\sqrt{\pi}}.$ In the third
	equality, we used \cite[Proposition 5.18]{chelkak2021correlations}.
	By construction of the correspondence $\MainBij,$ we have inside the correlations
	\[
	\Psi_{n}\MainBij(\Op)=\MainBij(\Psi_{n}(\Op)),\quad\bar{\Psi}_{n}\MainBij(\Op)=\MainBij(\bar{\Psi}_{n}(\Op)),
	\]
	which concludes the proof of the induction step. 
	
	The terms in (\ref{eq: corr_expansion_for_sl}) fall into two groups:
	those satisfying $\Delta_{i}'=\Delta_{\Op^{(i)}}$ for all $i=1,\dots,k$,
	and those for which $\Delta_{i}'\geq\Delta_{\Op^{(i)}}$ for all $i$, with at least one inequality strict. When rescaled by $\delta^{-\sum^{k}_{i=1}\Delta_{\Op^{(i)}}},$
	the former terms have scaling limit as above, and the latter ones
	converge to zero. Since keeping only the monomials satisfying $\Delta_{i}'=\Delta_{\Op^{(i)}}$
	is tantamount to applying the projection $P_{\Delta_{\Op^{(i)}}},$
	we arrive at the desired result. 
\end{proof}

\begin{proof}[Proof of Corollary \ref{thm: main1}]
	We define $\hat{\MainBij}(\hat{\Op})=\MainBij(\Op)$, where $\Op\mapsto\hat{\Op}$
	denotes the isomorphism (\ref{eq: spin-ferm-isomorphism}). The corollary
	is then a direct consequence of Theorem \ref{thm: main2}, in particular,
	its third point, and the considerations at the end of Section \ref{sec: lattice_Fock_structure}.
\end{proof}
\section{Examples of explicit computations}

In order to demonstrate the concrete flavor of the correspondences $\MainBij$ and $\MainBijSpins$ and their inverses, we compute, on the square lattice, an explicit lattice field that converges in the scaling limit  to the stress-energy tensor, as well as the Virasoro counterpart of a product of three adjacent spins. To this end, we will need formulae for the good monomials with \emph{negative indices}, more precisely: 
\begin{enumerate}
	\item the coefficients expressing bases of $\setG{E_{2k-1}}$ and $\setQ{E_{2k}}$ discussed in Lemmas \ref{lem: negative_monomial_existence}, \ref{lem: no_degree_jumps}, and Remark \ref{rem: various_bases}, with respect to each other;
	\item explicit values for $\CK{z}{a}$ and $\Gv{z}{a}$ for $z$ close to $a$.
\end{enumerate}
Both are easy to obtain from the representation using the integral expressions via discrete exponentials \eqref{eq: Cauchy_kernel_integral}, \eqref{eq: Gu_integral_dexp}, \eqref{eq: Gv_integral_dexp}. The resulting correspondence depends on the choice of the sequence $E_n$ of sets of rhombi, and the choices that lead to the most compact expressions are slightly different in the two cases.

We stress that given the above analysis, all computations of this sort are entirely straightforward and, in principle,  can be automated.

\subsection{Discretization of the stress-energy tensor} \label{sec: tensor} It is convenient to consider the rotated and \emph{shifted} square lattice of mesh size $\delta=1$,  $\isG=\SqLat=\sqrt 2 e^{\frac{\i \pi}{4}}\Z+1$, with the root $\root$ localted at $\frac12$; in this case, $\root^\bullet$ is at the origin and $\root^\circ$ at $1$. We define $E_n$ by the sequence of corners $a_n:=\frac12\i^n$, $n=0,1,2,3,$ so that $E_1=\{e_1\},$  $E_2=\{e_1,e_2\}$ and $E_3=\{e_1,\dots,e_4\}$, where the rhombi $e_1,\dots,e_4$ are the four unit squares incident to $\root^\bullet$, see Fig. \ref{fig: tensor}. The choice of $E_n$ for $n\geq 4$ will not affect our computations. In the computation below, we choose the signs by $\ds{a_n}=\lamb e^{-\frac{i\pi n}{4}}$.

\begin{figure}
	\centering{\includegraphics[width=0.4\textwidth]{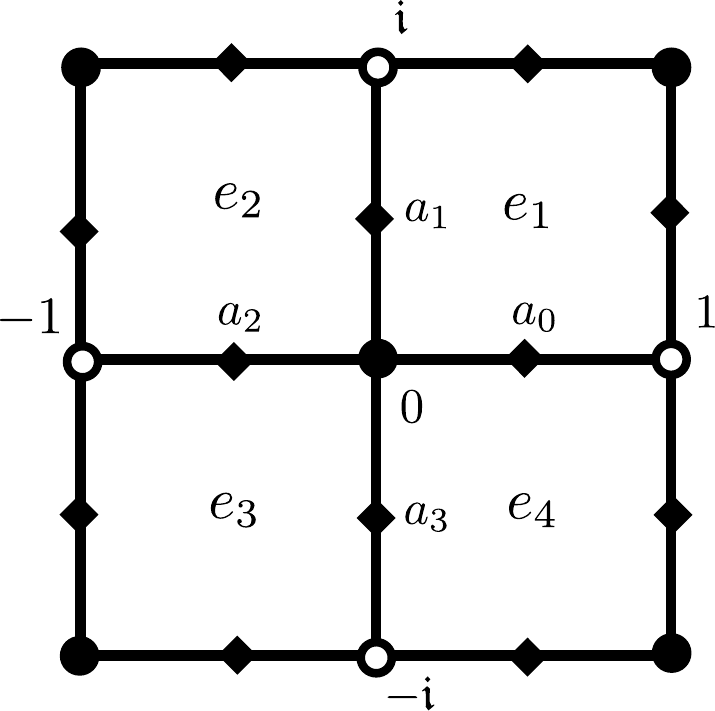}}
	\label{fig: tensor}
	\caption{The notation for the computation in Section \ref{sec: tensor}.} 
\end{figure}

\begin{lem}
	\label{lem: explicit_good_monomials}
	On the square lattice $\SqLat$ with $a_i$ chosen as above, we have 	
	\[
	[z^{-1}]= \Cpsimone\left(\sqrt{2}\CK{z}{a_0}-\CK{z}{a_1}\right);\quad [\i z^{-1}]=\Cpsimone \CK{z}{a_1};
	\]
	\[
	[z^{-2}]=-\frac{1}{\sqrt2}\Cpsimone\left(\sqrt{2}\CK{z}{a_1}-\CK{z}{a_2}-\CK{z}{a_0}\right);\quad  
	\]
	\[
	[\i z^{-2}]=\frac{1}{\sqrt{2}}\Cpsimone\left(\sqrt{2}\CK{z}{a_3}-\CK{z}{a_2}+\CK{z}{a_0}\right);\quad  
	\]

\end{lem}

\begin{proof}
	In general, proving such identities boils down to checking that the RHS belongs to the corresponding space $\setG{E_{2k+1}}$ (in our case, $k=0$ and $k=1$) and the correct asymtotics at infinity, which translates, as in Lemma \ref{lem: Cauchy_expansion}, to the expansion of $R$ at the origin. In this case, one can alternatively use \eqref{eq: G_a_asymp} to check the asymptotics.
\end{proof}

The holomoprhic (respectively, anti-holomoprhic) component of the \emph{stress-energy tensor} in CFT is defined as $T=\Lgen{-2}\ind$ (respectively, $\bar{T}=\Lbar{-2}\ind$). We now find its explicit discretization, first as a field in $\LFPee$ and then in $\LF(\isG)$. We focus on the holomorphic part $T$, as $\bar{T}$ is similar. Since $\Curr{n}\ind=0$ for $n>0$, the only $k$ in the Sugawara expression \eqref{eq:fermionic_sugawara_ns} for which $\nord{\Foper{-2-k} \Foper{k}} \, \ind\neq 0$ are $k=-\frac12$ and $k=-\frac32,$ hence we have, recalling \eqref{eq: def_Psi_Psibar}, 
\[
\Lgen{-2}\ind=\frac14 \Curr{-\frac32}\Curr{-\frac12}\ind = \frac{1}{32\pi}\left(\RCurr{-\frac32}{1}\RCurr{-\frac12}{1}-\RCurr{-\frac32}{\i}\RCurr{-\frac12}{\i}-\i \RCurr{-\frac32}{\i}\RCurr{-\frac12}{1}-\i\RCurr{-\frac32}{1}\RCurr{-\frac12}{\i}\right)\ind.
\]   
We have by the discrete residue formula \eqref{eq: Discrete residue formula}, 
\[
\im\oint \CK{z}{a}\psi_z\,dz\equiv 4\dsbar{a}\psi_a;
\] 
\[
\im\oint \CK{z}{a}\psi_z\dsbar{b}\psi_b\,dz\equiv\begin{cases}
4\dsbar{a}\dsbar{b}\psi_{a}\psi_b+4\CK{b}{a}\dsbar{b}& a\neq b, \\
0& a=b.
\end{cases}
\]
By combining the above expression with Lemma \ref{lem: explicit_good_monomials} and \eqref{eq: def_Psi_alpha}, we find
\begin{multline*}
\Lgen{-2}1\equiv \frac{1}{16}\left(\RCurr{-\frac32}{1}\left(\sqrt{2}\dsbar{a_0}\psi_{a_0}-\dsbar{a_1}\psi_{a_1}\right)-\RCurr{-\frac32}{\i}\dsbar{a_1}\psi_{a_1}\right)
\\
-\frac{\i}{16}\left( \RCurr{-\frac32}{1}\dsbar{a_1}\psi_{a_1}+\RCurr{-\frac32}{\i}\left(\sqrt{2}\dsbar{a_0}\psi_{a_0}-\dsbar{a_1}\psi_{a_1}\right)\right)
\\
\equiv\frac{\pi}{8}\left(\hihi{a_2}{a_0}-\hihi{a_3}{a_1}\right)
\\
-\frac{\i\pi}{8}\left(\sqrt2(\hihi{a_2}{a_1}+\hihi{a_3}{a_0})\right.
\\
\left.-\hihi{a_3}{a_1}-\hihi{a_2}{a_0}+\frac{4}{\pi}-2\right)
\end{multline*}
where the constants come from calculating, using \eqref{eq: Cauchy_kernel_integral}, 
\[
\Gab{a_2}{a_0}=\Gab{a_3}{a_1}=-\frac2\pi;\quad 
\]
\[
\Gab{a_2}{a_1}=\Gab{a_1}{a_0}=\Gab{a_3}{a_0}=-\frac{\sqrt{2}}{2}.
\]
Our sign conventions imply $\hihi{a_i}{a_j}\equiv\sigma_{a^\circ_i}\sigma_{a^\circ_j}$ for $0\leq j<i\leq 3,$ therefore, we finally arrive at 
\begin{equation*}
	T\equiv
	\frac{\pi}{8}\left(\sigma_{1}\sigma_{-1}-\sigma_{\i}\sigma_{-\i}\right)
	-\frac{\i\pi}{8}\left(\sqrt2(\sigma_{-1}\sigma_{\i}+\sigma_{-\i}\sigma_{1})-\sigma_{1}\sigma_{-1}-\sigma_{\i}\sigma_{-\i}+\frac{4}{\pi}-2\right).
\end{equation*}

\subsection{Product of three neighboring spins}
\label{sec: three_spins}
We now give an example of going in the opposite direction: on the lattice $\isG=\SqLat=\sqrt 2 e^{\frac{\i \pi}{4}}\Z$, with $\root^\circ=0$ and $\root^\bullet =\i$, we consider the field 
\[
\hat{\Op}:=\sigma_0\sigma_{1+\i}\sigma_{-1+\i}\equiv \dsbar{b_2}\dsbar{b_1}\psi_{b_2}\psi_{b_1}\sigma_{\root^\circ}=:\Op.
\] Here $b_1$ and $b_2$ are lifts of $\frac12+
\i$ and $-\frac12+\i$, respectively, to same sheet of the double cover $\Cgr(\isG_{[\root^\circ]})$ over the upper half-plane, see Fig. \ref{fig: three_spins} The signs will be chosen  as follows: we take $\ds{b_1}=\lamb$ and $\ds{b_2}=-\lambb$, so that $\dsbar{b_2}\dsbar{b_1}=-1$; this way, the above equivalence is checked by starting with $\dsbar{b_1}\dsbar{b_1}\psi_{b_1^+}\psi_{b_1}\sigma_0=\sigma_{0}$ and rotating $b_1^+$ around $\i$ counteclockwise to $b_2$. We calculate the corresponding vectors $\MainBij(\Op)
\in\FullFockNS\oplus\FullFockR$ and $\MainBijSpins(\hat\Op)\in M(4,3)$; in fact, since $\Op$ is an \emph{odd} spin polynomial, we know that $\MainBij(\Op)\in \FullFockR$ and $\MainBijSpins(\hat\Op)\in \HVM{\bspin}.$

\begin{figure}
	\centering{\includegraphics[width=0.4\textwidth]{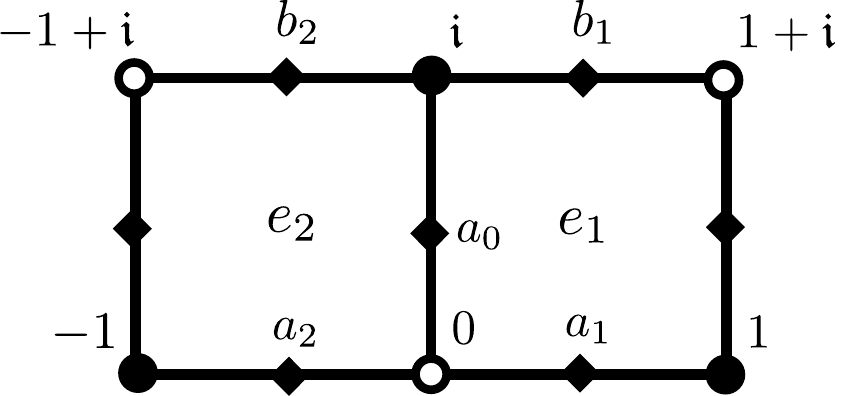}}
	\label{fig: three_spins}
	\caption{The notation for the computation in Section \ref{sec: three_spins}.} 
\end{figure}

It is convenient to choose $a_0=\root =\frac{\i}{2},$ $a_1=\frac12$ and $a_2=-\frac12,$ so that $e_1$ and $e_2$ are the two (upper) unit squares incident to the origin, $E_1=\{e_1\}$ and $E_2=\{e_1,e_2\}$. The choice of $E_k$ for $k\geq 3$ will not affect the result. 

Below, we will consider linear combinations of $\Gv{\cdot}{w}$ for $u
\in
\{0,1,-1,\i\},$ which are defined on different double covers. To this end, we use a sheet $\Sheet$ over the upper half-plane slit along the segment $(0,\i),$ and on that sheet, we fix the sign by  $\Gv{u+\frac{\i}{2}}{u}=\CCpsi\i$, $u\in\{-1,1,\i\}$ and $\Gv{\frac{\i}{2}}{0}=\CCpsi$; then \eqref{eq: asymp_zhalf_lamb}--\eqref{eq: asymp_zhalf_lambb}  are satisfied on $\Sheet$ with the same choice of the branch of the square root. 

\begin{lem}
	\label{lem: explicit_good_monomials_Ramond}
	With the above choices, we have 
	\[
	[\lambb z^{-\frac12}]=\Gv{z}{0};\quad [\lamb z^{-\frac12}]=\Gv{z}{\i}
	\]
	\[
	[z^{-\frac32}]=\sqrt{2}\left(\Gv{z}{\i}-\Gv{z}{-1}\right);\quad [\i z^{-\frac32}]=\sqrt{2}\left(\Gv{z}{\i}-\Gv{z}{1}\right)
	\]
\end{lem}

\begin{proof}
	 Simialrly to Lemma \ref{lem: explicit_good_monomials}, one readily checks that the RHS belongs to $\setQ{E_{2k}}$ ($k=0$ and $k=1$) and satisfies the required  asymptotics at infinity, using \eqref{eq: asymp_zhalf_lamb}, \eqref{eq: asymp_zhalf_lambb}.
\end{proof}

Following the proof of Theorem \ref{lem: surjectivity}, we start by observing that 
\[
4\dsbar{b_1}\psi_{b_1}\sigma_{\root^\circ}\equiv\im\oint \SCkx{b_1}{\lambb}(z)\psi_z\sigma_{\root^\circ}\,dz
\]
and
\begin{multline}
	\label{eq: three_spins}
	\dsbar{b_2}\dsbar{b_1}\psi_{b_2}\psi_{b_1}\sigma_0
	\equiv
	\frac14\im\oint \SCkx{b_2}{\lambb}(z)\psi_z\left(\dsbar{b_1}\psi_{b_1}\sigma_{\root^\circ}\right)\,dz
	-\dsbar{b_1}\SCkx{b_2}{\lambb}(b_1)\sigma_{\root^\circ}
	\\ 
	\equiv \frac{1}{16}\im\oint \SCkx{b_2}{\lambb}(z)\psi_z\im\oint \SCkx{b_1}{\lambb}(w)\psi_w\sigma_{\root^\circ}\,dw\,dz
	-\dsbar{b_1}\SCkx{b_2}{\lambb}(b_1)\sigma_{\root^\circ}.
\end{multline}

Next, we need expressions of $\SCkx{b_1}{\lambb}(z)$ and $\SCkx{b_2}{\lambb}(z)$ in terms of $[\alpha z^n]$. We first observe, in the vein of the proof of Lemma \ref{lem: SCK_exists} that, on $\Sheet$, we have 
\[
\SCkx{a_1}{\lambb}(z)=\Cpsimoroot\Gv{z}{1};\quad \SCkx{a_0}{\lambb}(z)=\Cpsimoroot\Gv{z}{\i}\quad \SCkx{a_2}{\lambb}(z)=\Cpsimoroot\Gv{z}{-1};
\]
with the following choice of signs:  $\ds{a_0}=1,\ds{a_1}=-\lambb,\ds{a_2}=-\lamb$; this guarantees the first property in the definition of $\SCkx{v}{\cdot}(\cdot)$.  
From \eqref{eq: CK_three_terms}, we find 
\begin{multline*}
	\SCkx{b_1}{\lambb}(z)=\sqrt{2}\SCkx{a_0}{\lambb}(z)+\SCkx{a_1}{\lambb}(z)=\Cpsimoroot\left(\sqrt{2}\Gv{z}{\i}+\Gv{z}{1}\right)\\
	=\Cpsimoroot\left((\sqrt{2}+1)[\lamb z^{-\frac12}]-\frac{\sqrt{2}}{2}[\i z^{-\frac32}]\right),
\end{multline*}
and similarly 
\begin{multline*}
\SCkx{b_2}{\lambb}(z)=-\sqrt{2}\SCkx{a_0}{\lambb}(z)-\SCkx{a_2}{\lambb}(z)=\Cpsimoroot\left(-\sqrt{2}\Gv{z}{\i}-\Gv{z}{-1}\right)
\\=\Cpsimoroot\left(-(\sqrt{2}+1)[\lamb z^{-\frac12}]+\frac{\sqrt{2}}{2}[ z^{-\frac32}]\right).
%=\Cpsimoroot\left((1-\frac{\sqrt{2}}{2})([z^{-\frac12}]-[\iz^{-\frac12}])+\frac{\sqrt{2}}{2}[z^{-\frac32}]\right).
\end{multline*}
From the last formula and \eqref{eq: Gv_integral_dexp}, we in particular obtain
\begin{multline*}
\SCkx{b_2}{\lambb}(b_1)=\Cpsimoroot\left(-\sqrt{2}\Gv{b_1}{\i}-\Gv{b_1}{-1}\right)
\\=-\frac{ \lambb\sqrt{2}}{\pi}\int_{-\infty}^0\zeta^{-\frac12}(1-\zeta)^{-1}\,d\zeta-\frac{ \lambb}{\pi}\int_{-\infty}^0\zeta^{-\frac12}\frac{(1+\zeta)(1+\i\zeta)}{(1-\i\zeta)(1-\zeta)^{2}}\,d\zeta=-2\lamb;
\end{multline*}
the branch of $\zeta^{-\frac12}$ is chosen so that $\Gv{\cdot}{v}$, $v\in\{\i,-1\}$ as given by \eqref{eq: Gv_integral_dexp} has the correct sign at $v+\frac{\i}{2}\in\Sheet$. Plugging this into \eqref{eq: three_spins} and recalling the definitions of $\RCurr{n}{\alpha}$, we obtain
\begin{multline*}
	\dsbar{b_2}\dsbar{b_1}\psi_{b_2}\psi_{b_1}\sigma_0\equiv 2\sigma_{0}
	+\frac{1}{8\pi}\left(\left(1+\frac{\sqrt{2}}{2}\right)(\RCurr{0}{\lamb}\RCurr{-1}{\i}+\RCurr{-1}{1}\RCurr{0}{\lamb})-\frac{1}{2}\RCurr{-1}{1}\RCurr{-1}{\i}\right)\sigma_0
	\\
\equiv 2\sigma_{0}+\frac{1}{8\pi}(\sqrt{2}+1)\RCurr{-1}{\lamb}\RCurr{0}{\lamb}\sigma_{0}-\frac{1}{16\pi}\RCurr{-1}{1}\RCurr{-1}{\i}\sigma_0
\end{multline*}
where we have also used the anti-commutation relations, in particular, $\RCurr{0}{\lamb}\RCurr{0}{\lamb}=0$. Taking into account \eqref{eq: Curr_to_Rcurr} and also $\Curr{0}\sigma_{\root^\circ}=\i\Currbar{0}\sigma_{\root^\circ},$ we arrive at 
\begin{multline*}
\sigma_0\sigma_{1+\i}\sigma_{-1+\i}\equiv 2\sigma_{0}+\frac{\sqrt{2}+1}{4} \left(-2\i \Curr{-1}\Curr{0}+2\i \Currbar{-1}\Currbar{0}\right)\sigma_0-
\frac{\i}{4}\Currbar{-1}\Curr{-1}\sigma_0\\
\equiv 2\sigma_0-2\i(\sqrt{2}+1)(\Lgen{-1}-\Lbar{-1})\sigma_0+4\Lbar{-1}\Lgen{-1}\sigma_0;
\end{multline*}
in the last line, we used that by \eqref{eq:fermionic_sugawara_r} $\Lgen{-1}\sigma_0=\frac14 \Curr{-1}\Curr{0}\sigma_0$, and we also substituted $\i\sigma_0=\Currbar{0}\Curr{0}\sigma_0$. 

The interpretation of this result in terms of scaling limits is as follows. First, $\delta^{-\frac18}\sigma_0\sigma_{1+\i}\sigma_{-1+\i}$ has the same scaling limit as $2\delta^{-\frac18}\sigma_0,$ that is,  
\begin{multline*}
	\delta^{-\frac18+\Delta_1+\dots+\Delta_n}\E[\sigma_z\sigma_{z+\delta+\i\delta}\sigma_{z-\delta+\i\delta}\Op_{\delta,z_1}^{(1)}\dots\Op_{\delta,z_n}^{(n)}]\\
= 2\Csigma \ccor{\sigma_{z}P_{\Delta_1}(\MainBij(\Op^{(1)})_{z_1}\dots P_{\Delta_n}(\MainBij(\Op^{(n)})_{z_n}}+o(1),
\end{multline*}
for any $\Op^{(i)}\in \LFP(\Gamma)$ and $\Delta_i=\Delta_{\Op^{(i)}}.$ Furthermore, the rescaled difference 
\[
\delta^{-\frac98}\sigma_0(\sigma_{1+\i}\sigma_{-1+\i}-2)
\] 
has a scaling limit given by $-2\i(\sqrt{2}+1)(\Lgen{-1}-\Lbar{-1})\sigma$. Since, as checked in Appendix, the action of $\Lgen{-1}$ (respectively, $\Lbar{-1}$) in the continuum amounts to applying $\pa$ (respectively, $\bar{\pa}$) to the correlation, we conclude that 

\begin{multline*}
	\delta^{-\frac98+\Delta_1+\dots+\Delta_n}\E[\sigma_z(\sigma_{z+\delta+\i\delta}\sigma_{z-\delta+\i\delta}-2)\Op_{\delta,z_1}^{(1)}\dots\Op_{\delta,z_n}^{(n)}]\\
	= -\Csigma 2(1+\sqrt{2})\pa_y\ccor{\sigma_{x+
		\i y}P_{\Delta_1}(\MainBij(\Op^{(1)})_{z_1}\dots P_{\Delta_n}(\MainBij(\Op^{(n)})_{z_n}}+o(1), \quad z=x+\i y,
\end{multline*}
for any $\Op^{(i)}\in \LFP(\Gamma)$ and $\Delta_i=\Delta_{\Op^{(i)}}$. Finally, the reader may check that 
\[
\RCurr{-1}{\lamb}\RCurr{0}{\lamb}\sigma_{0}\equiv 8\pi (2\sigma_0-\sigma_{1+\i}-\sigma_{-1+\i}).
\]
This means that the field $\Op:=\sigma_{0}\sigma_{1+\i}\sigma_{-1+\i}-2\sigma_0-(\sqrt{2}+1)(2\sigma_0-\sigma_{1+\i}-\sigma_{-1+\i})$ has the scaling dimension $\Delta_\Op=\frac{17}{8}$, and 
\begin{multline*}
	\delta^{-\frac{17}{8}+\Delta_1+\dots+\Delta_n}\E[\Op_{\delta,z}\Op_{\delta,z_1}^{(1)}\dots\Op_{\delta,z_n}^{(n)}]\\
	= \Csigma (\pa^2_x+\pa^2_y)\ccor{\sigma_{x+
			\i y}P_{\Delta_1}(\MainBij(\Op^{(1)})_{z_1}\dots P_{\Delta_n}(\MainBij(\Op^{(n)})_{z_n}}+o(1), \quad z=x+\i y.
\end{multline*}

\begin{rem}
	We stress that the expansions as above \emph{are not} asymptotic expansions of correlations in powers of $\delta$. Indeed, obtaining even the first correction, say, for one-point spin expectation with $+$ boundary conditions would require stronger assumptions on the approximations $\Od\to\Omega.$
\end{rem}
\section*{Appendix: Ward identities and $\Lgen{n},\Lbar{n}$ as differential operators}
In this Appendix, we check the properties of the stress-energy tensor $T_z=(\Lgen{-2}\ind)_z$ (and its anti-holomorphic counterpart $\SETbar_z=(\Lbar{-2}\ind)_z$) in the continuum, in particular, the Ward identities and the interpretation of $\Lgen{-1}$ (respectively, $\Lbar{-1}$) as the holomorphic derivative $\pa_z$ (respectively, $\bar\pa_z$). Since $\SET$ and $\SETbar$ are constructed in terms of the fermionic current action and not in the more usual way as variation of the underlying metric, these properties require checking. The equivalence of the two constructions is well known to physicists \cite{Yellow_book, mussardo2010statistical} and parts of the material in this appendix also appeared in \cite{ChelkakGlazmanSmirnov, fonseca2003ward}.

The most important consequence of this is that by a celebrated argument of Belavin, Polyakov, and Zamolodchikov \cite{BPZ}, the computation of general correlations $\ccor{\Op^{(1)}_{z_1}\dots,\Op^{(n)}_{z_n}},$ where each $\Op^{(i)}$ is in one of the Virasoro submodules $\HVM{\ind}$, $\HVM{\en}$, $\HVM{\psi}$, $\HVM{\psi^\star}$, $\HVM{\sigma}$, $\HVM{\mu}$ can be reduced to an application of a differential operator on the correlation function of the corresponding primaries. We stress that this does not involve any fermionic insertions, as would the definitions of Section \ref{subsec: intro:cont_cor}. 
%For simplicity, we showcase this expression in the case of full plane. 

We focus on the holomorphic case, the anti-holomorphic one being completely parallel. As noted in the previous section, the Sugawara expression reads 
\[
\SET_w:=(\Lgen{-2}\ind)_w=\frac14\left(\Curr{-\frac32}\Curr{-\frac12}\ind\right)_w=\frac{1}{8\pi \i}\oint (z-w)^{-2}\psi_z\psi_w\,dz,
\]
we integral being over a small counterclockwise contour surrounding $w$. When we write identites or asymptotic expansions between \emph{fields}, such as the above one or the OPE in the style of \cite[Section 6]{chelkak2021correlations}, those should be understood as valid inside correlations with any other (distant) fields. In particular, $\partial_z \Op_z$ is understood to be a ``field" whose correlation with any other fields $\Op^{(1)},
\dots,\Op^{(n)}$ is given by 
\[
\ccor{\pa_z\Op_z\Op^{(1)}_{z_1}\dots\Op^{(n)}_{z_n}}_\Omega=\pa_z\ccor{\Op_z\Op^{(1)}_{z_1}\dots\Op^{(n)}_{z_n}}_\Omega.
\]

\begin{lem}
The field $\SET_z$ thus defined is holomorphic in $z$ inside correlations, more precisely,   \[\ccor{\SET_z\Op^{(1)}_{z_1}\dots\Op^{(n)}_{z_n}}_\Omega.\] is a holomorphic function of $z$ in $\Omega\setminus \{z_1,\dots,z_n\}$ for all $ \Op^{(1)},\dots,\Op^{(n)}\in\FullFockNS\oplus\FullFockR$ Moreover, when $\Op\in\FullFockNS$ or $\Op\in\FullFockR$, the Virasoro modes defined by 
\[
(\Lgen{n}\Op)_\zeta := \frac{1}{2\pi \i}\oint (z-\zeta)^{n+1}\SET_z\Op_\zeta\,dz 
\]
conicide with the ones given by Sugawara construction \eqref{eq:fermionic_sugawara_ns}, \eqref{eq:fermionic_sugawara_r} applied to $\Curr{n}$ as in Definition \ref{def: correlations-functions}. In particular, $\Lgen{n}$ indeed obey the Virasoro commutation relations.
\end{lem}
\begin{proof}
	The holomorphicity is obvious. For the second claim, to calculate $(\Lgen{n}\Op)_\zeta$ for $\Op\in\FullFockNS$, we assume for notational simplicity that $\zeta=0$ and write, for small counterclockeise contours $\gamma_3$ enclosing $\gamma_2$ enclosing $\gamma_1$ enclosing $0$, 
	\begin{multline*}
	(\Lgen{n}\Op)_0=\frac{1}{2\pi\i}\oint_{w \in \gamma_2}w^{n+1}\frac{1}{8\pi \i}\oint_{|z-
		w|=
		\frac12 |w|} (z-w)^{-2}\psi_z\psi_w\Op_0\,dzdw\\
		=-\frac{1}{16\pi^2}
		\left(\oint_{w \in \gamma_2}\oint_{z \in \gamma_3}w^{n+1}(z-w)^{-2}\psi_z\psi_w \Op_0\,dz dw\right.\\
		\left.+\oint_{w \in \gamma_2}\oint_{z \in \gamma_1}w^{n+1}(z-w)^{-2}\psi_w\psi_z \Op_0\,dz dw\right).
	\end{multline*}
Substituting into the two intergrals the expansions  
\begin{equation}
	\label{eq: z-w minus square expansions}
(z-w)^{-2}=z^{-2}
\sum_{m=1}^{
	\infty}m\left(\frac{w}{z}\right)^{m-1}\quad \text{and}\quad (z-w)^{-2}=w^{-2}
	\sum_{m=1}^{
	\infty}m\left(\frac{z}{w}\right)^{m-1}
\end{equation}
respectively, we arrive at the desired expression
\begin{multline*}
(\Lgen{n}\Op)_0=\frac14\left(\sum_{m=0}^\infty m \Curr{-m-\frac12}\Curr{n+\frac12+m}+\sum_{m=1}^\infty m \Curr{n-m+\frac12}\Curr{m-\frac12}\right)\Op_0\\=\begin{cases}\frac14\sum_{k\in \Z+\frac12} k\Curr{n-k}\Curr{k}\Op_0,& n\neq 0,\\ 
\frac12 \sum_{k
\in \Z_{\geq 0}+
\frac12}k\Curr{-k}\Curr{k}\Op_0,&n=0.\end{cases}
\end{multline*}
In the last identity, we changed the summation index to $k=-m-\frac12$ in the first sum and $k=m-\frac12$ in the second one,  and used the anti-commutation $\Curr{n-k}\Curr{k}=-
\Curr{k}\Curr{n-k}$ for $n
\neq 0$, as well as \eqref{eq: Sugawara_shift}. 

When $\Op\in \FullFockR$, we need to take into account that the fermionic correlations are ramified at the location of $\Op$. We observe that due to the expansion 
\[
z^{\frac12}w^{-\frac12}(z-w)^{-2}=(z-w)^{-2}+\frac12w^{-1}(z-w)^{-1}-\frac18w^{-2}+O(z-w)
\] and the fermionic OPE $\psi_z\psi_w=\frac{2}{z-w}+O(z-w),$ we have an alternative expression 
\[
\SET_w=\frac{1}{8\pi \i} \oint z^{\frac12}w^{-\frac12}(z-w)^{-2}\psi_z\psi_w\,dz+\frac{1}{16}w^{-2}.
\]

Hence, for $\Op\in\FullFockR$, we can write
\begin{multline*}
	(\Lgen{n}\Op)_0=\frac{1}{2\pi\i}\oint_{w \in \gamma_2}w^{n+1}\frac{1}{8\pi \i}\oint_{|z-
		w|=
		\frac12 |w|} z^{\frac12}w^{-\frac12}(z-w)^{-2}\psi_z\psi_w\Op_0\,dzdw\\+\frac{1}{32\pi\i}\oint_{w\in\gamma_2} w^{n-1}\Op_0\,dz\\
	=-\frac{1}{16\pi}
	\left(\oint_{w \in \gamma_2}\oint_{z \in \gamma_3}w^{n+1}z^{\frac12}w^{-\frac12}(z-w)^{-2}\psi_z\psi_w \Op_0\,dz dw\right.\\
	\left.+\oint_{w \in \gamma_2}\oint_{z \in \gamma_1}w^{n+1}z^{\frac12}w^{-\frac12}(z-w)^{-2}\psi_w\psi_z \Op_0\,dz dw\right)+\frac{\delta_{n,0}}{16}\Op_0.
\end{multline*}
Plugging in the expansions \eqref{eq: z-w minus square expansions} again, we arrive at 
\begin{multline*}
	(\Lgen{n}\Op)_0=\frac14\left(\sum_{m=0}^\infty m \Curr{-m}\Curr{n+m}+\sum_{m=1}^\infty m \Curr{n-m}\Curr{m}\right)\Op_0\\=\begin{cases}\frac14\sum_{k\in \Z} k\Curr{n-k}\Curr{k}\Op_0,& n\neq 0,\\ 
		\frac{1}{16}\Op_0+\frac12 \sum_{k
			\in \Z_{\geq 0}}k\Curr{-k}\Curr{k}\Op_0,&n=0.\end{cases}
\end{multline*}
\end{proof}

\begin{lem}
	Inside the correlations, we have the Ward identity \begin{align*}
		\SET_z \Op_w=\frac{\SDhol{\Op}\Op_w}{(z-w)^2}+\frac{\pa_w\Op_w}{z-w}+O(1),
	\end{align*}
	for $\Op\in\{\ind,\sigma,\mu,\psi,\psi^\star,\en\},$ with $\SDhol{\Op}$ being the corresponding $\Lgen{0}$ eigenvalue, i.e., $0,\frac{1}{16},\frac{1}{16},\frac12,0,\frac12$ respectively.
\end{lem}

\begin{proof}
	The statement of the lemma can be reformulated as \[(\Lgen{n}\Op)_w=\begin{cases}
		0&n>0;\\
		\SDhol{\Op}\Op_w& n=0\\
		\pa_w\Op_w.& n=-1;\\
	\end{cases}\] The $n\geq 0$ assertions follow algebraically from $\Psi_n\ind=0$ and $\Psi_n\sigma=0,$ $n>0$, and the Sugawara formulae. The $n=-1$ case is obvious for $\Op=\ind$. For $\Op=\sigma,$ we write by \eqref{eq:fermionic_sugawara_r}
	\[
	(\Lgen{-1}\sigma)_w=\frac14(\Curr{-1}\Curr{0}\sigma)_w=\frac{\lambb}{4}(\Curr{-1}\mu)_w=\frac{\lambb}{8\pi \i}\oint(z-w)^{-\frac32}\psi_z\mu_w\,dz,
	\]
	and conclude using the OPE \cite[Section 6]{chelkak2021correlations} \[\psi_z\mu_w=\lamb(z-w)^{-\frac12}(\sigma_w+4(z-w)\pa_w\sigma_w+o(z-w)).\] 
	
	For $\Op\in \{\mu,\psi,\psi^\star,\en\},$ the identity $(L_{-1}\Op)_w=\pa_w\Op_w$ follows from the next lemma; there is no circular reasoning as we do not need these cases for the base of induction there.   
\end{proof}
\begin{lem}
	\label{lem: Lmone=pa}
	We have, for every $\Op\in\FullFockNS\oplus\FullFockR,$ $(\Lgen{-1}\Op)_z=(\partial_z \Op_z)$ and $(\Lbar{-1}\Op)_z=(\bar\partial_z \Op_z)$, in the sense that the identities
	\[
	\ccor{(\Lgen{-1}\Op)_z\Op^{(1)}_{z_1}\dots,\Op^{(n)}_{z_n}}=\pa_z\ccor{\Op_z\Op^{(1)}_{z_1}\dots,\Op^{(n)}_{z_n}}
	\]
	\[
	\ccor{(\Lbar{-1}\Op)_z\Op^{(1)}_{z_1}\dots,\Op^{(n)}_{z_n}}=\bar{\pa}_z\ccor{\Op_z\Op^{(1)}_{z_1}\dots,\Op^{(n)}_{z_n}}
	\]  
hold true for every $\Op^{(1)},\dots,\Op^{(n)}$ and any $z,z_1,\dots,z_n$.
\end{lem}
\begin{proof}
	For $\Op\in\{\ind,\sigma\},$ this is the previous lemma. In general, we check that for any $n$, we have inside the correlations \[\pa_w \Curr{n}-\Curr{n}\pa_w =\left(-n+\frac12\right)\Curr{n-1}=\Lgen{-1}\Curr{n}-\Curr{n}\Lgen{-1},\] which allows us to prove the results for $\Op=\Curr{n_1}\dots\Curr{n_k}\hat\Op,$ where $\hat\Op\in \{\ind,\sigma\}$, by induction in $k$. We then extend it to any $\Op$ by linearity.
	
	Of the above identities, the first one is a simple computation
	\begin{multline*}
	\pa_w\ccor{(\Curr{n}\Op)_w\Op^{(1)}_{z_1}\dots,\Op^{(n)}_{z_n}}=\partial_w\left(\frac{1}{2\pi\i}\oint_w(z-w)^{n-\frac12}\ccor{\Op_w\Op^{(1)}_{z_1}\dots,\Op^{(n)}_{z_n}}\,dz\right)\\
	-\frac{n-\frac12}{2\pi \i}\oint_w(z-w)^{n-\frac32}\ccor{\Op_w\Op^{(1)}_{z_1}\dots,\Op^{(n)}_{z_n}}\,dz+\frac{1}{2\pi\i}\oint_w(z-w)^{n-\frac12}\pa_w\ccor{\Op_w\Op^{(1)}_{z_1}\dots,\Op^{(n)}_{z_n}}\,dz\\
	=\left(-n+\frac12\right)\ccor{(\Curr{n-1}\Op)_w\Op^{(1)}_{z_1}\dots,\Op^{(n)}_{z_n}}+\ccor{(\Curr{n}\pa_w\Op)_w\Op^{(1)}_{z_1}\dots,\Op^{(n)}_{z_n}},
	\end{multline*}
	and the second one follows from the Sugawara formulae \eqref{eq:fermionic_sugawara_ns}, \eqref{eq:fermionic_sugawara_r}.
\end{proof}
\begin{rem}
	The induction step in the above proof involves Leibnitz formula, which fails to hold exactly in the discrete. In particular, in the discrete, it is not clear to how to realize $\Loper{-1},\Loperbar{-1}$ as simple finite difference operators (independent on what they are acting on), even in the case of highly symmetric lattices.  
\end{rem}

\begin{cor}
	We have the OPE
	\[
	\SET_z\SET_w = \frac{1/4}{(z-w)^4}+\frac{2\SET_w}{(z-w)^2}+\frac{\pa_w \SET_w}{(z-w)}+O(1).
	\]
\end{cor}
\begin{proof}
	By definition and Lemma \ref{lem: Lmone=pa} applied to $\Op=T$, the assertion is equivalent to 
 \[
(\Lgen{n}\SET)_w = (\Lgen{n}\Lgen{-2}\ind)_w=
 \begin{cases}
 0&n>2\text{ or } n=1;\\
 1/4\ind_w& n=2\\
 (\Lgen{-2}\ind)_w& n=0;\\
 (\Lgen{-1}\Lgen{-2}\ind)_w& n=-1,
 \end{cases}
 \]
 which follow readily from the Virasoro commutation relations and the fact that $\Lgen{n}\ind=0,$ $n\geq -1$.
\end{proof}

\iffalse

\begin{lem}
	In the full plane we have, for any fields $\Op^{(1)},\dots, \Op^{(n)}$ and $\Op$, and any $n$, 
	\[
	\ccor{(\Lgen{n}\Op)_z \Op^{(1)}_{z_1} \dots \Op^{(n)}_{z_n}}=\sum_{i=1}^n	(z_i-z)^{n+1}\pa_{z_i}\ccor{\Op_z \Op^{(1)}_{z_1} \dots \Op^{(n)}_{z_n}}
	\] 
\end{lem}

\begin{proof}
	We have \begin{multline*}
		\ccor{(\Lgen{n}\Op)_z \Op^{(1)}_{z_1} \dots \Op^{(n)}_{z_n}}=\frac{1}{2\pi \i}\oint (w-z)^{n+1}\ccor{T_w\Op_z \Op^{(1)}_{z_1} \dots \Op^{(n)}_{z_n}}\,dw\\
		=-\sum_{i=1}^n\res_{z_i}(w-z)^{n+1}\ccor{T_w\Op_z \Op^{(1)}_{z_1} \dots \Op^{(n)}_{z_n}}\,dw
	\end{multline*}
\end{proof}

\fi
\subsection*{Acknowledgements} The work was supported by Research Counsil of Finland project 363549 ``Lattice models and conformal field theory" and project 346309 ``CoE FIRST". We thank Mikhail Basok, Cl\'ement Hongler, and Christian Webb for useful discussions. Generative AI was not used in either the research or the writing of this paper.

\bibliographystyle{plain}
\bibliography{Mixedcorr}

\end{document}